\documentclass[a4paper,onecolumn,11pt,accepted=2026-05-25]{quantumarticle}
\pdfoutput=1
\usepackage[utf8]{inputenc}
\usepackage[english]{babel}
\usepackage[T1]{fontenc}
\usepackage{amsmath}
\usepackage{hyperref}

\usepackage{tikz}
\usepackage{lipsum}

\usepackage{preamble}
\usepackage{xcolor}
\usepackage{amsthm,xspace,stmaryrd,bbm, amsfonts}
\usepackage[ruled,vlined,linesnumbered]{algorithm2e}
\usepackage{cancel}
\usepackage{algorithmic}
\usepackage{lineno}
\newcommand\mcP{\mathcal{P}}

\newcommand\mbZ{\mathbb{Z}}

\begin{document}

\title{On the dynamical Lie algebras of quantum approximate optimization algorithms}

\author{Jonathan Allcock}
\affiliation{Tencent Quantum Laboratory, Shenzhen, China}
\email{jonallcock@tencent.com}
\author{Miklos Santha}
\email{miklos.santha@gmail.com}
\affiliation{CQT, National University of Singapore, Singapore}
\affiliation{CNRS, IRIF, Université de Paris, Paris, France}
\author{Pei Yuan}
\email{peiyuan@tencent.com}
\affiliation{Tencent Quantum Laboratory, Shenzhen, China}
\author{Shengyu Zhang}
\email{shengyzhang@tencent.com}
\affiliation{Tencent Quantum Laboratory, Shenzhen, China}
\maketitle

\begin{abstract}
Dynamical Lie algebras (DLAs) have emerged as a valuable tool in the study of parameterized quantum circuits, helping to characterize both their expressiveness and trainability. In particular, the absence or presence of barren plateaus (BPs)---flat regions in parameter space that prevent the efficient training of variational quantum algorithms---has recently been shown to be intimately related to quantities derived from the associated DLA.

In this work, we investigate DLAs for the quantum approximate optimization algorithm (QAOA), one of the most studied variational quantum algorithms for solving graph MaxCut and other combinatorial optimization problems.  While DLAs for QAOA circuits have been studied before, existing results have either been based on numerical evidence, or else correspond to DLA generators specifically chosen to be universal for quantum computation on a subspace of states.  We initiate an analytical study of barren plateaus and other statistics of QAOA algorithms, and give bounds on the dimensions of the corresponding DLAs and their centers for general graphs.  We then focus on the $n$-vertex cycle and complete graphs. For the cycle graph we give an explicit basis, identify its decomposition into the direct sum of a $2$-dimensional center and a semisimple component isomorphic to $n-1$ copies of $\.{su}(2)$. We give an explicit basis for this isomorphism, and a closed-form expression for the variance of the cost function, proving the absence of BPs. For the complete graph we prove that the dimension of the DLA is $\+O(n^3)$ and give an explicit basis for the DLA.
\end{abstract}

\section{Introduction}

Variational quantum algorithms (VQA)~\cite{cerezo2021variational} are one of the two main paradigms in quantum algorithm design.  While structured algorithms---the other main pillar, which includes Grover search~\cite{grover1996fast}, Shor's factorization algorithm~\cite{shor1994algorithms}, the HHL linear systems algorithm~\cite{harrow2009quantum} and many others~\cite{quantumalgorithm}---are amenable to mathematical analysis and often have provable performance guarantees, the need for fault-tolerant quantum computers to run them comes with large resource overheads~\cite{babbush2021focus, hoefler2023disentangling}.  Variational algorithms, on the other hand, were originally proposed~\cite{peruzzo2014variational} to be run on current and near-future noisy, intermediate-scale quantum (NISQ) computers.  Comprising parameterized quantum circuits (PQC) with tunable parameters, the structure and size of the circuits can be tailored to the hardware available, and the parameters tuned to optimize the value of a given loss function in a manner similar to artificial neural network training.  

The broad and flexible VQA framework comes, however, with a number of challenges, including how to (i) design the circuit to balance expressiveness with trainability (i.e., the ability to produce good solutions without requiring a surfeit of parameters); (ii) effectively encode input data into the circuit; (iii) maintain performance in the presence of noise; and (iv) efficiently compute and optimize the loss function. 

While some challenges may be coped with by a trial-and-error approach and many experiments, training efficiency is a prominent difficulty encountered in almost all VQAs. The optimization of the loss function can be problematic for several reasons. Firstly, evaluating gradients on a quantum computer is typically a time-intensive procedure, requiring many measurements. Secondly, even if the gradients can be computed, or gradient-free optimization methods used, VQAs can be susceptible to the `Barren Plateaus' (BP) phenomenon~\cite{mcclean2018barren}, where the loss function becomes exponentially flat in large regions of parameter space as the problem size increases.  Various causes of this phenomenon have been studied, including the expressiveness of the PQC~\cite{mcclean2018barren, ortiz2021entanglement, 
patti2021entanglement, pesah2021absence, holmes2022connecting, larocca2022diagnosing,sharma2022trainability,friedrich2023quantum,martin2023barren}, the presence of entanglement or randomness in the initial state \cite{mcclean2018barren, cerezo2021cost,abbas2021power,holmes2021barren,shaydulin2022importance,thanasilp2023subtleties}, the locality of the measurement operator that appears in the loss function~\cite{cerezo2021cost,uvarov2021barren,kashif2023impact,khatri2019quantum}, and noise~\cite{wang2021noise,stilck2021limitations,garcia2024effects}. In two recent works, these seemingly disparate factors were identified as all being mathematically connected to the dynamical Lie algebra (DLA) of the VQA~\cite{ragone2023unified,fontana2024characterizing}, at least when the initial state or measurement operator lies in the DLA (a requirement relaxed in~\cite{diaz2023showcasing}). For an $n$-qubit VQA, the DLA $\.g$ is a subalgebra of the Lie algebra $\.{su}(2^n)$, and the variance of the loss function decomposes into a sum of contributions, each corresponding to a simple Lie subalgebra $\.g_j\subseteq \.g$  of the DLA.  For each subspace, the contribution to the variance scales proportionally with the squared norm of the initial state and measurement operator, and inversely proportional to the dimension of $\.g_j$~\cite{ragone2023unified}. Viewed through this lens, the tradeoff between expressiveness and trainability becomes more clear: fully expressive circuits -- those that are universal for quantum computation -- are those with DLAs that are equal to $\.{su}(2^n)$, but in this case the loss function variance is exponentially suppressed in $n$. On the other hand, smaller DLA dimensions are compatible with larger loss function variances, and hence favor trainability. In addition, the dimension of the DLA has been linked to the classical simulatability of VQAs~\cite{goh2023lie}, phase-transitions in solving satisfiability problems~\cite{zhang2022quantum}, and has been shown to upper-bound the number of parameters required for overparameterization~\cite{larocca2023theory} -- where there are enough degrees of freedom to explore all relevant directions in state space -- and facilitate convergence to global optima~\cite{kiani2020learning,wiersema2020exploring,monbroussou2023trainability}.

While the connection between DLAs and BPs is of fundamental conceptual significance, making use of this as a practical tool for understanding specific VQAs relies on our ability to compute various properties of interest: (i) the dimension and set of basis vectors for a particular DLA; (ii) its decomposition into abelian and simple components; and (iii) basis vectors for these components (which may not be easily deduced from a basis for the whole DLA).  For a specific VQA with a fixed number of qubits, it is straightforward to evaluate the dimension and basis numerically although, in the worst case, the time complexity can scale exponentially in the number of qubits.  However, the evaluation of (ii) and (iii) is more challenging, yet necessary for the loss function variance to be lower-bounded.

In this work, we investigate the DLAs for one of the most well-known variational algorithms: the quantum approximate optimization (QAOA) algorithm~\cite{farhi2014quantum, blekos2024review} for solving the MaxCut problem on graphs (QAOA-MaxCut). In contrast to the DLAs studied in~\cite{wiersema2024classification} and~\cite{aguilar2024full}, whose generating elements are each a single string of Pauli operators, QAOA-MaxCut DLAs have generators which are each sums of Pauli strings. This corresponds to a single tunable parameter being shared by all Pauli strings in a generator, similar to parameter sharing that occurs in classical neural networks such as a convolutional neural network.  
The fact that the generators are sums of Pauli strings may benefit the training, but it makes the DLA analysis considerably more challenging, and 
previous studies where this has been the case have either relied on numerical studies~\cite{ choquette2021quantum,larocca2022diagnosing, monbroussou2023trainability, meyer2023exploiting}, generators chosen specifically so that the resulting DLAs were universal~\cite{lloyd2018quantum, morales2020universality}, or focused on \emph{subspace-controllable} systems~\cite{albertini2018controllability, kazi2024universality, schatzki2024theoretical} where the DLA allows for arbitrary unitary operators to be applied to each symmetric-group invariant subspace, and where tools from representation theory can be applied.

Here we initiate an analytical study of QAOA-MaxCut DLAs.  We make use of symmetry properties to give upper-bounds on the dimension of QAOA-MaxCut DLAs for arbitrary graphs, and upper-bound the dimension of their centers (a subalgebra that commutes with the DLA). We then focus on two particular families of graphs: the $n$-vertex cycle graphs $C_n$ and complete graphs $K_n$.  These two graph families have the dihedral group and symmetric group as automorphism groups \cite{harary1969graph}, respectively, two important nonabelian groups with longstanding significance in quantum computing through their connections to quantum hidden subgroup algorithms for certain hard lattice problems~\cite{regev2004quantum} and the graph isomorphism problem~\cite{boneh1995quantum, beals1997quantum, ettinger1999quantum}.  For $C_n$ we show that the corresponding DLA has a $2$-dimensional center, and a $3(n-1)$ dimensional semisimple component isomorphic to $n-1$ copies of $\.{su}(2)$. In addition, we give an explicit basis for this isomorphism which enables, for the first time, an analytical expression for the expectation and variance of a QAOA-MaxCut loss function to be given, and shows that the loss function variance {essentially does not decrease with}  $n$, and thus BPs are not present.  Despite the low dimension of the DLA, and hence low expressiveness, QAOA-MaxCut is guaranteed (in the limit of long circuit depths) to give an optimal solution to the target problem \cite{farhi2014quantum}.  For $K_n$ we pin down the exact dimension of the DLA, which turns out to be $\+O(n^3)$, and give an explicit basis for it.  The dimension we find is strictly smaller than the space of all permutation invariant matrices in $\.{su}(2^n)$~\cite{albertini2018controllability}, and the DLA thus does not correspond to a subspace-controllable system. 

Our results initiate a rigorous analysis of barren plateaus in specific variational quantum algorithms through the lens of DLAs, which helps to deepen our understanding of the algorithms and sheds light on the future design of them.

\subsection{Related work}

It is well-known that if a DLA $\.g$ has two randomly chosen generators, then, with probability one, $\.g= \.{su}(2^n)$~\cite{lloyd1995almost} or \cite{jurdjevic1997geometric} (Theorem 12, ch.6), in which case the system is said to be universal for computation, or \emph{controllable}. Controllability was also shown for the DLAs corresponding to a family of one-dimensional spin chains~\cite{lloyd2018quantum},  Hamiltonians involving $ZZ$ and $ZZZ$ terms arranged according to the adjacency matrix of certain graphs and hypergraphs~\cite{morales2020universality} and, in~\cite{larocca2022diagnosing}, those corresponding to the Hardware-Efficient Ansatz (HEA)~\cite{kandala2017hardware} and a particular spin-glass Hamiltonian. A full classification of the DLAs for 
two-local spin systems on linear, circular and fully connected topologies was analytically carried out~\cite{wiersema2024classification}, and it was found that the DLA dimension of any such system scales as either $\+O(n)$, $\+O(n^2)$ or $\+O(4^n)$. But note that there the generators are individual Pauli terms, whereas in QAOA  the generators are the \textit{summation} of Pauli terms ($\sum_j X_j$ and $\sum_{(j,k)\in E} Z_j Z_k$), and the two-local interactions are between edges of a general graph.

{A full classification of all DLAs generated by individual Pauli terms was recently given in ~\cite{aguilar2024full}, where it was shown that all such DLAs fall into one of four possible categories depending on the structure of an associated anticommutation graph. For the specific case of \emph{multi-angle QAOA}~\cite{herrman2022multi}---a variant of QAOA, where the sums of Pauli terms are split up into individual Pauli strings---applied to MaxCut, a complete classification of the DLA for all graphs was given in~\cite{kokcu2024classification,kazi2024analyzing} where it was shown that $P_n$ and $C_n$ have DLA dimensions $\+O(n^2)$, and all other graphs, including $K_n$, fall into one of four categories, all with DLA dimension exponential in $n$. In contrast, our results for $C_n$ and $K_n$ show that, in some cases, the standard QAOA ansatz, with summations of Pauli strings as generators (and being significantly more challenging to analyze), can lead to a significant reduction in DLA dimension.

QAOA-MaxCut for a graph $G$ is an example of an \emph{equivariant quantum neural network}~\cite{meyer2023exploiting,nguyen2024theory}, where the DLA generators commute with representations of a group, in this case the automorphism group of the corresponding graph. 
The DLA dimension for  $G=P_n$ (the $n$-vertex path graph) has been numerically
investigated and found to be $n^2$~\cite{larocca2022diagnosing,meyer2023exploiting}. In~\cite{larocca2022diagnosing}, the authors also performed numerical studies of the DLAs for random Erd\H{o}s-R\'enyi graphs, and the dimension was found to increase exponentially with system size. The DLA for $G=C_n$ has been analyzed in ~\cite{matos2023characterization} and stated to have dimension $3n-2$, although a proof was not given. In comparison, our result for $C_n$ proves the dimension of the DLA, gives an explicit basis, and also gives a decomposition into a center and simple subalgebras, each with an explicit basis. Our follow-up study \cite{mao2025qaoa} analyzes the DLAs of QAOA-MaxCut for both weighted and unweighted graphs. We prove that for almost all weighted graphs with continuously distributed weights, as well as for unweighted graphs, the DLA dimension scales as $\Theta(4^n)$. Combined with the results in this paper, we can see that QAOA for most graphs has exponentially large DLAs, but graphs with enough symmetry can have polynomially small DLAs. These analytical results complement and partially explain numerical studies of graphs~\cite{herrman2021impact, shaydulin2021classical} which found that the amount of symmetry in a graph is a strong predictor of QAOA success.

\section{Results}

We use the following notation and conventions. For any  positive integer $n$, we let $[n]\defeq\{1,2,\ldots,n\}$, and define the parity function 
$\parity(n)=1$ if $n$ is odd and $\parity(n)=0$ if $n$ is even. We denote $n$-qubit Pauli strings by $P=P_1 P_2\ldots P_n\defeq P_1\otimes P_2\otimes\ldots\otimes P_n$, where $P_j\in\{I,X,Y,Z\}$, $I$ is the $2\times 2$ identity matrix, and 
\eq{
X &=  \begin{pmatrix} 0 & 1 \\ 1 & 0\end{pmatrix}, 
\quad Y = \begin{pmatrix} 0 & -i \\ i & 0\end{pmatrix}, 
\quad Z = \begin{pmatrix} 1 & 0 \\ 0 & -1\end{pmatrix}
}
are the standard Pauli operators. The set of all $n$-qubit Pauli strings is denoted $\+P_n$. Subscripts indicate the qubit number that a Pauli operator (or identity) acts on, e.g., $X_1 Y_2 Z_3 I_4 I_5$ means to apply $X$, $Y$ and $Z$ on the first, the second and the third qubits, respectively. Identity operators may be dropped, e.g.,  we may simply write $X_1 Y_2 Z_3$ for $X_1 Y_2 Z_3 I_4 I_5$. $S_n$ denotes the symmetric group on $n$ elements.  For sets $T_1$ and $T_2$, we set $T_1\backslash T_2\defeq \{t: ~t\in T_1,~t\notin T_2\}$. {For any $A,B\in \mbC^{n\times n}$, their Hilbert–Schmidt inner product is defined as $\langle A,B\rangle \defeq \tr(A^\dagger B)$. Finally, $\|\cdot \|_F$ denotes the matrix Frobenius norm.

\subsection{Variational quantum algorithms}

We consider variational quantum algorithms with a layered structure. Such algorithms can be specified by (i) an initial $n$-qubit state $\rho$; (ii) a parameterized quantum circuit acting on $n$-qubits corresponding to  a unitary operator
\eql{
U(\theta)&=\prod_{l=1}^L\lp \prod_{k=1}^{K}e^{i\theta_{l,k}H_k}\rp, \label{eq:pqc}
}
where $\theta = (\theta_{1,1}, \ldots, \theta_{L,K})$ are real parameters, and $H_1, \ldots, H_{K}\in \.{su}(2^n)$ are Hermitian operators; (iii)  a measurement operator $O\in\.{u}(2^n)$, $\norm{O}_F^2 \le 2^n$. For a specific choice of parameters $\theta$, the output of the circuit is given by the loss function
\eq{
\ell(\rho, O; \theta)&= \tr(U(\theta)\rho U^\dag(\theta)O),
}
and training the quantum circuit involves tuning the parameters to minimize the loss.  The ease with which this can be done can be captured, in part, by the variance of the loss function, 
$
\var_\theta[\ell(\rho, O;\theta)]. 
$ 
If the variance decreases exponentially in $n$, the loss is said to exhibit a \textit{barren plateau,} and the VQA is generally extremely hard to train by gradient-based methods---the methods most useful in practice for both classical machine learning and variational quantum algorithms.

\subsection{Dynamical Lie algebras}
Our analysis employs tools from Lie algebra. A Lie algebra $\.g$ is a vector space over some field $\mb{F}$ equipped with a bilinear map $[\cdot,\cdot]$ called the \emph{Lie bracket} that satisfies, for all $A,B,C\in\.g$, (i) $[A,A]=0$ and (ii) $[A,[B,C]] + [B,[C,A]]+[C,[A,B]]=0$ (Jacobi identity). Two important families of Lie algebras for us will be 
\eq{
\.{sl}(n,\mb{C}) &\defeq \{H\in\mb{C}^{n\times n} : \tr(H) = 0\},\quad\quad\.{su}(n) \defeq \{ H\in \mb{C}^{n\times n} : H^\dag = -H, \tr(H)=0\}.
}

The \emph{dynamical Lie algebra} corresponding to the PQC (Eq. \eqref{eq:pqc}), denoted
\eq{
\.g &= \langle \+A \rangle_{\text{Lie}, \mbR},
}
is the Lie algebra over $\mb{R}$ spanned by all possible nested commutators (such as $[A,[B,C]]$) of elements in $\+A$. Here,  $\+A=\{iH_1, \ldots, i H_L\}$ is the set of generators of $\.g$, 
and the Lie bracket is taken to be $[A,B] \defeq AB-BA$. Due to the Jacobi identity, without loss of generality, all nested commutators can be assumed to take the form $[H_1, [H_2, [\ldots, H_s]]]$. As a subalgebra of $\.{su}(2^n)$, the DLA of a parameterized quantum circuit admits a decomposition \cite{knapp1996lie}
\eq{
\.g &= \.c \oplus \.g_1 \oplus \ldots \oplus \.g_k
} 
into simple Lie algebras $\.g_i$ and a \emph{center} $\.c$ of elements that commute with all of $\.g$. The \emph{semisimple} component of $\.g$ is $[\.g,\.g] \defeq \spn\{[A,B] : A,B\in\.g\}=\.g_1\oplus\ldots\oplus\.g_k$. If $O\in i\.g$ or $\rho\in i\.g$, and the circuit is deep enough to form a unitary 2-design, then it was shown that the variance of the loss function $\ell$ is related to the DLA in the following way~\cite{ragone2023unified}:
\begin{equation}\label{eq:VQA-exp-var}
    \av_\theta[\ell(\rho,O;\theta)] = \tr(\rho_{\.c}O_{\.c}), \quad 
    \var_\theta[\ell(\rho,O;\theta)] = \sum_{j=1}^k \frac{\mcP_{\.g_j}(\rho) \mcP_{\.g_j}(O)}{\dim(\.g_j)}.
\end{equation}
Here $\+P_{\.s}(H)$, the $\.s$-purity of a Hermitian operator $H\in i\.{u}(2^n)$, with respect to a subalgebra $\.s$ of $\.u(2^n)$, is defined as
    \[\+P_{\.s}(H)\defeq \tr(H_{\.s}^2) = \|H_{\.s}\|_F^2 =  \sum_{j=1}^{\dim(\.s)}|\tr(E_j^\dag H)|^2,\]
    where $H_\.s$ is the orthogonal projection of operator $H$ onto $\.s_{\mbC}$ (the complexification of $\.s$), and $\{E_j: j=1,2,\ldots,\dim(\.s)\}$ is an orthonormal basis for $\.s$. 

Note that the dimension of $\.g$ alone is not sufficient for making use of Eq.~\eqref{eq:VQA-exp-var}, nor is a basis for the whole of $\.g$, as the purities of $\rho$ and $O$ may distribute on different subspaces $\.g_j$, in which case the variance can be 0. Therefore, lower bounding the variance requires understanding the decomposition of the Lie algebra into simple subalgebras, and finding a basis for these subalgebras, so that one can determine how the purities distribute amongst the $\.g_j$.

\subsection{QAOA}

The quantum approximate optimization algorithm (QAOA) ~\cite{farhi2014quantum} was originally proposed to approximately solve combinatorial optimization tasks specified by $n$ bits and $m$ clauses $C_1, \ldots, C_m:\{0,1\}^n\ra\{0,1\}$, where the aim is to maximize the objective function
\begin{equation*}
C(z)=\sum_{j=1}^{m} C_j(z),
\end{equation*}
over the set of bit strings $z\defeq z_1\ldots z_n\in\{0,1\}^n$. QAOA begins by initializing an $n$-qubit system in the state $\ket{+}^{\otimes n} = \frac{1}{\sqrt{2^n}}\sum_{z\in\{0,1\}^n}\ket{z}$, and then executes the layered variational circuit 
\eq{
U(\gamma, \beta) &= U_B(\beta_p) U_C(\gamma_p)\ldots U_B(\beta_1) U_C(\gamma_1),
}
where $p$ is a positive integer, $\gamma = (\gamma_1, \ldots, \gamma_p)$ and $\beta = (\beta_1, \ldots, \beta_p)$ with $\beta_j\in [0,\pi)$ and $\gamma_j\in [0,2\pi)$}, and 
\eq{
U_B(\beta_j) =\prod_{j=1}^n e^{-i\beta_j X_j}, \quad
U_C(\gamma_j)= \prod_{j=1}^m e^{-i\gamma_j C_j}. 
}
At the end of the circuit one measures in the computational basis and evaluates $C(z)$. When applied to the MaxCut problem on an $n$-vertex graph $G=(V,E)$, the $U_C$ and $U_B$ operators take the form
\eq{
U_B(\beta)= e^{-i\beta \sum_{j\in V}X_j}, \quad U_C(\gamma) = e^{-i\gamma \sum_{(j,k)\in E} Z_j Z_k}.
}
By comparison with Eq.~\eqref{eq:pqc}, for an undirected graph $G = (V,E)$, the DLA of the corresponding QAOA MaxCut algorithm is then 
\eq{
\.g_G &= \langle \{i\sum_{j\in V} X_j, i\sum_{(j,k)\in E} Z_j Z_k \}\rangle_{\text{Lie},\mbR}.
}
As in \cite{ragone2023unified}, we normalize the measurement $O$ so that $\|O\|_F^2 = 2^n$. This can be motivated by noticing that $\|P\|_F^2 = 2^n$ for any Pauli string measurement such as $X_1$ or $Z_1Z_2\ldots Z_n$. In the case of QAOA, we can thus take the measurement to be $O=\frac{1}{\sqrt{|E|}}\sum_{(j,k)\in E} Z_j Z_k$. Note that $O\in i\.g_G$ and $\|O\|_F^2 = 2^n$.



\subsection{Results for general DLAs}

We first prove two general results that apply to any DLA $\.g$. The first is a simple upper bound on the dimension of the center. 

\begin{restatable}{thm}{dimcenter}\label{lem:center}Let $\.g = \langle\+
A\rangle_{{\rm Lie},\mb{F}}$, and $\.c$ its center. If $\+A$ is a set of $s$ linearly independent elements then $\dim(\.c) \le s$.
\end{restatable}

In the case of QAOA, which has two linearly independent DLA generators, this immediately gives $\dim(\.c) \le 2$. The second result is an upper bound on the dimension of $\.g$, when the elements of $\+A$ (which are sums of Pauli strings) are invariant under certain symmetries. More precisely, for any Pauli string $P = P_1 P_2 \ldots P_n\in\+P_n$, 
and any permutation $\pi\in S_n$, define \[\pi(P) = P_{\pi(1)} P_{\pi(2)} \ldots P_{\pi(n)}.\] 
By linearity, the action of $\pi$ can be extended to the whole of $\.{su}(2^n)$.  We then have: 

\begin{restatable}{thm}{dimfromorbits}\label{thm:dim_dla_orbits}Let $\+A\subseteq \.{su}(2^n)$, $\.g=\langle \+A\rangle_{{\rm Lie},\mb{R}}$ and $S\le S_n$. If $\pi(A)=A$ for all $A\in\+A$ and all $\pi\in S$, then 
\eq{
\dim(\.g) \le \abs{\+P_n/S} - 1,
}
where $\abs{\+P_n/S}$ is the cardinality of the set $\+P_n/S$ of orbits of the $n$-qubit Pauli strings under the group action of $S$.
\end{restatable}


Applied to QAOA-MaxCut on an arbitrary graph $G$, the generators of $\.g_G$ are easily seen to be $Aut(G)$ invariant, where $Aut(G)$ is the automorphism group of $G$. Thus, $\dim(\.g_G)\le \abs{\+P_n/Aut(G)} - 1$. As an application of this result,  by additionally showing that $\.g_G$ is contained in a subalgebra spanned by Pauli strings where the total number of $Y$ and $Z$ operators is even and non-zero, we obtain a simple dimension upper-bound for the DLA of the complete graph of
\eql{
\dim(\.g_{K_n})&\le \begin{cases}
\frac{1}{12}\lp n^3+6n^2 + 14n - 12\rp &\quad \text{($n$ even)} \\
\frac{1}{12}\lp n^3 +6n^2 + 11n -18\rp &\quad \text{($n$ odd)}\end{cases}. \label{eq:kn-upper-bounds}
}

We then focus more deeply on the cases where $G=C_n$ and $G=K_n$, the cycle and complete graphs on $n$ vertices, respectively. Among other things, this allows us to exactly determine the respective DLA dimensions.} 
\subsection{Results for cycle graphs} 


To state our results for cycle and complete graphs let us first define some notation. Let $P = P_1P_2\ldots P_n$ be an $n$-qubit Pauli string, and $S$ a subgroup of the symmetric group $S_n$. Define the corresponding \emph{orbit-sum of $P$} to be 
\eql{
S(P)&\defeq\sum_{\pi\in S}\pi(P). \label{eq:orbit-sum}
} 
The automorphism group of $C_n$ is the dihedral group $D_n$ (the group of symmetries of a regular $n$-gon), and we consider the following Pauli orbit-sums, where the $n$ qubits are indexed by $\mbZ_n = \{0,1,\ldots,n-1\}$ and the summation of the subscript indices is over the cyclic group $\mbZ_n$.
\eql{
    X & \defeq \frac{i}{2} D_n(X_0) = i\sum_{j=0}^{n-1} X_j \label{eq:Cn-basis-X}\\
    {X^{n-1}} & \defeq \frac{i}{2} D_n(X_0\ldots X_{n-2}) = i\sum_{j=0}^{n-1} X_j X_{j+1}\cdots X_{j+n-2} \\
    {Z X^t Z} & \defeq \frac{i}{2} D_n(Z_0 X_1\cdots X_t Z_{t+1}) = i\sum_{j=0}^{n-1} Z_j X_{j+1}\cdots X_{j+t} Z_{j+t+1},  \label{eq:Cn-basis} \\
    {Y X^t Y} & \defeq \frac{i}{2} D_n(Y_0 X_1\cdots X_t Y_{t+1}) = i\sum_{j=0}^{n-1} Y_j X_{j+1}\cdots X_{j+t} Y_{j+t+1},   \\
    {Y X^t Z} & \defeq i D_n(Y_0 X_1\cdots X_t Z_{t+1}) 
     = i\sum_{j=0}^{n-1} (Y_j X_{j+1}\cdots X_{j+t} Z_{j+t+1} + Z_j X_{j+1}\cdots X_{j+t} Y_{j+t+1}), \label{eq:Cn-basis-YXtZ}
}
where $t=0,1,\ldots,n-2$.

\begin{restatable}{thm}{CnDecompCenter}\label{thm:Cn-decomp-center} 
$\.g_{C_n}$ has the following direct sum decomposition
\eq{
    \.g_{C_n} &= \.c \oplus \.g_1 \oplus \cdots \oplus\.g_{n-1},
}
where each $\.g_j \cong \.{su}(2)$ and the center $\.c = \spn_{\mbR} \{c_1, c_2\}$, where  
\begin{align*}
    c_1 = - X +\sum_{t=1}^{\frac{n-1}{2}} \big({ZX^{2t-1}Z} + {YX^{2t-1}Y}\big),\text{~~and~~}
     c_2 = {X^{n-1}} +\sum_{t=0}^{\frac{n-3}{2}}\big({ZX^{2t}Z} + {YX^{2t}Y}\big)
\end{align*}
if $n$ is odd, and 
\begin{align*}
    c_1  = {X^{n-1}} - X + \sum_{t=1}^{\frac{n-2}{2}}\big({ZX^{2t-1}Z} + {YX^{2t-1}Y}\big),\text{~~and~~} c_2  = \sum_{t=0}^{\frac{n-2}{2}} \big({ZX^{2t}Z} + {YX^{2t}Y}\big)
\end{align*}
if $n$ is even.
\end{restatable}
It follows that $\dim(\.g_{C_n}) = 3n-1$. In addition, we give an explicit construction for the above isomorphism.  More specifically, we say that a basis $\bigcup_{k=1}^{n-1}\{\tilde{X}_k, \tilde{Y}_k,\tilde{Z}_k\}$ for $\.g_1\oplus \ldots \oplus\.g_{n-1}$, where each $\.g_j\cong\.{su}(2)$, is \emph{canonical} if it respects the $\.{su}(2)$ commutation relations 
\eq{
[\tilde X_k, \tilde Y_k] = 2\tilde Z_k, \quad [\tilde Y_k, \tilde Z_k] = 2\tilde X_k, \quad [\tilde Z_k, \tilde X_k] = 2\tilde Y_k,
}
as well as the Lie algebra direct sum decomposition, i.e., $[x,y]=0$ for any $x\in\{\tilde{X}_j, \tilde{Y}_j,\tilde{Z}_j\}$ and any $y\in \{\tilde{X}_k, \tilde{Y}_k,\tilde{Z}_k\}$, where $j\neq k$.  
\begin{thm}\label{thm:cycle-iso-su2}
For each $k=1, 2, \ldots, n-1$, let 
\eq{
    \tilde Z_k \defeq &\frac{1}{2n} \sum_{j=1}^{n-1} \sin \frac{kj\pi}{n} YX^{j-1}Z, \\ 
    \tilde X_k 
    \defeq& \frac{1}{2n} \Big( \sin \frac{k\pi}{n} X - \sum_{j=1}^{n-2} \big( \sin \frac{k(j+1)\pi}{n} YX^{j-1} Y - \sin \frac{kj\pi}{n} ZX^j Z \big) + \sin \frac{k(n-1)\pi}{n} X^{n-1} \Big),\nonumber\\
    \tilde Y_k 
    \defeq&\frac{1}{2n}\Big( -\cos\frac{k\pi}{n}X + \cos k\pi YX^{n-2}Y + \cos\frac{k(n-1)\pi}{n}X^{n-1} + ZZ\Big) 
    \\& +\frac{1}{2n}\sum_{j=1}^{n-2} \Big( \cos \frac{k(j+1)\pi}{n} YX^{j-1} Y + \cos \frac{kj\pi}{n} ZX^j Z \Big).
}
Then, $\{\tilde{X}_k, \tilde{Y}_k, \tilde{Z}_k: k = 1, 2, \ldots, n-1\}$ is a canonical basis for $\.g_1\oplus\ldots \oplus\.g_{n-1}$ where each $\.g_j\cong\.{su}(2)$.
\end{thm}

\begin{rem*}
While this manuscript was being revised, we became aware of the recent work of D'Alessandro and Isik~\cite{d2024controllability} who also derived the above results for $\.g_{C_n}$, i.e., that it has a $2$-dimensional center and $n-1$ simple components each isomorphic to $\.{su}(2)$, and provided a basis for the isomorphism. However, their proof technique differs from ours, as does the explicit form of their basis (in particular, the coefficients of the Pauli orbit-sums in the basis vectors are expressed in terms of recursively defined polynomials, compared with the trigonometric coefficients that appear in our Theorem~\ref{thm:cycle-iso-su2}. Our form also makes it easy to compute the purities and the variance of the standard cost function.
\end{rem*}

Using this canonical basis, we can explicitly compute the associated measurement and initial state $\.g_j$-purities. 
\begin{restatable}{thm}{CnExpVar}\label{thm:Cn-exp-var}
    For QAOA-MaxCut on the $n$-node cycle graph, if the circuit is a unitary $2$-design, then  the expectation and variance of the loss function $\ell(\rho,O;\theta)$ are 
    \[\av_\theta[\ell(\rho,O;\theta)] = \frac{\parity(n)}{\sqrt{n}}, \quad \text{ and } \quad   \var_\theta[\ell(\rho,O;\theta)] = \frac{2(n-\parity(n))}{3n} \approx \frac{2}{3},\]
where $\parity(n)$ is equal to $1$ if $n$ is odd and $0$ if $n$ is even.
\end{restatable}
We numerically simulate QAOA circuits for solving MaxCut on cycle graphs $C_n$, and compare the expectation and variance of the loss function with the theoretical values from Theorem~\ref{thm:Cn-exp-var}.  Simulations were carried out using the TensorCircuit~\cite{zhang2023tensorcircuit} software package. Means and variances were computed based on 1,000 random initializations of the circuit parameters, each in the range $[0, \pi]$.  Cycle graphs $C_n$ from $n=3$ to $n=10$ were simulated and, for each value of $n$, QAOA circuits comprising $L$-layers for various values of $L$ were tested. Results are shown in Fig.~\ref{fig:numerical-sim} and show that, as $L$ increases, the simulated values converge towards the theoretical values, indicating the correctness of our analysis and the depth needed for the circuit to form an approximate $2$-design.

\begin{figure}
     \centering
     \begin{subfigure}[b]{0.44\textwidth}
         \centering
         \includegraphics[width=\textwidth]{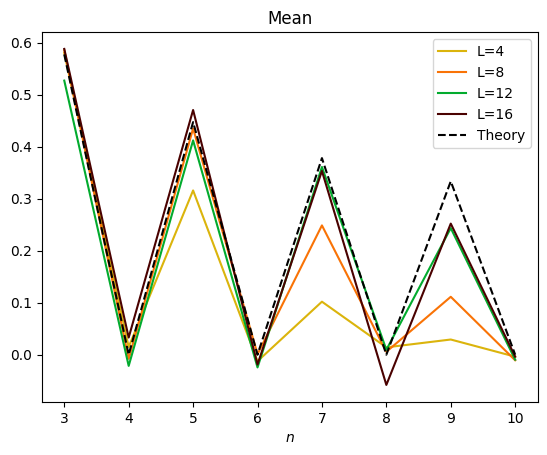}
     \end{subfigure}
     \begin{subfigure}[b]{0.45\textwidth}
         \centering
         \includegraphics[width=\textwidth]{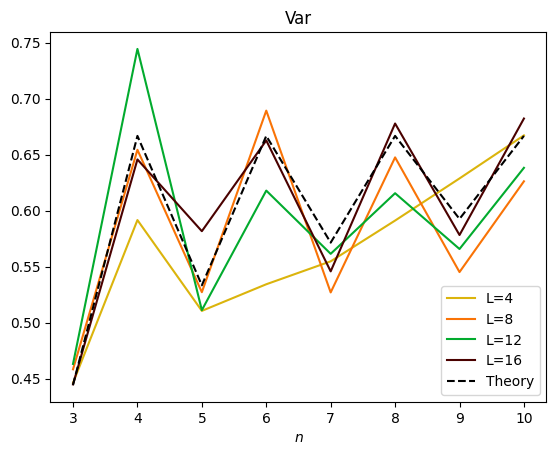}
     \end{subfigure}
        \caption{Numerical simulation of the cost function expectation value $\av_\theta[\ell(\rho,0;\theta)]$ (left) and variance $\var_\theta[\ell(\rho, O;\theta)]$ (right) for QAOA-MaxCut on cycle graphs $C_n$. $L$ denotes the number of layers of the QAOA ansatz. The dashed \emph{Theory} lines correspond to the values from Theorem~\ref{thm:Cn-exp-var} of $\av_\theta[\ell(\rho,0;\theta)] = \parity(n)/\sqrt{n}$ and $\var_\theta[\ell(\rho,O;\theta)]=2(n-\parity(n))/(3n)$. }
        \label{fig:numerical-sim}
\end{figure}

\subsection{Results for complete graphs}
For the complete graph $K_n$, the automorphism group is the full symmetric group $S_n$, and we define, for non-negative integers $p,q,r$ with $p+q+r \le n$,  the following Pauli orbit-sums:
\begin{align}
    X^p Y^q Z^r & \defeq iS_n (X_1\cdots X_p Y_{p+1}\cdots Y_{p+q} Z_{p+q+1}\cdots Z_{p+q+r}) \nonumber\\
    &= i\cdot \sum_{\text{ mutually disjoint} R, S, T\subseteq [n]  \atop \text{ with } |R|=p,\ |S|=q,\ |T|=r} X_R Y_S Z_T,  \label{eq:XpYqZr-basis}
\end{align}
where $X_R \defeq \otimes_{j\in R} X_j$, and similarly for $Y_S$ and $Z_T$.
That is, $X^pY^pZ^r$ is the orbit-sum containing all Pauli strings with the number of $X$, $Y$ and $Z$ operators being $p$, $q$ and $r$ respectively. If $p,q$ or $r$ is zero, we omit the corresponding Pauli operator, e.g. we write $Z^2$ instead of $X^0Y^0Z^2$.

\begin{restatable}{thm}{KnSemiSimpleCenter}\label{thm:kn-semisimple-center} $\.g_{K_n} = \.c \oplus [\.g_{K_n},\.g_{K_n}]$, where the semisimple component has dimension
\begin{equation*}
  \dim([\.g_{K_n},\.g_{K_n}])=\begin{cases}
        \frac{1}{12}\lp n^3+6n^2+2n \rp& \quad (\text{$n$ even})\\
        \frac{1}{12}\lp n^3+6n^2-n-6\rp&\quad (\text{$n$ odd})
    \end{cases},
\end{equation*}
and basis 
\eql{
\{X^pY^qZ^r,[X^1,X^pY^qZ^r], [Z^2,X^pY^1Z^r]:p\ge 0, ~\text{ $q$ and $r$ are odd}\}; \label{eq:kn-semisimple-basis}
}
and the center $\.c$ has dimension $\dim(\.c) = 1 + \parity(n)$. 
\end{restatable}

In the process of proving Theorem~\ref{thm:kn-semisimple-center}, we also identify a basis for $\.g_{K_n}$, and pin down the exact dimension of $\.g_{K_n}$ as:
\eql{
    \dim(\.g_{K_n})=\begin{cases}
        \frac{1}{12}\lp n^3+6n^2+2n+12\rp&\quad (\text{$n$  even})\\
        \frac{1}{12}\lp n^3+6n^2-n+18\rp &\quad (\text{$n$ odd})
    \end{cases}, \label{eq:dim-kn}
}
showing that the above simple upper bound in Eq.~\eqref{eq:kn-upper-bounds} is close to tight.

While we are able to show this exact polynomial expression for $\dim(\.g_{K_n})$, for complete graphs we are not able to give a closed form expression for expectation and variance of the cost function (as we were able to for the cycle graph), and cannot draw a definitive conclusion about the presence or absence of barren plateaus. 



\section{Discussion} In this work we initiated a study of the dynamical Lie algebras for QAOA, proved a number of dimension bounds that apply to arbitrary graphs and, in particular, gave a full characterization --- including an isomorphism to a direct sum of simple Lie algebras, and explicit basis vectors for each component --- of the DLA corresponding to the $n$-vertex cycle graph $C_n$, and a partial characterization of the DLA for the complete graph $K_n$. There are a number of natural open questions:
\begin{itemize}
    \item Comparison between Eqs.~\eqref{eq:kn-upper-bounds} and \eqref{eq:dim-kn} shows that the dimension bounds based on symmetry and $YZ$-parity are close to tight for $K_n$. This is not the case for graph families with smaller automorphism groups. Can tighter dimension bounds be proven for general graphs $G$? 
    \item What is the the decomposition of $[\.g_{K_n},\.g_{K_n}]$ into its constituent simple Lie algebra components? Since $K_3 = C_3$, our results for the cycle graph give $[\.g_{K_3},\.g_{K_3}]\cong \.{su}(2)\oplus \.{su}(2)$. In the Supplementary Information we show that $[\.g_{K_4},\.g_{K_4}]\cong \.{su}(2)\oplus \.{su}(2)\oplus\.{su}(3)$. What about in general? Based on the dimension of $\.g_{K_n}$ in Eq.~\eqref{eq:dim-kn}, the fact that $\dim(\.{su}(n)) = n^2-1$, we conjecture that, for $n\ge 2$,
\eq{
    [\.g_{K_n},\.g_{K_n}]&\cong \bigoplus_{j=1}^{n-1} \.{su}\lp \left\lfloor\frac{j+1}{2}\right\rfloor + 1\rp.
}
    \item For what families of graphs can a full characterization of the DLAs be given? For which families of graphs are barren plateaus not present?
    \item Beyond QAOA-MaxCut, which other variational quantum algorithm DLAs with generators that are sums of Pauli strings can be fully characterized? This scenario corresponds to parameter sharing, where a single tunable parameter is used to tune multiple Pauli strings at the same time.
\end{itemize}

 In the longer term, a better understanding of DLAs for variational circuits with parameter sharing can hopefully be useful not only for analyzing quantum circuits, but also as a tool in algorithm design. For example, to accelerate algorithm convergence, one can consider reducing the depth of a circuit and correlating parameters~\cite{holmes2022connecting}, or adding additional DLA generators that do not preserve the symmetries of solutions~\cite{choquette2021quantum}. Knowing how a set of $\+A$ of generators is related to the corresponding DLA, and how adapting $\+A$ quantitatively changes the expressiveness and trainability of the circuit, is key to tailoring variational algorithms for improved performance.


\section{Methods}

In this section we prove, or outline the proofs of, the theorems and other claims made in the Results section.  For brevity, we omit some of the proofs of supporting lemmas here. All missing proofs appear in the Supplementary Information.

Let us first review some concepts from the theory of Lie algebras. Given a Lie algebra $\.g$, if $[A,B]=0$ for all $A,B\in\.g$ then $\.g$ is said to be  \emph{commutative} or \emph{abelian}. A \emph{Lie subalgebra} $\.h\subseteq\.g$ of a Lie algebra $\.g$ is a linear subspace of $\.g$ which is closed under the Lie bracket, i.e. $[H_1,H_2]\in \.h$ for all $H_1, H_2\in \.h$.  A Lie algebra is said to be \emph{compact} if it is the Lie algebra of a compact Lie group (see \cite{hall2013lie} or \cite{knapp1996lie} for further details).  A Lie algebra $\.g$ is said to be the direct sum of subalgebras $\.g_1, \.g_2, \ldots, \.g_k$, denoted
\eq{\.g = \.g_1 \oplus \.g_2 \oplus \ldots \oplus \.g_k,}
    if it is the direct sum of $\.g_1, \ldots, \.g_k$ as vector spaces, and $[\.g_i, \.g_j] = 0$ for all $i\neq j$. An \textit{ideal} of a Lie algebra $\.g$ is a Lie subalgebra $\.i$ of $\.g$ satisfying the further requirement that $[A,B]\in \.i$ for all $A\in \.g$ and $B\in\.i$. An ideal $\.i$ of $\.g$ is \textit{trivial} if $\.i = \.g$ or $\.i = \{0\}$. For any $\.g$, one can define the following two ideals:
\eq{
Z(\.g) &\defeq \{ A\in \.g : [A,B]=0, \forall \,B\in \.g\}, &\quad (\text{center of }\.g) \\
[\.g, \.g]&\defeq \spn\{[A,B] : A, B\in\.g\},&\quad (\text{commutator ideal}) 
}
A \textit{simple Lie algebra} is a Lie algebra that is non-abelian and contains no nontrivial ideals. A \textit{semisimple Lie algebra} is one which is the direct sum of simple Lie algebras. Given a Lie algebra $\.g$ and any $A\in \.g$, the associated \textit{adjoint map} is the linear transformation $ad_A:\.g\ra \.g$ given by $ad_A(B) = [A,B]$.    A \emph{Cartan subalgebra} $\.h\subseteq\.g$ of a complex semisimple Lie algebra is a maximal commutative subalgebra of $\.g$, where $ad_H$ is diagonalizable for all $H\in\.h$.

\subsection{General DLAs}
Here we give proofs of Theorems~\ref{lem:center} and \ref{thm:dim_dla_orbits}, as well as the upper bound for the DLA of the complete graph given in Eq.~\eqref{eq:kn-upper-bounds}.

Computing a basis $\mathcal{B}$ for a DLA (and hence its dimension) from a set of generators $\+A$ is algorithmically straightforward, but has time polynomial in the DLA dimension, which can be exponential in the number of qubits $n$. The basis is constructed iteratively: first set $\+B=\+A$. Then compute the set of all possible commutators $\{[A,B] : A\in\+A, B\in\+B\}$ and, for each commutator, if it is linearly independent from the elements in $\+B$, add it to the basis.  Repeat the process until no new linearly independent elements are found. Algorithm~\ref{alg:DLA-generation} in the Supplementary Information makes this procedure precise.  


We now prove Theorem~\ref{lem:center}, which we recall below.

\dimcenter*
\begin{proof}
    Suppose that $\mathcal A = \{H_1, \ldots, H_s\}$, and  Algorithm \ref{alg:DLA-generation} outputs the basis  \[\{H_1, \ldots, H_s, \ldots, H_m\}.\] Then by the algorithm, any element $H_{t}$ with $t\ge s+1$ is generated by taking the commutator of two previously generated matrices $H_j$ and $H_{j'}$.  $H_j$ and $H_{j'}$ are both in $\.g$, and thus $H_t = [H_j, H_{j'}]$ is in $[\.g,\.g]$. Thus $H_{s+1}, \ldots, H_m$ are all in $[\.g,\.g]$. Since they are also linearly independent as required by the algorithm, we know that $[\.g,\.g]$ has dimension at least $m-s$. As $\.g = \.c \oplus [\.g,\.g]$, we obtain the claimed bound by
    \[\dim(\.c) = \dim(\.g) - \dim([\.g,\.g]) \le s.\] 
\end{proof}
We remark that if the elements of $\+A$ are not linearly independent, then $\dim(\.c)$ is upper-bounded by the maximum number of linearly independent elements of $\+A$.

To prove Theorem~\ref{thm:dim_dla_orbits}, first recall that a (left) \emph{group action} of a group $G$ on a set $X$ is a function $f:G\times X\ra X$, $(g,x)\mapsto g.x$, such that (i) $e.x=x$, (ii) $g_1.(g_2.x) = (g_1 g_2).x$, where $e$ is the identity element of $G$. The \emph{orbit} of $x\in X$ under the action of $G$ is
\eql{
G_X(x)&\defeq\{g.x : g\in G\}, \label{eq:orbit-under-group-action}
}
and the set of all orbits of $X$ under the action of $G$, denoted $X/G$, forms a partition of $X$. Note that the orbit-sum of a Pauli string $P$ (see Eq.~\eqref{eq:orbit-sum}) is just an equally weighted linear combination of all elements in the orbit of $P$ under the action of a subgroup of $S_n$. 


For any Pauli string $P = P_1 P_2 \ldots P_n\in\+P_n$, 
and any permutation $\pi\in S_n$, define \[\pi(P) = P_{\pi(1)} P_{\pi(2)} \ldots P_{\pi(n)}.\] 
This defines an action $S_n$ on the set $\mcP_n$ of all $n$-qubit Pauli strings. 
By linearity this can be extended to the whole of $\.s\.l(2^n,\mb{C})$, by expressing any element in $\.s\.l(2^n,\mb{C})$ as a linear combination of Pauli strings. For $\pi\in S_n$, we say that $A\in\.{su}(2^n)$ is $\pi$-invariant if $\pi(A)=A$.

\begin{restatable}{lemma}{LemPiInvar}\label{lem:pi-invariance} Let $\+A\subseteq \.{su}(2^n)$, $\.g=\langle \+A\rangle_{{\rm Lie},\mb{F}}$ and $\pi\in S_n$. If $\pi(A)=A$ for all $A\in\+A$, then $\pi(H) = H$ for all $H\in\.g$.
\end{restatable}

\begin{proof} It is straightforward to show (see Supplementary Information) that $\pi$-invariance is preserved by the commutator, i.e., if $\pi(A)=A$ and $\pi(B)=B$ then $\pi([A,B])= [A,B]$. The result then holds by linearity.
\end{proof}

Now, any $n$-qubit traceless skew-Hermitian can be written as $H = i\sum_{P\in \+P_n} \alpha_P P$, where $\alpha_P\in\mb{R}$. If the generators of $\.g$ are $\pi$-invariant for all $\pi\in S$,  the coefficients $\alpha_P$ are the same for all Pauli strings $P$ in the same orbit-sum, and we obtain:

\begin{cor}\label{cor:lin-com-orbits}Let $\+A\subseteq \.{su}(2^n)$, $\.g=\langle \+A\rangle_{{\rm Lie},\mb{R}}$ and $S\le S_n$. If $\pi(A)=A$ for all $A\in\+A$ and all $\pi\in S$, then every element of $\.g$ is a linear combination (over $\mb{R}$) of Pauli orbit-sums.
\end{cor}

Now recall Theorem \ref{thm:dim_dla_orbits}.

\dimfromorbits*


\begin{proof}
     By Corollary~\ref{cor:lin-com-orbits}, any $H\in \.g$ can be written as 
    \eq{
        H&=\sum_{O\in \+P_n/S} \alpha_O \sum_{P\in O} P,
    } 
    where $\alpha_O\in\mb{R}$. Thus, $\dim(\.g)\le \abs{\+P_n/S}$. As $\.g$ consists of only traceless matrices, we exclude the Pauli string $I_1I_2\ldots I_n$, which is the sole element in its orbit. 
\end{proof}




The above can be applied to QAOA-Maxcut on a graph $G=(V,E)$ by considering  $Aut(G)\defeq \{\pi\in S_{|V|}: (u,v)\in E \Leftrightarrow (\pi(u),\pi(v))\in E\}$, the group of automorphisms of $G$. 
\begin{thm}\label{thm:DLA-aut(G)}For any $G=(V,E)$ on $n$ vertices,  $\dim(\.g_G)\le \abs{\+P_n/Aut(G)} - 1$.
\end{thm}
\begin{proof}
Note that 
\begin{enumerate}
    \item $\pi$ is a permutation of vertices, and therefore $\pi(i\sum_{j=1}^n X_j) = i\sum_{j=1}^n X_{\pi(j)} = \sum_{j=1}^n X_j$;
    \item $\pi$ induces a permutation of edges, i.e. $\{(\pi(j),\pi(k)): (j,k)\in E\} = E$, thus $\pi(i\sum_{(j,k)\in E} Z_j Z_k) = i\sum_{(j,k)\in E} Z_{\pi(j)} Z_{\pi(k)} = i\sum_{(j,k)\in E} Z_j Z_k$.
\end{enumerate}
Thus, by Lemma~\ref{lem:pi-invariance}, $\pi(H)=H$ for all $\pi\in Aut(G)$ and all $H\in\.g_G$. The result follows from Theorem~\ref{thm:dim_dla_orbits}.
\end{proof}
For the complete graph $K_n$ on $n$ vertices, the DLA dimension then satisfies:
\begin{cor} $\dim(\.g_{K_n}) \le {n+3\choose 3}$. 
\end{cor}
\begin{proof} 
$Aut(K_n) = S_n$, and thus, by Theorem~\ref{thm:DLA-aut(G)}, $\dim(\.g_{K_n}) \le \abs{\+P_n/S_n}$. Any two Pauli strings are in the same orbit under the action of the symmetric group if and only if they contain the same numbers of $I, X, Y, Z$ operators.  $\abs{\+P_n/S_n}$ is therefore equal to number of partitions of $n$ into four nonnegative integers that sum to $n$.  The number of $k$-tuples of non-negative integers whose sum is $n$ is $\binom{n+k-1}{n} = \binom{n+k-1}{k-1}$ because we choose $n$ elements from $k$ with repetition.  Taking $k=4$ gives $\abs{\mcP_n/Aut(K_n)}=\binom{n+3}{3}$.
\end{proof}
We remark that ${n+3 \choose 3}$, which is the tetrahedral number $\operatorname{Te}(n+1)$, is the dimension of the space of $S_n$-equivariant $\.{su}(2^n)$ operators (those that commute with representations of $S_n$) \cite{albertini2018controllability,anschuetz2023efficient, nguyen2024theory}, and is also the dimension of the DLA spanned by $\{\sum_i X_i, \sum_i Y_k, \sum_{i<j}Z_iZ_j\}$~\cite{albertini2018controllability, schatzki2024theoretical}, of which $\.g_{K_n}$ is a subalgebra. 


We now show how the above upper bound on $\dim(\.g_{K_n})$ can be improved, essentially by a factor of $2$. While the $Aut(G)$ invariance of the DLA generators restricts all $g\in\.g_G$ to be linear combinations of Pauli orbit-sums, we shall see that not all orbit-sums are permitted.  
Let $n_X, n_Y, n_Z$ be the numbers of Pauli $X,Y,Z$ in a given Pauli string $P=P_1P_2\cdots P_n$, respectively. A Pauli string will be called $YZ$-even if $n_Y + n_Z$ is even. 

\begin{restatable}{lemma}{DLAParity}\label{thm:DLA-parity}
 For any $G =(V,E)$,  $\abs{V} =n \ge 2$, 
 \eq{
 \.g_{G}\subseteq \spn\lp \{i S : S \text{ is YZ-even}\}\setminus \{iX_1X_2\ldots X_n, iI_1I_2\ldots I_n\}\rp.
 }   
\end{restatable}


The upper bound of Eq.~\eqref{eq:kn-upper-bounds} is then stated as the following corollary.

\begin{cor}\label{cor:Kn-bounds-almost-tight}$\dim(\.g_{K_n}) \le \sum_{s=0}^{\lfloor n/2 \rfloor} (2s+1)(n-2s+1) -2$ \\

\hspace{3.85cm}$ = \begin{cases}
\frac{1}{12}\lp n^3+6n^2 + 14n - 12\rp, & \text{if $n$ is even}; \\
\frac{1}{12}\lp n^3 +6n^2 + 11n -18\rp, & \text{if $n$ is odd}.\end{cases}$
\end{cor}
\begin{proof} For each $s\in\{0,1,\ldots \lfloor n/2 \rfloor\}$, the number of ways to choose $n_Y+n_Z=2s$ is $2s+1$ ($(n_Y,n_Z)\in\{(0,2s), (1, 2s-1),\ldots, (2s,0)\}$).  For each $(n_Y,n_Z)$ pair, there are $n-2s+1$ values of $(n_X,n_I)$. Thus, the number of $YZ$-even Pauli string orbits is
$\sum_{s=0}^{\lfloor n/2 \rfloor} (2s+1)(n-2s+1)$. Excluding the $I_1\ldots I_n$ and $X_1\ldots X_n$ strings gives the result.
\end{proof}
We remark that as $\dim(\.g_{K_n}) < {n+3 \choose 3}$, the DLA for $\.g_{K_n}$ is a strict subalgebra of the $S_n$-equivariant $\.{su}(2^n)$ operators, and thus $\.g_k$ does not correspond to a subspace-controllable system. 


Lemma~\ref{thm:DLA-parity} shows that, despite being $YZ$-even, $I_1I_2\ldots I_n$ and $X_1X_2\ldots X_n$ cannot lie in $\.g_G$.  Interestingly, these are the only Pauli strings that lie in the centralizer of $\.g_G$ in $\mb{C}^{2^n\times 2^n}$.

\begin{restatable}{proposition}{XICommute}\label{lem:X_I_commute} 
    Let $G$ be a connected graph on $n$ vertices, and $P=P_1P_2\ldots P_n$ a Pauli string. If $[P,H] = 0$ for all $H\in\.g_{G}$ then $P=X_1X_2\cdots X_n$ or $P=I_1I_2\cdots I_n$. 
\end{restatable}

\label{sec:DLA-general-results}


\subsection{Cycle graphs}

We now focus on the DLA for QAOA-MaxCut on the $n$-node cycle graph $C_n$, and outline the proofs of Theorems~\ref{thm:Cn-decomp-center}-\ref{thm:Cn-exp-var} and related results. Throughout this section, we will use $\.g$ to denote $\.g_{C_n}$.

We will make use of the following well-known classification result for Lie algebras. 
\begin{thm}[\cite{hall2013lie}]\label{thm:classification}
    Every finite-dimensional simple Lie algebra over $\mathbb{C}$ is isomorphic to precisely one of the following.
    \begin{enumerate}
        \item $A_n=\.{sl}(n+1,\mathbb{C})$ for $n\ge 1$, of dimension $n^2+2n$;
        \item $B_n=\.{so}(2n+1,\mathbb{C})$ for $n\ge 2$, of dimension $2n^2+n$;
        \item $C_n=\.{sp}(n,\mathbb{C})$ for $n\ge 3$, of dimension $2n^2+n$;
        \item $D_n=\.{so}(2n,\mathbb{C})$ for $n\ge 4$, of dimension $2n^2-n$;
        \item The ``exceptional'' Lie algebras $G_2$, of dimension $14$; $F_4$, of dimension $52$; $E_6$, of dimension $78$; $E_7$, of dimension $133$; and $E_8$, of dimension $248$. 
    \end{enumerate}
    The dimensions of Cartan subalgebras of $A_n$, $B_n$, $C_n$, and $D_n$ are $n$, while those of $G_2$, $F_4$, and $E_j$ are $2$, $4$, and $j$ respectively, for any $j\in \{6,7,8\}$.
\end{thm}
As this classification result holds for complex simple Lie algebras, we will at times consider the complexification $[\.g,\.g]_\mb{C}$ of the semisimple component of $\.g$, and use results for $[\.g,\.g]_\mb{C}$ to make conclusions about $[\.g,\.g]$.


\subsubsection{Basis and dimension}

Let $\+B=\{X, X^{n-1}, ZX^tZ, YX^tY, YX^tZ : t\in\{0,1,\ldots, n-2\}\}$ be the set of Pauli orbit-sums defined in Eqs. \eqref{eq:Cn-basis-X}-\eqref{eq:Cn-basis-YXtZ}, and note that we adopt the convention that when $t=0$, there is no Pauli $X$; for example, ${ZX^0Z} = iD_n(Z_0Z_1)$, and we sometimes also write ${ZZ}$ for ${ZX^0Z}$. Table~\ref{tab:Cn-adA-adB-on-orbits} gives the commutators of the Pauli orbit-sums of $\+B$ with the DLA generators $X$ and $ZZ$.

\begin{restatable}{thm}{CnBasis}\label{thm:Cn-basis} For $n\ge 3$, $\dim(\.g)=3n-1$, and 
$\+B=\{X, X^{n-1}, ZX^tZ, YX^tY, YX^tZ : t\in\{0,1,\ldots, n-2\}\}$ is a basis of $\.g$.
\end{restatable}

\begin{proof} As the elements of $\+B$ are each summations of distinct Pauli strings, they are orthogonal (with respect to the Hilbert-Schmidt inner product) and thus also linearly independent. In the Supplementary Information, using Table~\ref{tab:Cn-adA-adB-on-orbits}, we show that $\.g = \spn_\mb{R} \+B$ from which the result follows.
\end{proof}


\begin{table}[htb!]
    \centering
     \caption{Commutators of the DLA generators with the elements of $\+B$. Values of $t$ are in the range $\{0,1,\cdots, n-2\}$. We adopt the convention that ${YX^{-1}Y} \defeq  {-X}$, ${ZX^{n-1}Z} \defeq {X^{n-1}}$ and $YX^{-1}Z=YX^{n-1}Z\defeq 0$.
    } 
    \label{tab:Cn-adA-adB-on-orbits}
    {\begin{tabular}{c|c c c c c}
    $[\cdot, \cdot]$  & $X$  & ${Z X^t Z}$ & ${Y X^t Y}$ & ${Y X^t Z}$ & ${X^{n-1}}$\\
    \hline
    $X$  & $0$ & ${2} {Y X^t Z}$ & ${-2}{Y X^t Z}$ & ${4} ({Y X^t Y}-{Z X^t Z})$ & $0$    \\
    ${ZZ}$ & ${-2}{YZ}$ & ${-}{2} {YX^{t-1}Z}$ 
    & ${2} {YX^{t+1}Z}$  
    & ${4} ({ZX^{t+1}Z}-{YX^{t-1}Y})$ 
    &  ${-2}{YX^{n-2}Z}$
    \end{tabular}}
\end{table}



\subsubsection{Center, Cartan subalgebra, and simple Lie algebra decomposition}

With a basis for $\.g$ identified, we now outline the proof of Theorem~\ref{thm:Cn-decomp-center} which has two components: a characterization of the center $\.c$, which appears below as Theorem~\ref{thm:cycle-center}; and a direct sum decomposition of $\.g$ into the center and $n-1$ simple Lie algebras, each isomorphic to $\.{su}(2)$, which we show in Theorem~\ref{thm:Cn-decomp}. 
\begin{restatable}{thm}{CycleCenter}\label{thm:cycle-center}
    The Lie algebra $\.g$ has a $2$-dimensional center $\.c = \spn_{\mbR} \{c_1, c_2\}$, where  
   \begin{align*}
    c_1 = - X +\sum_{t=1}^{\frac{n-1}{2}} \big({ZX^{2t-1}Z} + {YX^{2t-1}Y}\big),\text{~~and~~}
     c_2 = {X^{n-1}} +\sum_{t=0}^{\frac{n-3}{2}}\big({ZX^{2t}Z} + {YX^{2t}Y}\big)
\end{align*}
if $n$ is odd, and 
\begin{align*}
    c_1  = {X^{n-1}} - X + \sum_{t=1}^{\frac{n-2}{2}}\big({ZX^{2t-1}Z} + {YX^{2t-1}Y}\big),\text{~~and~~} c_2  = \sum_{t=0}^{\frac{n-2}{2}} \big({ZX^{2t}Z} + {YX^{2t}Y}\big)
\end{align*}
if $n$ is even.
\end{restatable}
\begin{proof}
Any $H\in\.g$ can be expressed as a linear combination of the basis vectors, viz., 
    \[H = \alpha X + \sum_{t=0}^{n-2}(\beta_{1,t} {ZX^tZ} + \beta_{2,t} {YX^tY} + \beta_{3,t} {YX^tZ}) + \gamma {X^{n-1}},\quad \text{for }\alpha, \beta_{j,k},\gamma \in\mathbb{R}.\]
    It is easily verified by induction and the Jacobi identity that $H\in \.c$ if and only if $[X,H] = [ZZ,H] = 0$. Using Table \ref{tab:Cn-adA-adB-on-orbits}, $[X,H]$ and $[ZZ,H]$ can be computed, giving two linear equations for the coefficients $\alpha, \beta_{j,k}, \gamma$ which, when solved, yield the result. See Supplementary Information for more details.
\end{proof}

Additionally, from the Jacobi identity and Table~\ref{tab:Cn-adA-adB-on-orbits}, it is not hard to verify that the semisimple component $[\.g,\.g]_\mb{C}$ of $\.g_\mb{C}$ satisfies: 
\begin{thm}\label{thm:semi_dim}
    $\dim([\.g,\.g]_{\mbC})=3(n-1)$, and 
    \begin{align*}
        {[\.g,\.g]_{\mbC} = \spn_\mbC \{{ZX^{t+1} Z}-{YX^{t-1} Y},\ {YX^{t}Y} - {ZX^tZ}, \ {YX^tZ}: 0\le t \le n-2 \}.}
    \end{align*}
\end{thm}

The next step will be to determine the dimension of a Cartan subalgebra of $[\.g,\.g]_\mb{C}$. To proceed, note that $[YX^tZ, YX^sZ]=0$ for all $0\le s,t\le n-2$, and thus the vector space $\.h'\defeq\spn_{\mb{C}}\{YX^tZ : 0\le t \le n-2\}$ is commutative.  We then make use of the following facts.

\begin{restatable}{fact}{AdHDiag}\label{fact:adH-diag}
{The operator} $ad_H$ is diagonalizable, for all $H\in\.h'$.
\end{restatable}

\begin{restatable}{fact}{CartanSameDim}\label{fact:cartan-same-dim}All Cartan subalgebras of a complex semisimple Lie algebra have the same dimension.
\end{restatable}

\begin{restatable}{fact}{CartanSemisimple}\label{fact:cartan-semisimple}Let $\.g=\.g_1\oplus\.g_2\oplus\cdots \oplus \.g_k$ be a complex semisimple Lie algebra, with the $\.g_j$ simple. If $\.h_j$ is a Cartan subalgebra of $\.g_j$, then $\.h=\.h_1 \oplus \.h_2 \oplus \cdots\oplus\.h_k$ is a Cartan subalgebra of $\.g$. 
\end{restatable}

\begin{restatable}{fact}{DimCartanGreaterThree}\label{fact:dim-cartan-greater-three}The dimension of any simple complex Lie algebra is at least three times the dimension of its Cartan subalgebra.
\end{restatable}


Facts \ref{fact:adH-diag} and \ref{fact:cartan-same-dim} imply that $\dim(\.h)\ge\dim(\.h') = n-1$, for any Cartan subalgebra $\.h\subseteq[\.g,\.g]_\mb{C}$. Fact \ref{fact:cartan-semisimple} shows that the dimension of a Cartan subalgebra is $\dim(\.h)=\sum_{j=1}^k \dim(\.h_j)$.  Accounting for Fact \ref{fact:dim-cartan-greater-three} then gives 
\begin{align}\label{eq:Cn-subspace-dim}
    3n-3 = \dim([\.g,\.g]_\mbC) = \sum_{j=1}^k \dim(\.g_j) \ge \sum_{j=1}^k 3 \dim(\.h_j) = 3\dim(\.h),  
\end{align}
which implies $\dim(\.h)\le n-1$.  It follows that:




\begin{thm}\label{thm:cycle-Cartan}
  Let $\.h$ be a Cartan subalgebra of $[\.g,\.g]_{\mbC}$. Then, $\dim(\.h)=n-1$.
\end{thm}

Note that Theorem \ref{thm:cycle-Cartan} implies that the inequality in Eq.~\eqref{eq:Cn-subspace-dim} actually takes equality, and furthermore $k=n-1$. From Theorem~\ref{thm:classification}, apart from $\.{sl}(2,\mb{C})$ which has dimension $3$ and a one dimensional Cartan subalgebra, the dimension of all other simple complex Lie algebras is strictly larger than three times the dimension of their Cartan subalgbras.  We therefore must have that $\.g_j\cong\.{sl}(2,\mb{C})$ for all $j=1,2,\ldots, n-1$, i.e.,

\begin{thm}\label{thm:Cn-g'-decomp}
    $[\.g,\.g]_{\mathbb C}\cong \underbrace{\.{sl}(2,\mathbb{C})\oplus \cdots \oplus\.{sl}(2,\mathbb{C})}_{n-1~{\rm times}}$.
\end{thm}

We now use the decomposition of $[\.g,\.g]_{\mbC}$ to identify $[\.g,\.g]$. 
\begin{restatable}{lemma}{DirectSumComplexification}\label{lem:direct-sum-complexification}
    Consider a matrix Lie algebra $\.g\subseteq \mathcal M_n(\mb{C})$ over the real field $\mb{R}$, and its complexification $\.g_{\mb{C}}$. 
    \begin{enumerate}
        \item If $\.g = \.g_1 \oplus \cdots \oplus \.g_n$, then  $\.g_{\mb{C}} = (\.g_1)_{\mb{C}} \oplus \cdots \oplus (\.g_n)_{\mb{C}}$. 
        \item Conversely, if $\.g_{\mb{C}} = \.g_1' \oplus \cdots \oplus \.g_n'$ for some complex Lie algebras $\.g_1', \ldots, \.g_n'$, then $\.g = \.g_1 \oplus \cdots \oplus \.g_n$, and $(\.g_j)_{\mb{C}} = \.g_j'$ for each $j=1, \ldots, n$.
    \end{enumerate}
\end{restatable}

In addition, the following is well known.
\begin{lem}[\rm\cite{hota2014real}]\label{lem:sl2c-complexification}
$\.{sl}(2,\mbR)_\mbC = \.{su}(2)_\mbC = \.{sl}(2,\mbC)$, and the only real Lie algebras whose complexification is $\.{sl}(2,\mbC)$ are $\.{sl}(2,\mbR)$ (non-compact) and  $\.{su}(2)$ (compact).
\end{lem}
Combined with the results shown earlier, these two lemmas allow us to conclude:
\begin{thm}\label{thm:Cn-decomp} 
The Lie algebra $\.g$ has the following direct sum decomposition
\[
    \.g = \.c \oplus \.g_1 \oplus \cdots \oplus\.g_{n-1},
\]
where the center $\.c$ is given in Theorem \ref{thm:cycle-center}.
\end{thm}
\begin{proof}
    By Theorem~\ref{thm:Cn-g'-decomp} and Lemma \ref{lem:direct-sum-complexification},  $[\.g,\.g] = \.g_1\oplus \cdots \oplus \.g_{n-1}$ for $(n-1)$ 3-dimensional real Lie algebras $\.g_1, \ldots, \.g_{n-1}$, whose complexifications are all isomorphic to $\.{sl}(2,\mb{C})$. Lemma~\ref{lem:sl2c-complexification} says that the only real Lie algebras whose complexification is $\.{sl}(2,\mb{C})$ are $\.{su}(2)$ and $\.{sl}(2,\mb{R})$.  Since a subalgebra of a compact Lie algebra is compact, $[\.g,\.g]$---as a subalgebra of the compact Lie algebra $\.{su}(2^n)$---is compact. Thus each $\.g_j$ must be isomorphic to $\.{su}(2)$.
\end{proof}

\subsubsection{Isomorphism}

Having shown that $[\.g,\.g]_\mb{C}\cong \bigoplus_{n-1}\.{sl}(2,\mb{C})$, we now outline how to explicitly construct such an isomorphism, by giving a basis for each simple component in Theorem~\ref{thm:Cn-g'-decomp} that respects the $\.{sl}(2,\mb{C})$ commutation relations.


\begin{defi}[Canonical basis] \label{def:sl2C-canonical-basis}
    We say that a set of nonzero matrices $\Gamma = \bigcup_{k=1}^{n-1}\{U_k, V_k, H_k\}$ is a canonical basis for $[\.g,\.g]_{\mbC} \cong \bigoplus_{n-1}\.{sl}(2,\mb{C})$ if $\Gamma$ satisfies the following relations:
    \begin{align}\label{eq:sl2C-canonical-basis}
        [H_k, U_k]&=2U_k, \quad [H_k, V_k]=-2V_k, \quad [U_k, V_k]=H_k,
    \end{align}
    and $[x,y]= 0$ for any $x\in\{U_k, V_k, H_k\}$ and any $y\in\{U_j, V_j, H_j\}$ where $j \neq k$.
\end{defi}
Note that $\Gamma$ is indeed a basis for $[\.g,\.g]_{\mbC}$ as a vector space, i.e., the elements are linearly independent. To see this, for any $x = \sum_k (a_k U_k + b_k V_k + c_k H_k)$ we must show that $x=0$ implies $a_k = b_k = c_k = 0$ for all $k$. One can easily verify that $[H_k,[U_k,x]] = -4c_kU_k$, thus $x=0$ and $U_k\ne 0$ imply $c_k = 0$. The other coefficients can be similarly handled by noticing $[U_k,[H_k,x]] = -2b_k H_k$ and $[V_k,[V_k,x]] = -2a_kV_k$. 


\begin{restatable}[Explicit isomorphism]{thm}{ExplicitIso}\label{thm:Cn-isomorphism} 
    For each $k=1, 2, \ldots, n-1$, let 
    \begin{align*}
        H_k & \defeq {-}\frac{{i}}{2n} \sum_{j=1}^{n-1} \sin \frac{kj\pi}{n} YX^{j-1}Z, \\
        U_k & \defeq 
        -\frac{1}{4n} \sum_{j=0}^{n-1} \lp e^{-ik(j+1)\pi/n} YX^{j-1} Y + e^{i kj\pi/n} ZX^j Z \rp, \\
        V_k & \defeq 
        \frac{1}{4n} \sum_{j=0}^{n-1} \lp e^{ik(j+1)\pi/n} YX^{j-1} Y + e^{-i kj\pi/n} ZX^j Z\rp.
    \end{align*}
Then, $\{U_k, V_k, H_k: k = 1, 2, \ldots, n-1\}$ is a canonical basis for $[\.g, \.g]_\mbC$.
\end{restatable}

Note that, given a basis $\{U_k, V_k, H_k\}_{k=1}^{n-1}$ for $\.{sl}(2,\mbC)$, one can define a set of vectors $\{\tilde X_k = i(U_k+V_k), \tilde Y_k = V_k-U_k, \tilde Z_k=iH_k\}_{k=1}^{n-1}$ which satisfies the $\.{su}(2)$ commutation relations 
\eq{
[\tilde X_k, \tilde Y_k] = 2\tilde Z_k, \quad [\tilde Y_k, \tilde Z_k] = 2\tilde X_k, \quad [\tilde Z_k, \tilde X_k] = 2\tilde Y_k,}
and therefore $\{\tilde{X_k}, \tilde{Y}_k ,\tilde{Z}_k\}_{k=1}^{n-1}$ forms a canonical basis for $n-1$ copies of $\.{su}(2)$ (defined in an analogous fashion to a canonical basis for $n-1$ copies of $\.{sl}(2,\mb{C})$). This immediately gives Theorem~\ref{thm:cycle-iso-su2}.

To prove Theorem~\ref{thm:Cn-isomorphism}, one approach is to simply verify that $\{U_k, V_k, H_k\}_{k=1}^{n-1}$ satisfies the $\.{sl}(2,\mb{C})$ commutation relations, and that elements from different simple ideals of $[\.g,\.g]_\mb{C}$ commute.  This involves elementary but tedious manipulations of trigonometric identities, and we defer these calculations to the Supplementary Material.  Here we outline another proof technique, which is how we originally identified the isomorphism.

The idea is based on the following theorem,  which gives conditions under which one canonical basis for $\.{sl}(2,\mb{C})$ can be transformed into another canonical basis for $\.{sl}(2,\mb{C})$ --- one which is expressed in terms of two DLA elements (which we will take to be the DLA generators) and a chosen Cartan subalgebra.

\begin{restatable}{thm}{CanonicalFromAB}\label{thm:AB} Let $\bigcup_{k=1}^{n-1}\{U_k, V_k, H_k\}$ be a canonical basis for $[\.g,\.g]_\mb{C}\cong\bigoplus_{n-1}\.{sl}(2,\mb{C})$.
Without loss of generality, any two elements $A,B\in[\.g,\.g]_\mb{C}$ can be expressed as
\eql{
A&= \sum_{j=1}^{n-1}\lp  a_j^{(u)}U_j + a_j^{(v)} V_j+ a_j^{(h)}H_j\rp, \quad
B= \sum_{j=1}^{n-1}\lp  b_j^{(u)}U_j + b_j^{(v)} V_\alpha +b_j^{(h)}H_j\rp, \label{eq:AB-wlog}
}
with $a_j^{(u)}, a_j^{(v)}, a_j^{(h)}, b_j^{(u)}, b_j^{(v)}, b_j^{(h)}\in\mb{C}$ for all $j$.
Define, for each $k\in\{1,\ldots, n-1\}$,
\eql{
\tilde{U}_k^{(m_k, c_k)}&:= \left[H_k, -\frac{i}{c_k}\lp A + e^{-i\pi m_k}B\rp\right], \quad 
\tilde{V}_k^{(m_k, c_k)}:= \left[H_k, \frac{i}{c_k}\lp A + e^{i\pi m_k}B\rp\right], \label{eq:uvh-alpha}
}
where $m_k, c_k\in\mb{R}, c_k > 0$. Then, $\bigcup_{k=1}^{n-1}\{\tilde{U}_k^{(m_k, c_k)}, \tilde{V}_k^{(m_k, c_k)}, H_k\}$ is a canonical basis for $[\.g,\.g]_\mb{C}$ if and only if
\eql{
a_k^{(u)} + e^{i\pi m_k}b_k^{(u)}&=0, \label{eq:C1}\tag{C1} \\
a_k^{(v)} + e^{-i\pi m_k}b_k^{(v)}&=0, \label{eq:C2}\tag{C2} \\
a_k^{(u)} + e^{-i\pi m_k}b_k^{(u)}&\neq 0, \label{eq:C3}\tag{C3} \\
a_k^{(v)} + e^{i\pi m_k}b_k^{(v)}&\neq 0, \label{eq:C4}\tag{C4} \\
c_k^2  &=4\lp a_k^{(v)}b_k^{(u)} - a_k^{(u)} b_k^{(v)} \rp
 \lp e^{i\pi m_k} - e^{-i\pi m_k}\rp. \label{eq:C5}\tag{C5}
 }
\end{restatable}
We remark that Eqs. \eqref{eq:C1}-\eqref{eq:C4} imply that $x\neq 0$, for all $x \in\ \bigcup_{k=1}^{n-1}\{a_k^{(u)}, a_k^{(v)}, b_k^{(u)}, b_k^{(v)}\}$. 
The aim will be to consider the case where $A=X$ and $B=ZZ$, i.e., $A,B$ correspond to the DLA generators of $[\.g,\.g]_{\mb{C}}$\footnote{Strictly speaking,  we view $X$ and $ZZ$ as generators of $\.g_{\mb{C}}=\.c_\mb{C}\oplus [\.g,\.g]_\mb{C}$ and thus, without loss of generality they can be expressed as $X= A + A_c$, $ZZ = B + B_c$, where $A,B$ are as in Eq. \eqref{eq:AB-wlog} and $A_c,B_c$ are the central components of $X$ and $ZZ$.  However, as the central components vanish under taking commutator, they will not impact our analysis, and we can therefore ignore $A_c$ and $B_c$ in our discussion.}, and choose the $\{H_k\}$ in such a way that conditions Eqs.~\eqref{eq:C1}-\eqref{eq:C5} are satisfied, so that one can construct a canonical basis $\{\tilde{U}_k, \tilde{V}_k, H_k\}$ (suppressing the $(m_k,c_k)$ superscripts) via Eq.~\eqref{eq:uvh-alpha}, even without initially knowing an explicit form for the $\{U_k\}$ and $\{V_k\}$.


To proceed, let us define some notation. For any $A,B\in\.g,~ k\in\mb{N}$, we use $(AB)^k$ to denote $(ad_A ad_B)^{k-1}[A,B]$, and similarly $A(AB)^k = ad_A((AB)^k)$, $B^2(AB)^k = ad_B (ad_B ((AB)^k))$ etc.
When $A=X$, $B=ZZ$, we use Table~\ref{tab:Cn-adA-adB-on-orbits} and induction on $k$ to prove the following lemma.
\begin{restatable}{lemma}{CnGenBasis}\label{thm:Cn-generating-basis}
    {For $A=X, B=ZZ$, and any $k\in\{1,2,\ldots,n-1\}$, }
    \begin{align}
        (AB)^k & = \sum_{j=0}^{k-1} c_{k,j} YX^j Z, \label{eq:ABk-recursion}
    \end{align}  
    where the $c_{k,j}$ can be computed from the recursive relation $ c_{k,j} = 8(c_{k-1,j-1} + c_{k-1,j+1})$, with initial values $c_{1,0} = 2$ and $c_{0,j} = 0, \forall j\in \{0,1,\ldots,k\}$; and $c_{k,j} = 0$ if $k\notin [n]$ or $j\notin\{0,\ldots, k-1\}$. 
\end{restatable}
Furthermore, the recursive relation for the coefficients $c_{k,j}$ can be solved to give
\eql{
c_{k,j}&= \frac{2^{4k-2}}{n} \sum_{j'=1}^{n-1} \sin\frac{(j+1)j'\pi}{n} \sin \frac{j'\pi}{n} \cos^{k-1} \frac{j'\pi}{n},  \label{eq:ckj-solved}
}
which, when substituted into Eq.~\eqref{eq:ABk-recursion} can be used to prove the following relation between $(AB)^n$ and lower powers of $AB$.
\begin{restatable}{thm}{CnABnbyAbk}\label{thm:Cn-ABn-by-ABk}
    \eql{
    (AB)^n = \sum_{k=1}^{n-1} a_k (AB)^k, \label{eq:el-sym-poly}
    }
    where \[a_k = (-1)^{n-k+1} \sigma_{n-k} (16 \cos(\pi/n), \ldots, 16 \cos((n-1)\pi/n)),\]
    with $\sigma_k$ the $k$-th elementary symmetric polynomial in $n-1$ variables.
\end{restatable}
The expression Eq.~\eqref{eq:ABk-recursion} is useful, and from it a number of facts can be proved (see Supplementary Information). In particular:
\begin{restatable}{fact}{ABisCartan}\label{fact:ab-is-cartan}$\.h\defeq\spn_{\mb{C}}\{(AB), (AB)^2,\ldots, (AB)^{n-1}\}$ is a Cartan subalgebra of $[\.g,\.g]_\mb{C}$.
\end{restatable}

\begin{restatable}{fact}{GenNoOverlap}\label{fact:gen-no-overlap}The generators $A=X$ and $B=ZZ$ have no overlap in $\.h$, i.e., $\langle A, H\rangle = \langle B, H\rangle = 0$ for any $H\in \.h$.
\end{restatable}

\begin{restatable}{fact}{AsquaredAB}\label{fact:A2AB}$A^2(AB)^j = B^2(AB)^j = -16(AB)^j$.
\end{restatable}

We are now in a position to prove Theorem~\ref{thm:Cn-isomorphism}.
\begin{proof}[Proof of Theorem~\ref{thm:Cn-isomorphism}]

As the $n-1$ vectors $\{H_k\}_{k=1}^{n-1}$ 
in Theorem~\ref{thm:AB}
mutually commute, they form a Cartan subalgebra of $[\.g,\.g]_\mb{C}$. In addition, for any two Cartan subalgebras $\.h_1, \.h_2$ of a complex semisimple Lie algebra $\.k$ 
there exists a bijective linear map $\phi:\.k\ra\.k$ satisfying $[\phi(A), \phi(B)] =\phi([A,B])$ and $\phi(\.h_1) = \.h_2$ (see e.g,~\cite{knapp1996lie}). Thus, from Fact~\ref{fact:ab-is-cartan} above, without loss of generality we can take the $\{H_k\}_{k=1}^{n-1}$ such that $\spn_\mb{C}\{H_1, \ldots,  H_{n-1}\} = \spn_\mb{C}\{AB, (AB)^2,\ldots, (AB)^{n-1}\}$. 
Fact~\ref{fact:gen-no-overlap} then implies that the coefficients $a_j^{(h)}=b_j^{(h)} = 0$ in Eq.~\eqref{eq:AB-wlog}, for all $j$. Direct calculation gives $[A,[A,H_j]] = 4a_j^{(u)}a_j^{(v)} H_j $ and $ [B,[B, H_j]] = 4b_j^{(u)}b_j^{(v)} H_j$ and, together with Fact~\ref{fact:A2AB}, one can see that $a_j^{(u)}a_j^{(v)} = b_j^{(u)}b_j^{(v)} = -4$.  With  $a_j^{(h)}=b_j^{(h)} = 0$, evaluating the commutators $[A,B]$ and $[A,[B,H_j]]$ gives
\eql{
[A,B]&= \sum_{j=1}^{n-1}\nu_j H_j, \label{eq:AB-nuH}\\
[A,[B,H_j]]&= \lambda_j H_j, \quad \forall j\in\{1,2,\ldots, n-1\},\label{ABH-lambdaH}
}
where $\nu_j=a_j^{(u)}b_j^{(v)}-a_j^{(v)}b_j^{(u)}$ and $\lambda_j = 2(a_j^{(u)}b_j^{(v)}+a_j^{(v)}b_j^{(u)})$. Repeatedly applying Eq.~\eqref{ABH-lambdaH} to Eq.~\eqref{eq:AB-nuH} gives 
\eql{(AB)^k &= \sum_{j=1}^{n-1}\nu_j\lambda_j^{k-1}H_j \label{eq:ABk-nulambdaH}
}
for all $k\in\{1,\ldots, n\}$.  Substituting Eq.~\eqref{eq:ABk-nulambdaH} into Eq.~\eqref{eq:el-sym-poly} then implies, via Vieta's formulas (see Supplementary Information), that 
\eql{
\lambda_j = 16\cos(j\pi/n), \label{eq:lambda_j}
}
from which it follows that $a_j^{(u)}b_j^{(v)}+a_j^{(v)}b_j^{(u)} = 8\cos(j\pi/n)$. Combined with $a_j^{(h)}=b_j^{(h)} = 0$ and $a_j^{(u)}a_j^{(v)} = b_j^{(u)}b_j^{(v)} = -4$ gives
\eq{
\frac{a_j^{(u)}b_j^{(v)}+a_j^{(v)}b_j^{(u)}}{b_j^{(u)}b_j^{(v)}} = \frac{a_j^{(u)}}{b_j^{(u)}} + \frac{a_j^{(v)}}{b_j^{(v)}} = \frac{a_j^{(u)}}{b_j^{(u)}} + \frac{b_j^{(u)}}{a_j^{(u)}} = -2\cos(j\pi/n).
}
Thus, $a_j^{(u)}/b_j^{(u)}=-e^{\pm i j \pi/n}$. We take $a_j^{(u)}/b_j^{(v)} = -e^{i j\pi/n}$ (taking $a_j^{(u)}/b_j^{(v)} = -e^{-i j\pi/n}$ leads to another valid isomorphism), and $b_j^{(u)}=-2$, obtaining the following parameters:
\eq{
a_j^{(u)}= 2e^{i j\pi/n}, \quad a_j^{(v)}=-2e^{-i j\pi/n}, \quad b_j^{(u)}=-2,\quad b_j^{(v)}=2.
}
Taking $c_j = 8 \sin(j\pi/n)$, these parameters satisfy Eqs.~\eqref{eq:C1}-\eqref{eq:C5} and thus, by Theorem~\ref{thm:AB},
\begin{multline}\label{eq:new-canonical-basis}
   \left\{\left[H_k, -\frac{i}{8\sin(k\pi/n)}\lp X + e^{-i\pi k/n }ZZ\rp\right]\right.,\\ \left.\left[H_k, \frac{i}{8\sin(k\pi/n)}\lp X + e^{i\pi k/n }ZZ\rp\right],  H_k : 1\le k\le n-1\right\} 
\end{multline}
is a canonical basis for $n-1$ copies of $\.{sl}(2,\mb{C})$. What remains to do is express the $H_k$ in terms of $A=X$ and $B=ZZ$. Eq. \eqref{eq:ABk-nulambdaH} can be expressed as the matrix equation
\eql{
\bpm (AB) \\(AB)^2 \\ \vdots \\(AB)^{n-1}\epm &=\bpm 1 & 1 & \ldots & 1 \\ \lambda_1 & \lambda_2 &\ldots & \lambda_{n-1}\\\vdots \\
\lambda_1^{n-1} & \lambda_2^{n-1} &\ldots & \lambda_{n-1}^{n-1}\epm
\bpm \nu_1 H_1 \\ \nu_2 H_2 \\ \vdots \\\nu_{n-1}H_{n-1}\epm \label{eq:AB-vandermonde}
}
where 
$\lambda_j = 16\cos(j\pi/n)$ and 
$\nu_j=8i\sin(j\pi/n)$.  As the $\lambda_j$ are distinct, the Vandermonde matrix in Eq.~\eqref{eq:AB-vandermonde} can be inverted, to express the $H_j$ in terms of the $(AB)^k$. In the Supplementary Information, we show how this can be done which, when combined with Eq. \eqref{eq:ABk-recursion}, gives
\eql{
H_k &= -\frac{i}{2n}\sum_{j=1}^{n-1}\sin\frac{jk\pi}{n}YX^{j-1}Z. \label{eq:H-YXZ}
}
Finally, using Table~\ref{tab:Cn-adA-adB-on-orbits} to evaluate the commutators in Eq. \eqref{eq:new-canonical-basis} then gives the canonical basis of Theorem~\ref{thm:Cn-isomorphism}.
\end{proof}

\subsubsection{$\.g$-purities}\label{sec:purity}
Recall the expectation and variance of the cost function from Eq.~\eqref{eq:VQA-exp-var}:
\eq{
\av_\theta[\ell(\rho,{O};\theta)] &= \tr(\rho_{\.c}O_{\.c}),\\
\var_\theta[\ell(\rho,{O};\theta)] &=\sum_{k=1}^{n-1}\frac{ \mcP_{\.g_k}(\rho)\mcP_{\.g_k}(O)}{\dim(\.g_k)},
}
where, the $\.g_k$-purity of a Hermitian operator $H$ is defined as $\mcP_{\.g_k}(H)\defeq \sum_{j=1}^{\dim(\.g_k)}\abs{\tr (E_j^\dag H )}$, and $\{E_j : j=1,\ldots, \dim(\.g_k) \}$ is an orthonormal basis for $\.g_k$.

For QAOA-MaxCut, the normalized measurement operator $O = -\frac{i}{\sqrt{n}}ZZ \in i\.g$, and it is straightforward to use Theorems~\ref{thm:cycle-iso-su2} and \ref{thm:cycle-center}  to show that
\eql{
\mcP_{\.c}(O) &= \mcP_{\.g_k}(O) = \frac{2^n}{n}. \label{eq:op-purities}
}
To compute the purities of the initial state $\rho = \ket{+^n}\bra{+^n}$, where $\ket{+^n}=\frac{1}{\sqrt{2^n}}\sum_{z\in\{0,1\}^n}\ket{z}$, note that $\rho$ can be expressed as
\begin{align}
    \rho &= \frac{1}{2^n}\lp I + \sum_j X_j + \sum_{j\ne k} X_j X_k + \cdots + X_0 \ldots X_{n-1}\rp \nonumber\\
    &= \frac{1}{2^n}\lp I - i X - i X^{(2)} - \cdots - i {X^{(n-1)}} -i {X^{(n)}}\rp,\label{eq:rho-express}
\end{align}
where $X^{(t)}\defeq i\sum_{\text{distinct }j_1, j_2, \ldots, j_t\in \{0,1,\ldots,n-1\}}X_{j_1}X_{j_2}\cdots X_{j_t}$.  From Theorem~\ref{thm:Cn-basis},  
\eq{
\+B&=\{X, X^{n-1}, ZX^t Z, YX^t Y, YX^t Z : t \in\{0, 1, \ldots, n-2\}\}
}is a basis for $\.g$, from which we can see that the $X$ and $X^{n-1}$ components of $\rho$ are in $\.g$, and all other components $X^{(t)}$ with $t\ne 1$ or $n-1$ have zero overlap with $\.g$. Thus, for any Lie subalgebra $\.s$ of $\.g$, $\mcP_{\.s}(\rho)= -\frac{i}{2^n}\mcP_{\.s}(X + X^{(n-1)})$. By some basic calculations and trigonometric identities, it can then be shown that
\eql{
\mcP_{\.c}(\rho)=\frac{1}{2^{n-1}},\quad \mcP_{\.g_k}(\rho)= \frac{\parity(k)}{2^{n-2}}. \label{eq:state-purities}
}
From Eqs. \eqref{eq:op-purities} and \eqref{eq:state-purities}, and noting that $\dim(\.g_k)=3$ for all $1\le k\le n-1$, the variance is thus
\begin{align*}
\var_\theta[\ell(\rho,{O};\theta)] &=\sum_{k=1}^{n-1}\frac{ \mcP_{\.g_k}(\rho)\mcP_{\.g_k}(O)}{\dim(\.g_k)}= \sum_{k=1}^{n-1}\frac{\frac{\parity(k)
        }{2^{n-2}} \frac{2^n}{n}}{3}\\ 
        &= \frac{4}{3n} \sum_{k=1}^{n-1} \parity(k)
        =  \frac{2(n-\parity(n))}{3n}.
\end{align*}
The calculation of the expectation value of the cost function is elementary, given the basis for $\.c$ in Theorem~\ref{thm:cycle-center}. We defer details to the Supplementary Information, and merely state the result:
\eq{
\av_\theta[\ell(\rho,{O};\theta)] &= \tr(\rho_{\.c}O_{\.c}) = \frac{\parity(n)}{\sqrt{n}}.
}


\label{sec:DLA-cycle}

\subsection{Complete graphs}
Here we show how to construct a basis for $\.g=\.g_{K_n}$ (in the remainder of this section we will omit the subscript $K_n$).  Constructing a basis for $[\.g,\.g]$ is conceptually similar, and we defer details to the Supplementary Information.  The dimension of $\.c$ (as stated in Theorem~\ref{thm:kn-semisimple-center}) then follows from the fact that $\dim(\.c) = \dim(\.g)-\dim([\.g,\.g])$. We remark that, while we can exactly determine the dimension of $\.g_{K_n}$, we are not able to give a closed form expression for the variance of the loss function for $K_n$ (as we were able to for $C_n$, and thus we cannot draw a definitive conclusion about BPs in this case.

Recall the Pauli orbit-sums $X^pY^qZ^r$ from Eq.~\eqref{eq:XpYqZr-basis}, with the convention that we drop the superscript if it is $1$ and omit writing the Pauli operator if the superscript is $0$.  With this notation, the two generators are $X^1$ and $Z^2$. For convenience, we also adopt the convention that 
\[X^p Y^q Z^r=0, \quad \text{when}~p<0,~ q<0,~ r<0 \text{ or }p+q+r>n.\]
The following commutation relations are easily verified:
 \begin{align}
    [X^1, X^pY^qZ^r] =& 2(q+1)X^pY^{q+1}Z^{r-1}-2(r+1)X^pY^{q-1}Z^{r+1}, \label{eq:Kn-adA}\\
    [Z^2,X^pY^qZ^r] = & 
    \alpha_1 X^{p+1}Y^{q-1}Z^{r-1}+\alpha_2 X^{p+1}Y^{q-1}Z^{r+1} \nonumber\\
    &- \alpha_3 X^{p-1}Y^{q+1}Z^{r-1}-\alpha_4X^{p-1}Y^{q+1}Z^{r+1}, \label{eq:Kn-adB}
\end{align}    
    where 
\eq{
\alpha_1&=2(n-p-q-r+1)(p+1), &\alpha_2=2(p+1)(r+1),\\
\alpha_3&=2(n-p-q-r+1)(q+1), &\alpha_4=2(q+1)(r+1).
}   

In particular, $ad_{X^1}$ changes the parities of $Y$ and $Z$ and leaves the number of $X$ unchanged, and $ad_{Z^2}$ changes the parities of all $X$, $Y$ and $Z$. 

In the Supplementary Information, we prove the following and other related facts.

\begin{restatable}{lem}{AdZZPQR}\label{lem:adzz-pqr} For any $p$ and any odd $q\ge 3$ and any odd $r$,
\begin{align*}
    ad_{Z^2}(X^pY^qZ^{r}) =  & \beta_1 ad_{X^1}(X^{p+1}Y^{q-2}Z^{r})+\beta_2ad_{X^1}(X^{p+1}Y^{q-2}Z^{r+2})  -\\
&\beta_3ad_{X^1}(X^{p-1}Y^{q}Z^{r})-\beta_4ad_{X^1}(X^{p-1}Y^{q}Z^{r+2})+\beta_5 ad_{Z^2}(X^{p}Y^{q-2}Z^{r+2}),
\end{align*}
where 
    \eq{
        & \beta_1=\frac{(n-p-q-r+1)(p+1)}{q-1}, 
         \qquad\beta_2 =\frac{(r+1)(p+1)}{q-1}, \\
         &\beta_3=n-p-q-r+1, \qquad \beta_4= r+1,
          \qquad \beta_5=\frac{r+1}{q-1}.
    }
\end{restatable}

Below we use shorthand notation to denote the $p,q,r$ values in $X^pY^qZ^r$, e.g.,  ``$10e$'' denotes $p=1$, $q=0$, $r$ is even, ``$*oo$'' denotes $p$ can be anything, $q,r$ odd, etc. 
\begin{restatable}[``$01o,11o,10e$'']{fact}{ZeroOneOdd}\label{fact:Kn-01o-11o-10e}
     $Y^1Z^{r}, X^1Y^1Z^{r} \in [\.g,\.g]$ for all odd $r$, and $X^1Z^{r}\in \.g$ for all even $r$. 
\end{restatable}

\begin{restatable}[``$*oo$'' and ``$oee$''] {fact}{StarOddOdd}\label{fact:anyoddodd-oddeveneven}
\begin{enumerate}
\item []
    \item $X^pY^qZ^r\in[\.g,\.g]$,  for any $p$ and any odd $q,r$.
    
    \item $X^pY^qZ^r\in\.g$, for any odd $p$ and any even $q,r$, when $n$ is even.

    \item $X^1Y^qZ^r\in\.g$, for any even $q,r$, when $n$ is odd.
\end{enumerate}
\end{restatable}


\begin{restatable}[linear independence]{fact}{LinearIndependence}\label{fact:linear-independent}
{Let  $L_{X^1}\defeq \{[X^1, X^p Y^q Z^r] : q,r \text{ odd}\}$, $L_{Z^2}\defeq \{[Z^2, X^p Y^1 Z^r] : r \text{ odd}\}$. Then, the vectors in $L\defeq L_{X^1}\cup L_{Z^2}$ are linearly independent, and $\abs{L} = \abs{L_{X^1}} + \abs{L_{Z^2}}$. }
\end{restatable}

Using the above facts, we now construct basis for $\.g_{K_n}$ when $n$ is even. The odd $n$ case is similar and can be found in the Supplementary Information.  
In the following, let $ad_A(S)\defeq\{[A, t]: ~ t\in S\}$ for any set $S$.


Define the following sets.
\begin{align*}
    & J_1\defeq\{X^{p}Y^{q}Z^{r}: p \text{ is even, $q$ and $r$ are odd}\}\backslash\{Y^1Z^1\}, \\ 
    & J_2\defeq\{X^{p}Y^{1}Z^{r}: \text{$p$ and $r$ are odd} \}, \\
    & J_3\defeq \{X^{p}Y^{q}Z^{r}: \text{$p$, $q$ and $r$ are all odd and } q\ge 3\}, \\
    & J_4\defeq\{X^{p}Y^{q}Z^{r}: p \text{ is odd, $q$ and $r$ are even}\}, \\
     & J_5\defeq\{Z^2,Y^2,Y^1Z^1\},\\
    & L_1\defeq\{[X^1, v]: v\in J_1\}, \\
    & L_2\defeq\{[Z^2,v]: v\in J_2\}. 
\end{align*}
Note that (i) $J_2 \cup J_3$ contains those $X^{p}Y^{q}Z^{r}$ with odd $p, q, r$ (``ooo''); (ii) $J_1 \cup J_5$ is a superset of those $X^{p}Y^{q}Z^{r}$ with even $q$ and odd $q$ and $r$ (``eoo''); and (iii) vectors in $L_1$ and $L_2$ all have even $p,q,r$ (``eee''). These are indicated in the second row of Table \ref{table:ad-on-basis-even-with-explanation}.

\begin{restatable}{thm}{KnBasisEven}\label{thm:Kn-basis-even}
When $n$ is even,
    \begin{equation*}
        R\defeq\bigcup_{k=1}^5 J_k \cup L_1 \cup L_2  
    \end{equation*}
    is a basis for $\.g$, and $\dim(\.g) = \frac{1}{12}\lp n^3 + 6n^2 + 2n + 12\rp$.
\end{restatable}

\begin{proof}
Let $\.r\defeq \spn_{\mathbb{R}} R$. In the remainder of this section we will show that  (1) the basis vectors in $R$ are linearly independent, (2) $|R| = \frac{n^3+6n^2+2n+12}{12}$,(3) $\.g\subseteq \.r$, (4) $\.r\subseteq \.g$. 
\end{proof}
\paragraph{(1) The basis vectors in $R$ are linearly independent.} 
\begin{itemize}
    \item For any $i\in[5]$, the Pauli  strings in $J_i$ are mutually orthogonal, as distinct elements in $J_i$ have distinct values of $(p,q,r)$. For the same reason, all Pauli strings in $J_i$ are orthogonal to those in $J_j$, for $i\neq j$. The vectors in $\bigcup_{i=1}^5 J_i$ are thus linearly independent, {and $\abs{J} = \sum_{i=1}^5\abs{J_i}$}.
    
    \item {$L_1\subset L_{X^1}$, and $L_2 \subset L_{Z^2}$, where $L_{X^1},L_{Z^2}$ are defined in Fact \ref{fact:linear-independent}. By Fact \ref{fact:linear-independent}, the vectors in $L_1\cup L_2$ are linearly independent, and $\abs{L_1\cup L_2} = \abs{L_1} + \abs{L_2}$}. 
    \item Vectors in $L_1\cup L_2$ have (even, even, even) parity for $(p,q,r)$, and thus are linearly independent from vectors in $\bigcup_{j=1}^4 J_j \cup \{Y^1Z^1\}$. Finally, it can be seen from Eqs.~\eqref{eq:Kn-adA} and \eqref{eq:Kn-adB} that the Pauli string orbits $Y^2$ and $Z^2$ do not appear as any component term in $L_1\cup L_2$.
\end{itemize}

\paragraph{(2) $|R| = \frac{n^3+6n^2+2n+12}{12}$:} {The above proof of linear independence also implies that  $\abs{R} = \sum_{i=1}^5\abs{J_i} + \sum_{i=1}^2\abs{L_i}$.} The results follows from the fact that
\begin{align*}
    &|J_1|=|L_1|=\sum_{t=1}^{n/2} t(n/2-t+1)-1=\frac{n^3+6n^2+8n-48}{48},\\
    & |J_2|=|L_2|= \sum_{t=1}^{n/2-1} (n/2-t)=\frac{n^2-2n}{8},\\
    &|J_3|=\sum_{t=2}^{n/2-1}(t-1)(n/2-t)=\frac{n^3-6n^2+8n}{48},\\
    & |J_4|=\frac{n}{2}+\sum_{t=1}^{n/2-1}(t+1)(n/2-t)=\frac{n^3+6n^2+8n}{48},\\
    & |J_5|=3.
\end{align*}

\paragraph{(3) $\.g\subseteq \.r$.} As $X^1\in J_4$ and $Z^2 \in J_5$, thus $X^1,Z^2\in \.r$. It suffices to show that $\.r$ is closed under $ad_{X^1}$ and $ad_{Z^2}$. This can be verified using Table \ref{table:ad-on-basis-even-with-explanation}, where the terms in parentheses indicate the proof methods we used, summarized as follows:
\begin{table}[]
    \centering
     \caption{Adjoint map on basis vectors of $\.r$ ($n$ even). The second row gives the signature of $(p,q,r)$ for the elements $X^p Y^q Z^r$ in the sets. Parentheses indicate proof method (see main text).}
    \label{table:ad-on-basis-even-with-explanation}\resizebox{\textwidth}{!}{
    \begin{tabular}{c|c c c c c | c c}
         &  $J_1$ & $J_2$ & $J_3$ &$J_4$ & $J_5$ & $L_1$ & $L_2$ \\
        \hline 
        \multirow{2}*{\color{blue}{ $(p,q,r)$}} & \multirow{2}*{\color{blue}{$eoo - 011$}} & \multirow{2}*{\color{blue}{$o1o$}} & \multirow{2}*{\color{blue}{$ooo - o1o$}} & \multirow{2}*{\color{blue}{$oee$}} & \multirow{2}*{\color{blue}{$020, 002, 011$}}& \color{blue}{$[X^1, \text{$eoo$}]$} & \color{blue}{$[Z^2, \text{$o1o$}]$}  \\
        & & & & & & \color{blue}{$\subseteq \text{$eee$}$} & \color{blue}{$\subseteq \text{$e0e$}\cup \text{$e2e$} \subseteq \text{$eee$}$} \\
        \hline
        \multirow{2}*{$ad_{X^1}$}  & $L_1$ & $J_4$ & $J_4$  & $J_2\cup J_3$  & $J_5$ & $J_1\cup J_5$ & $J_1\cup J_5$\\
        & \color{blue}{(by def)} & \color{blue}{(parity)} & \color{blue}{(parity)} & \color{blue}{(parity)} & \color{blue}{(computation)} & \color{blue}{(parity)} & \color{blue}{(parity)} \\
        \hline
        \multirow{2}*{$ad_{Z^2}$}  & $J_4$ & $L_2$ & $L_1\cup L_2$ & $J_1\cup J_5$ & $J_2\cup J_4$ & $J_2\cup J_3$ & $J_2\cup J_3$ \\
        & \color{blue}{(parity)} & \color{blue}{(by def)} & \color{blue}{(complicated)} & \color{blue}{(parity)} & \color{blue}{(computation)} & \color{blue}{(parity)} & \color{blue}{(parity)} 
    \end{tabular}}
\end{table}
\begin{itemize}
    \item ``by def'': $ad_{X^1}(J_1) \subseteq L_1$ and $ad_{Z^2}(J_2) \subseteq L_2$ hold by definition of $L_1$ and $L_2$. 
    
    \item ``parity'': The result follows from how the adjoint map changes the parities. For instance, $ad_{X^1}(J_4) \subseteq \spn_{\mathbb R} J_2 \cup J_3$ because the parities of $(p,q,r)$ for $J_4$ is $(odd, even, even)$ (as indicated by ``$oee$'' under $J_4$ in the second row), and the $ad_{X^1}$ map changes the parities of $Y$ and $Z$, which makes the new parity $(odd, odd, odd)$, thus in $\spn_{\mbR}J_2 \cup J_3$. Other ``parity'' entries can be easily proved in a similar way. 
    
    \item ``computation'': $J_5$ has three specific orbit-sums, and the adjoint actions on them can be verified by direct computation: $ad_{X^1}(Z^2)=2Y^1Z^1$, $ad_{X^1}(Y^1Z^1)=4Y^2 -4Z^2$ and $ad_{X^1}(Y^2)=-2Y^1Z^1$ are all in  $\spn_{\mathbb R}  J_5$; $ad_{Z^2}(Z^2)=0$; $ad_{Z^2}(Y^1Z^1)= 2(n-1)X^1+4X^1Z^2 \in  \spn_{\mathbb R} J_4$ and $ad_{Z^2}(Y^2)= 2X^1Y^1Z^1\in  \spn_{\mathbb R}  J_2$;
    
    \item ``complicated'': $ad_{Z^2}(J_3)$ is the only complicated case. We now show $ad_{Z^2}(J_3)$ to be in $\spn_{\mathbb R} (L_1\cup L_2)$.
    Let $p,q,r$ be odd, $q\ge 3$. From Lemma~\ref{lem:adzz-pqr}
    \begin{multline}\label{eq:Kn-adB-induction}
          ad_{Z^2}(X^pY^qZ^{r})=  \beta_1\underbrace{ad_{X^1}(X^{p+1}Y^{q-2}Z^{r})}_{\in L_1}+\beta_2\underbrace{ad_{X^1}(X^{p+1}Y^{q-2}Z^{r+2})}_{\in L_1} \\
        -\beta_3\underbrace{ad_{X^1}(X^{p-1}Y^{q}Z^{r})}_{\in L_1}-\beta_4\underbrace{ad_{X^1}(X^{p-1}Y^{q}Z^{r+2})}_{\in L_1}+\beta_5 ad_{Z^2}(X^{p}Y^{q-2}Z^{r+2}),
    \end{multline}
for some $\beta_1, \ldots \beta_5\in\mb{R}$.
The first 4 terms in Eq.~\eqref{eq:Kn-adB-induction} are all in $L_1$, and the last term has $Y$'s exponent decreased from $q$ to $q-2$. Apply Eq.~\eqref{eq:Kn-adB-induction} again to this last term, and continue this process until $Y$'s exponent is decreased to 1, yielding an expression for $ad_{Z^2}(X^pY^qZ^{r})$ as $4(q-1)/2 = 2(q-1)$ terms in $L_1$, plus a last term of form $ad_{Z^2}(X^{p}Y^{1}Z^{q+r-1})$. Since $p$ and $q+r-1$ are both odd, $ad_{Z^2}(X^{p}Y^{1}Z^{q+r-1})\in L_2$.
\end{itemize}

\paragraph{(4) $\.r\subseteq \.g$.}
We will show that  $J_k\subseteq \.g$ for all $k\in [5]$. $L_1 = ad_A(J_1)\subseteq \.g, L_2=ad_B(J_2)\subseteq \.g$ then follows immediately.

\begin{itemize}
    \item By definition,
    \eq{
    J_1\cup \{Y^1 Z^1\} &= \{X^p Y^q Z^r: p \text{ is even, } q \text{ and }r \text{ are odd}\},\\
    J_2\cup J_3 &=\{X^p Y^q Z^r: p,q,r \text{ are odd}\},
    }
    and, by Fact~\ref{fact:anyoddodd-oddeveneven}, $X^p Y^q Z^r\in \.g$ for all $p$, all odd $q$ and odd $r$.  Thus, $J_1\cup \{Y^1Z^1\}\subseteq \.g$ and $J_2\cup J_3\subseteq \.g$.
    
    \item By Fact~\ref{fact:anyoddodd-oddeveneven}, $X^p Y^q Z^r\in \.g$ for all odd $p$ and all even $q, r$. Thus, $J_4\subseteq \.g$.

    \item Fact \ref{fact:Kn-01o-11o-10e} gives $Y^1 Z^1\in \.g$ and therefore $ad_{X^1}(Y^1 Z^1)=Y^2 - Z^2\in \.g$. Since $Z^2\in\.g$, $Y^2$ is in $\.g$ as well. Thus, $J_5\subseteq\.g$.
\end{itemize}

\label{sec:DLA-complete-graph}

\bibliographystyle{quantum}
\bibliography{dla}

\newpage

\appendix
{\centering\section*{Supplementary Information for ``On the dynamical Lie algebras of quantum approximate
optimization algorithms''}}
Additional details and proofs of our results stated in the main text can be found here.
\appendix

\section{General DLAs}

In the Methods section, a procedure for computing a basis for a DLA was informally described. A formal method for doing so is given in Algorithm~\ref{alg:DLA-generation} below.

\begin{algorithm}[!htb]
    \caption{Generate DLA}
    \label{alg:DLA-generation}
    \textbf{Input}: a set of elements $\+A = \{H_1, \ldots, H_s\}$ in a finite-dimensional Lie algebra $\.g$ over field $\mbF$.\\
    \textbf{Output}: $(\+B,k)$, where $\+B$ is a basis of $\langle \+A \rangle_{\text{Lie},\mbF}$, i.e. $\langle \+A \rangle_{\text{Lie},\mbF} = \spn_{\mbF}\{H: H\in \+B\}$, and $k=\operatorname{Deg}(\langle \+A\rangle_{\text{Lie},\mb{F}})$.\\  
    \begin{algorithmic}[1]
        \STATE Let $\+B_0$ be a maximal subset of linearly independent elements from  $\+H$
        \STATE $V_0 := \spn_{\mbF} \{H: H\in \+B_0\}$
        \FOR{ $k = 1, 2, \ldots$ }
        \STATE $\+B_k := \+B_{k-1}$, $V_k:= V_{k-1}$
            \FOR{$H^{(0)}\in \+B_{0}$} 
                \FOR{ $H^{(k-1)}\in \+B_{k-1} \backslash \+B_{k-2}$}
                \STATE compute $H = [H^{(0)}, H^{(k-1)}]$
                \IF{ $H\notin V_k$}
                    \STATE  let $\+B_k := \+B_{k-1}\cup \{H\}$ and $V_k:= \spn_{\mbF} \{V_k,H\}$
                \ENDIF
                \ENDFOR
            \ENDFOR
            \IF{$\+B_k = \+B_{k-1}$}
                \STATE  return $(\+B_k,k-1)$
            \ENDIF
        \ENDFOR
    \end{algorithmic}
\end{algorithm}

In the remainder of this section we give proofs of Lemmas~\ref{lem:pi-invariance} and \ref{thm:DLA-parity}, Proposition~\ref{lem:X_I_commute}, and Facts~\ref{fact:adH-diag}-\ref{fact:dim-cartan-greater-three}. Recall that for any Pauli string $P = P_1 P_2 \ldots P_n\in\+P_n$, 
and any permutation $\pi\in S_n$, \[\pi(P) = P_{\pi(1)} P_{\pi(2)} \ldots P_{\pi(n)}\] 
defines an action $S_n$ on the set $\mcP_n$ of all $n$-qubit Pauli strings. 
By linearity this can be extended to the whole of $\.s\.l(2^n,\mb{C})$:
\eql{\label{eq:aut(G)-action}
    \pi\left(\sum_P \alpha_P P\right) = \sum_P \alpha_P P_{\pi(1)} P_{\pi(2)} \ldots P_{\pi(n)}, \quad \forall \alpha_P\in \mbC. 
}

\LemPiInvar*

\begin{proof}Let $A,B\in \+A$, and consider their Pauli string decompositions $A = i\sum_P \alpha_P P_1P_2\ldots P_n$ and $B = i\sum_Q \beta_Q Q_1Q_2\ldots Q_n$. Then, 
    \begin{align*}
        &\pi([A,B]) \\ =& \pi\left(\left[i\sum_P \alpha_P P_1P_2\ldots P_n, i\sum_Q \beta_Q Q_1Q_2\ldots Q_n\right]\right) \\
         = & \pi\left(i^2\sum_P \alpha_P P_1P_2\ldots P_n \sum_Q \beta_Q Q_1Q_2\ldots Q_n - i^2\sum_Q \beta_Q Q_1Q_2\ldots Q_n \sum_P \alpha_P P_1P_2\ldots P_n \right) \\
         = & \pi\left(-\sum_{P,Q} \alpha_P \beta_Q \otimes_{j=1}^n (P_j Q_j) + \sum_{P,Q} \alpha_P \beta_Q\otimes_{j=1}^n (Q_j P_j) \right)  \\
         = & -\sum_{P,Q} \alpha_P \beta_Q \otimes_{j=1}^n (P_{\pi(j)} Q_{\pi(j)}) + \sum_{P,Q} \alpha_P \beta_Q \otimes_{j=1}^n (Q_{\pi(j)} P_{\pi(j)}) \\
         = & -\sum_P \alpha_P P_{\pi(1)}P_{\pi(2)}\ldots P_{\pi(n)} \sum_Q \beta_Q Q_{\pi(1)}Q_{\pi(2)}\ldots Q_{\pi(n)} \\ 
        & \quad + \sum_Q \beta_Q Q_{\pi(1)}Q_{\pi(2)}\ldots Q_{\pi(n)} \sum_P \alpha_P P_{\pi(1)}P_{\pi(2)}\ldots P_{\pi(n)} \\
         = & \pi(A)\pi(B) - \pi(B)\pi(A) \\
         = & AB-BA = [A,B].
    \end{align*}
Thus, $\pi$-invariance is preserved by the commutator. The result follows from linearity (Eq.~\eqref{eq:aut(G)-action}).
\end{proof}

We say that a Pauli string $P = P_1 P_2 \cdots P_n$ is of \emph{type $(n_X,n_Y,n_Z)$} if the numbers of Pauli $X$, $Y$, and $Z$ in the string are $n_X$, $n_Y$, and $n_Z$, respectively. A Pauli string will be called $YZ$-even if $n_Y + n_Z$ is even. 


\DLAParity*

\begin{proof}
   Let $\+X = \{X_1X_2\ldots X_n, I_1I_2\ldots I_n\}$. We prove this by induction on the generation steps. 
    The base cases are the DLA generators (ignoring the prefactor of $i$), $\sum_j X_j$ and $\sum_{(j,k)\in E} Z_j Z_k$, which have $(n_Y,n_Z) = (0,0)$ and $(0,2)$, respectively, thus satisfying the claimed parity. 
    Now consider any $YZ$-even Pauli string $P = P_1 P_2 \cdots P_n\notin \+X$, and we will show 
    that $[\sum_j X_j, P]$ and $[\sum_{(j,k)\in E} Z_j Z_k, P]$ are also in the span of Pauli strings satisfying these conditions. To that end, it suffices to consider each individual $[X_j, P]$ and $[Z_j Z_k, P]$. Without loss of generality, assume $j=1$ and $k=2$. From the $\.{su}(2)$ commutation relations
    \eq{
    [X_1,X_1]&=0,\quad [X_1, Y_1]= 2iZ_1, \quad [X_1, Z_1]=-2iY_1,
    }
    and the fact that $[X_1,I_1]=0$, it follows that (i) $[X_1, P] = [X_1,P_1]P_2\ldots P_n$ either has $n_Y+n_Z$ equal to that of $P$, or else  $[X_1, P] = 0$; (ii) $P\notin\+X\Ra[X_1,P]\notin\+X$.

    Similarly, the values of $[P_1P_2,Z_1Z_2]$ are given in the following table.
    \begin{center}
    \resizebox{\textwidth}{!}{\begin{tabular}{c|c|c|c|c|c|c|c|c|c|c|c|c|c|c|c|c}
        $[\cdot,\cdot]$ & $I_1I_2$ & $I_1X_2$ & $I_1Y_2$ & $I_1Z_2$ & $X_1I_2$ & $X_1X_2$ & $X_1Y_2$ &$X_1Z_2$ & $Y_1I_2$ & $Y_1X_2$ & $Y_1Y_2$ &$Y_1Z_2$ & $Z_1I_2$ & $Z_1X_2$ & $Z_1Y_2$ &$Z_1Z_2$ \\
        \hline
        $Z_1Z_2$ & $0$ & $2iZ_1Y_2$ & $-2iZ_1X_2$ & $0$ & $2iY_1Z_2$ & $0$ & $0$ & $2iY_1$ & $-2iX_1Z_2$ & $0$ & $0$ & $-2iX_1$ & $0$ & $2iY_2$ & $-2iX_2$ & $0$\\
    \end{tabular}}
\end{center}
    Again, it follows that (i) $[Z_1Z_2,P] = [Z_1Z_2, P_1P_2]P_3\ldots P_n$ either has $n_Y+n_Z$ of $[Z_1Z_2,P]$ equal to that of $P$ plus $a$, with $a \in \{0, 2, -2\}$ or else has  $[Z_1Z_2, P] = 0$; (ii) $P\notin\+X\Ra[Z_1Z_2,P]\notin \+X$. 
\end{proof}

\XICommute*

\begin{proof}
    Recall that $\.g_G$ has the generating set $\+A=\{i\sum_{j=1}^n X_j,i\sum_{(s,t)\in E}Z_sZ_t\}$. From the Jacobi identity it is straightforward to see that $[P,H]=0$ for all $H\in\.g_G$ if and only if $[P,A]=0$ for all $A\in\+A$.

    Let $P_{-j}=\bigotimes_{k\in [n]-\{j\}}P_k$, then $[P,i\sum_{j=1}^n X_j]=i\sum_{j=1}^n [P_j,X_j]\otimes P_{-j}$. For the Pauli string $[P_j,X_j]\otimes P_{-j}$, the $j$-th position is $[P_j,X_j]$, which is either 0, or some Pauli operator different from $P_j$. But for the Pauli string $[P_k,X_k]\otimes P_{-k}$, the $j$-th position is $P_j$ if $j\neq k$. Therefore, to ensure $[P,i\sum_{j=1}^n X_j]=0$, we need $[P_j,X_j]=0$ for all $1\le j \le n$. Therefore, $P_j=X_j$ or $P_j=I_j$ for all $j$. 
    
    Let $P_{-\{s,t\}}=\bigotimes_{k\in [n]-\{s,t\}}P_k$, then $[P,i\sum_{(s,t)\in E} Z_sZ_t]=i\sum_{j=1}^n [P_sP_t,Z_sZ_t]\otimes P_{-\{s,t\}}$. 
    For the Pauli string $[P_sP_t,Z_sZ_t]\otimes P_{-\{s,t\}}$, the $s$-th and $t$-th positions are $[P_sP_t,Z_sZ_t]$. 
    But for the Pauli string $[P_kP_\ell,Z_kZ_\ell]\otimes P_{-\{k,\ell\}}$, the $s$-th position is $P_s$ or the $t$-th position is $P_t$ if $(s,t)\neq(k,\ell)$.
     Therefore, 
    to ensure $[P,i\sum_{(s,t)\in E} Z_sZ_t]=0$, we need $[P_sP_t,Z_sZ_t]=0$ for all $(s,t)\in E$. Since we have shown that $P_s,P_t\in\{I,X\}$, we know that $P_sP_t=I_sI_t$ or $P_sP_t=X_sX_t$, i.e., $P_s=P_t$ for any $(s,t)\in E$. As $G$ is connected the result follows. 
\end{proof}

\AdHDiag*
\begin{proof} $H$ has the form $H=\sum_j c_j M_j$, where $c_j\in\mb{C}$ and $M_j\in\.{su}(2^n)$. Each $M_j$ is diagonalizable, and diagonalizable operators that mutually commute are simultaneously diagonalizable. Therefore $H$ is diagonalizable. The result follows from that fact that, for $\.g$  a matrix Lie algebra over $\mb{C}$ or $\mb{R}$, if $H\in\.g$ is diagonalizable then $ad_H:\.g\ra\.g$ is diagonalizable (see e.g.,\cite{erdmann2006introduction}). 
\end{proof}

\CartanSameDim*

\begin{proof} This follows from the fact that, for any two Cartan subalgebras $\.h_1, \.h_2$ of a complex semisimple Lie algebra $\.k$, 
there exists a bijective linear map $\phi:\.k\ra\.k$ satisfying $[\phi(A), \phi(B)] =\phi([A,B])$ and $\phi(\.h_1) = \.h_2$ (see e.g,~\cite{knapp1996lie}).
    
\end{proof}

\CartanSemisimple*

\begin{proof}
    Elements in the same $\.h_i$ commute as $\.h_i$ is a Cartan subalgebra. Elements in different $\.h_i$ and $\.h_j$ commute because of the Lie algebra direct sum definition. This $\.h$ is also maximal Abelian because any other element $x_1+\cdots + x_k$ with $x_i \in \.g_i$ cannot commute with any $\.h_i$ with $x_i\ne 0$---otherwise $\.h_i$ is not maximal and thus not Cartan subalgebra in $\.g_i$. By Fact~\ref{fact:adH-diag}, $ad_H$ is diagonalizable. 
\end{proof}

\DimCartanGreaterThree*

\begin{proof} Follows directly from Theorem~\ref{thm:classification}.
\end{proof}


\section{Cycle Graphs - Verification of Commutation Relations}

Here we prove Theorem \ref{thm:Cn-isomorphism}, restated below, by showing that the $\{U_k, V_k, H_k\}$ vectors satisfy the $\.{sl}(2,\mb{C}) $ commutation relations.  Throughout this section, we will use $\.g$ to denote $\.g_{C_n}$ .

\ExplicitIso*

For ease of later computations, we first extend the definition of Pauli orbit-sums. Recall that the orbit-sums $YX^k Y$, $ZX^k Z$ and $YX^k Z$ (see Table \ref{tab:Cn-adA-adB-on-orbits}) are only defined for $k=0, 1, \ldots, n-2$, with the exceptions of $YX^{-1}Y \defeq -X$ and $ZX^{n-1}Z \defeq X^{n-1}$. Now we extend the definition to all integers $k$.
\begin{defi}
    Define the extended orbits $YX^k Y$, $ZX^k Z$ and $YX^k Z$ as in Table \ref{table:Cn-def-generalized-orbit}.
 
    \begin{table}[]
    \centering 
    \caption{Extended Pauli orbit-sums.}
    \label{table:Cn-def-generalized-orbit}
    \begin{tabular}{c|c c c c c}
      Range of integer $k$    & $(-\infty,-1]$ & $-1$ & $[0, n-2]$ & $n-1$ & $[n-1,+\infty)$ \\ \hline
      $YX^kY$  & $ZX^{-k-2}Z$ & $-X$ & $YX^kY$ & $X^{n-1}$ & $ZX^{2n-k-2}Z$\\
      $ZX^kZ$   &  $YX^{-k-2}Y$ & $-X$ & $ZX^kZ$ & $X^{n-1}$ & $YX^{2n-k-2}Y$\\
      $YX^kZ$  & $-YX^{-k-2}Z$ & $0$ & $YX^kZ$ & $0$ & $-YX^{2n-k-2}Z$\\
    \end{tabular}
    
    \end{table}
\end{defi}
By definition, these extended orbit-sums have a period of $2n$: $YX^kY = YX^{2n+k}Y$, $ZX^kZ = ZX^{2n+k}Z$, and $YX^kZ = YX^{2n+k}Z$ for any $k\in \mbZ$.
By direct calculation, the commutator relations of these orbits can be computed as in Table \ref{table:Cn-generalized-orbit-commutator}, which generalizes Table \ref{tab:Cn-adA-adB-on-orbits}.

\begin{table}[] 
    \centering 
    \caption{Commutators of extended Pauli orbit-sums.}
    \label{table:Cn-generalized-orbit-commutator}
     \resizebox{\textwidth}{!}{\begin{tabular}{c | c c c c}
        $[\cdot, \cdot]$  & $YX^tY$ & $ZX^tZ$ & $YX^tZ$ \\
        \hline
        $YX^sY$  & $-2YX^{s-t-1}Z$ & $-2YX^{s+t+1}Z$ & $4(ZX^{t-s-1}Z-YX^{t+s+1}Y)$ \\
        $ZX^sZ$ & $2YX^{s+t+1}Z$ & $2YX^{s-t-1}Z$ & $4(ZX^{t+s+1}Z-YX^{t-s-1}Y)$ \\
        $YX^sZ$ & $4(YX^{t+s+1}Y-ZX^{s-t-1}Z)$ & $4(YX^{s-t-1}Y-ZX^{t+s+1}Z)$ & $0$
    \end{tabular}}
    
\end{table}

We will also use the following elementary facts.

\begin{fact}\label{fact:expsin}
        For any $\alpha,\beta,\alpha', \beta'\in \mbR$,
        \begin{equation}\label{eq:expsin}
            e^{i\alpha}\sin \beta+e^{i\alpha'}\sin \beta'=-ie^{i\frac{\alpha+\beta+\alpha'+\beta'}{2}}\cos\frac{\alpha+\beta-\alpha'-\beta'}{2}+ie^{i\frac{\alpha-\beta+\alpha'-\beta'}{2}}\cos\frac{\alpha-\beta-\alpha'+\beta'}{2}.
        \end{equation}
    \end{fact}

    \begin{fact}\label{lem:sumodd}
        Let $\theta=k\pi/2n$ for some integers $1\le |k| \le 2n-1$ and $n\ge 1$. Then 
        \begin{equation*}
            \sum_{j=0}^{n-1}\cos(2j+1)\theta=0.
        \end{equation*}
    \end{fact}
We are now in a position to prove Theorem~\ref{thm:Cn-isomorphism}.

\begin{proof}[Proof of Theorem~\ref{thm:Cn-isomorphism}]
The canonical basis in Theorem~\ref{thm:Cn-isomorphism} has 9 commutator relations to verify (Eq. \eqref{eq:sl2C-canonical-basis}): three within each $\{U_k, V_k, H_k\}$  and six across different sets $\{U_k, V_k, H_k\}$ and $\{U_\ell, V_\ell, H_\ell\}$ (namely $[U_k, U_\ell]$, $[U_k, V_\ell]$, $[H_k, U_\ell]$, $[V_k, V_\ell]$, $[H_k, V_\ell]$, and $[H_k, H_\ell]$ all being 0). We can prove these by directly computing the commutators, which involves elementary but tedious manipulations of trigonometric identities. 

\noindent{\bf (1) $[H_k,U_\ell]=0$ if $k\neq \ell$ and $[H_k,U_k]=2U_k$.} Plugging the definitions of $H_k$ and $U_\ell$ into the commutator $[H_k, U_\ell]$ and using Table \ref{table:Cn-generalized-orbit-commutator}, we have 
\begin{multline*}
    [H_k,U_\ell] = \frac{i}{2n^2} \sum_{q=0}^{n-1} \sum_{p=1}^{n-1}\big(e^{-i\frac{\ell (q+1)\pi}{n}}\sin \frac{kp\pi}{n} (YX^{p+q-1}Y-ZX^{p-q-1}Z)\\ + e^{i\frac{\ell q\pi}{n}}\sin \frac{kp\pi}{n} (YX^{p-q-2}Y-ZX^{p+q}Z)\big).
\end{multline*}
According to the definitions in Table \ref{table:Cn-def-generalized-orbit}, $YX^{k}Y\defeq ZX^{2n-k-2}Z$ and $ZX^{k}Z\defeq YX^{2n-k-2}Y$ if $k\ge n-1$ and $YX^{k}Y\defeq ZX^{-k-2}Z$ and $ZX^{k}Z\defeq YX^{-k-2}Y$ if $k\le -1$.
Therefore, we transform (1)  $YX^{p+q-1}Y\to ZX^{2n-p-q-1}Z$ and $ZX^{p+q}Z  \to YX^{2n-p-q-2}Y$  if $n-q\le p \le n-1$; (2) $ZX^{p-q-1}Z \to YX^{-p+q-1}Y$ and $YX^{p-q-2}Y \to ZX^{-p+q}Z$ if $1\le p \le q$. In this way, $[H_k,U_\ell]$ can be written as follows, where the four lines correspond to the four terms in the RHS of the above equality of $[H_k,U_\ell]$. (We also slightly adjust the range of the $p$ to make the inner summation over $q$ in a legal range.)
\begin{align*}
   &[H_k,U_\ell] 
   \\ = &\frac{i}{2n^2}\Big(\sum_{q=0}^{n-2}\sum_{p=1}^{n-q-1} e^{-i\frac{\ell(q+1)\pi}{n}}\sin \frac{kp\pi}{n}YX^{p+q-1}Y+\sum_{q=1}^{n-1}\sum_{p=n-q}^{n-1} e^{-i\frac{\ell(q+1)\pi}{n}}\sin \frac{kp\pi}{n}ZX^{2n-p-q-1}Z\\
        &- \sum_{q=1}^{n-1}\sum_{p=1}^{q} e^{-i\frac{\ell(q+1)\pi}{n}}\sin\frac{kp\pi}{n}YX^{-p+q-1}Y- \sum_{q=0}^{n-2}\sum_{p=q+1}^{n-1}e^{-i\frac{\ell(q+1)\pi}{n}}\sin\frac{kp\pi}{n}ZX^{p-q-1}Z\\
        &+ \sum_{q=1}^{n-1}\sum_{p=1}^{q} e^{i\frac{\ell q \pi}{n}} \sin \frac{kp\pi}{n} ZX^{-p+q}Z+ \sum_{q=0}^{n-2}\sum_{p=q+1}^{n-1} e^{i\frac{\ell q \pi}{n}} \sin \frac{kp\pi}{n} YX^{p-q-2}Y\\
        &-\sum_{q=0}^{n-2}\sum_{p=1}^{n-q-1} e^{i\frac{\ell q \pi}{n}} \sin \frac{kp\pi}{n} ZX^{p+q}Z-\sum_{q=1}^{n-1}\sum_{p=n-q}^{n-1} e^{i\frac{\ell q \pi}{n}} \sin \frac{kp\pi}{n} YX^{2n-p-q-2}Y\Big). 
\end{align*}

Now we change the inner variable $p$ to $j$ to transform all $YX^{\cdots} Y$ terms to $YX^{j-1}Y$, and all $ZX^{\cdots} Z$ terms to $ZX^{j}Z$. 
Then we can obtain that
\begin{align*}
  [H_k,U_\ell]
   = &\frac{i}{2n^2}\Big(\sum_{q=0}^{n-2}\sum_{j=q+1}^{n-1} e^{-i\frac{\ell(q+1)\pi}{n}}\sin \frac{k(j-q)\pi}{n}YX^{j-1}Y \\&+\sum_{q=1}^{n-1}\sum_{j=n-q}^{n-1} e^{-i\frac{\ell(q+1)\pi}{n}}\sin \frac{k(2n-q-j-1)\pi}{n}ZX^{j}Z\\
        &- \sum_{q=1}^{n-1}\sum_{j=0}^{q-1} e^{-i\frac{\ell(q+1)\pi}{n}}\sin\frac{k(q-j)\pi}{n}YX^{j-1}Y \\& - \sum_{q=0}^{n-2}\sum_{j=0}^{n-q-2}e^{-i\frac{\ell(q+1)\pi}{n}}\sin\frac{k(j+q+1)\pi}{n}ZX^{j}Z\\
        &+ \sum_{q=1}^{n-1}\sum_{j=0}^{q-1} e^{i\frac{\ell q \pi}{n}} \sin \frac{k(q-j)\pi}{n}ZX^{j}Z+ \sum_{q=0}^{n-2}\sum_{j=0}^{n-q-2} e^{i\frac{\ell q \pi}{n}} \sin \frac{k(q+j+1)\pi}{n} YX^{j-1}Y\\
        &-\sum_{q=0}^{n-2}\sum_{j=q+1}^{n-1} e^{i\frac{\ell q \pi}{n}} \sin \frac{k(j-q)\pi}{n} ZX^{j}Z\\ &-\sum_{q=1}^{n-1}\sum_{j=n-q}^{n-1} e^{i\frac{\ell q \pi}{n}} \sin \frac{k(2n-q-j-1)\pi}{n} YX^{j-1}Y\Big).
\end{align*}
We add the following four terms which are all equal to $0$ to the RHS of the above equation:
\begin{align*}
   & e^{-i\frac{\ell(q+1)\pi}{n}}\sin \frac{k(q-q)\pi}{n}YX^{q-1}Y, &e^{-i\frac{\ell(q+1)\pi}{n}}\sin \frac{k(-(n-q-1)-q-1)\pi}{n}ZX^{n-q-1}Z, & \ \\ 
   & e^{i\frac{\ell q\pi}{n}}\sin \frac{k(q-q)\pi}{n}ZX^qZ, &e^{i\frac{\ell q \pi}{n}}\sin \frac{k((n-q-1)+q+1)}{n}YX^{n-q-2}Y. & \ 
\end{align*}
 Then we have
\begin{multline*}
     [H_k,U_\ell] =  \frac{i}{2n^2}\sum_{j=0}^{n-1}\sum_{q=0}^{n-1} \Big(\big(e^{-i\frac{\ell(q+1)\pi}{n}}\sin \frac{k(j-q)\pi}{n}+e^{i\frac{\ell q\pi}{n}}\sin \frac{k(j+q+1)\pi}{n}\big)YX^{j-1}Y 
        \\
         +\big(e^{-i\frac{\ell (q+1)\pi}{n}}\sin \frac{k(-j-q-1)\pi}{n}+e^{i\frac{\ell q\pi}{n}}\sin \frac{k(q-j)\pi}{n}\big)ZX^{j}Z\Big).
\end{multline*}
Using Fact \ref{fact:expsin}, we can obtain that
\begin{multline*}
    [H_k,U_\ell] =\\ \frac{i}{2n^2} \sum_{j=0}^{n-1}\sum_{q=0}^{n-1}\Big(\big( -ie^{i\frac{(k(2j+1)-\ell)\pi}{2n}}\cos\frac{(k+\ell)(2q+1)\pi}{2n}+ie^{-i\frac{(k(2j+1)+\ell)\pi}{2n}} \cos \frac{(k-\ell)(2q+1)\pi}{2n}\big) YX^{j-1}Y\\
         +\big( -ie^{-i\frac{(k(2j+1)+\ell)\pi}{2n}}\cos \frac{-(k+\ell)(2q+1)\pi}{2n}+ie^{i\frac{(k(2j+1)-\ell)\pi}{2n}} \cos\frac{(k-\ell)(2q+1)\pi}{2n}\big) ZX^jZ\Big).
\end{multline*}

By Fact \ref{lem:sumodd}, we know that if $k\ne \ell$, then $\sum_{q=0}^{n-1} \cos \frac{(k\pm \ell) (2q+1)\pi }{2n} = 0$ and thus $[H_k,U_\ell] = 0$. If $k = \ell$, then $\sum_{q=0}^{n-1} \cos \frac{(k+ \ell) (2q+1)\pi }{2n} = 0$ and \[[H_k,U_k] = -\frac{1}{2n} \sum_{j=0}^{n-1} \lp e^{-ik(j+1)\pi/n} YX^{j-1} Y + e^{i kj\pi/n} ZX^j Z \rp = 2U_k.\] 

\noindent{\bf (2) $[H_k,V_\ell]=0$ if $k\neq \ell$ and $[H_k,V_k]=-2V_k$.}\\
Plugging the definitions of $H_k$ and $V_\ell$ into the commutator $[H_k, V_\ell]$ and using Table \ref{table:Cn-generalized-orbit-commutator},
\begin{multline*}
    [H_k,V_\ell] 
        =\frac{-i}{2n^2}\sum_{q=0}^{n-1}\sum_{p=1}^{n-1}\Big( e^{i\frac{\ell(q+1)\pi}{n}}\sin \frac{kp\pi}{n} (YX^{p+q-1}Y-ZX^{p-q-1}Z)\\ +e^{-i\frac{\ell q\pi}{n}}\sin \frac{kp\pi}{n} (YX^{p-q-2}Y-ZX^{p+q}Z)\Big).
\end{multline*}

        According to the definitions in Table \ref{table:Cn-def-generalized-orbit}, $YX^{k}Y\defeq ZX^{2n-k-2}Z$ and $ZX^{k}Z\defeq YX^{2n-k-2}Y$ if $k\ge n-1$ and $YX^{k}Y\defeq ZX^{-k-2}Z$ and $ZX^{k}Z\defeq YX^{-k-2}Y$ if $k\le -1$.
Therefore, we transform (1)  $YX^{p+q-1}Y\to ZX^{2n-p-q-1}Z$ and $ZX^{p+q}Z  \to YX^{2n-p-q-2}Y$  if $n-q\le p \le n-1$; (2) $ZX^{p-q-1}Z \to YX^{-p+q-1}Y$ and $YX^{p-q-2}Y \to ZX^{-p+q}Z$ if $1\le p \le q$. In this way, $[H_k,V_\ell]$ can be written as follows, where the four lines correspond to the four terms in the RHS of the above equality of $[H_k,V_\ell]$. (We also slightly adjust the range of the $p$ to make the inner summation over $q$ in a legal range.)
\begin{align*}
     & [H_k,V_\ell] \\
        = & -\frac{i}{2n^2}\sum_{q=0}^{n-2}\sum_{p=1}^{n-q-1}e^{i\frac{\ell(q+1)\pi}{n}}\sin \frac{kp\pi}{n}YX^{p+q-1}Y-\frac{i}{2n^2}\sum_{q=1}^{n-1}\sum_{p=n-q}^{n-1}e^{i\frac{\ell(q+1)\pi}{n}}\sin \frac{kp\pi}{n}ZX^{2n-p-q-1}Z\\
        &+\frac{i}{2n^2}\sum_{q=1}^{n-1}\sum_{p=1}^{q}e^{i\frac{\ell(q+1)\pi}{n}}\sin \frac{kp\pi}{n}YX^{-p+q-1}Y+\frac{i}{2n^2}\sum_{q=0}^{n-2}\sum_{p=q+1}^{n-1}e^{i\frac{\ell(q+1)\pi}{n}}\sin \frac{kp\pi}{n}ZX^{p-q-1}Z\\
        & -\frac{i}{2n^2}\sum_{q=1}^{n-1}\sum_{p=1}^{q}e^{-i\frac{\ell q\pi}{n}}\sin \frac{kp\pi}{n}ZX^{-p+q}Z-\frac{i}{2n^2}\sum_{q=0}^{n-2}\sum_{p=q+1}^{n-1}e^{-i\frac{\ell q\pi}{n}}\sin \frac{kp\pi}{n}YX^{p-q-2}Y\\
        &+\frac{i}{2n^2}\sum_{q=0}^{n-2}\sum_{p=1}^{n-q-1}e^{-i\frac{\ell q\pi}{n}}\sin \frac{kp\pi}{n}ZX^{p+q}Z+\frac{i}{2n^2}\sum_{q=1}^{n-1}\sum_{p=n-q}^{n-1}e^{-i\frac{\ell q\pi}{n}}\sin \frac{kp\pi}{n}YX^{2n-p-q-2}Y.
\end{align*}
Now we change the inner variable $p$ to $j$ to transform all $YX^{\cdots} Y$ terms to $YX^{j-1}Y$, and all $ZX^{\cdots} Z$ terms to $ZX^{j}Z$.
\begin{align*}
    [H_k,V_\ell]
        = & -\frac{i}{2n^2}\sum_{q=0}^{n-2}\sum_{j=q+1}^{n-1}e^{i\frac{\ell(q+1)\pi}{n}}\sin \frac{k(j-q)\pi}{n}YX^{j-1}Y \\ &-\frac{i}{2n^2}\sum_{q=1}^{n-1}\sum_{j=n-q}^{n-1}e^{i\frac{\ell(q+1)\pi}{n}}\sin \frac{k(2n-j-q-1)\pi}{n}ZX^{j}Z\\
        &+\frac{i}{2n^2}\sum_{q=1}^{n-1}\sum_{j=0}^{q-1}e^{i\frac{\ell(q+1)\pi}{n}}\sin \frac{k(q-j)\pi}{n}YX^{j-1}Y \\ & +\frac{i}{2n^2}\sum_{q=0}^{n-2}\sum_{j=0}^{n-q-2}e^{i\frac{\ell(q+1)\pi}{n}}\sin \frac{k(j+q+1)\pi}{n}ZX^{j}Z\\
        & -\frac{i}{2n^2}\sum_{q=1}^{n-1}\sum_{j=0}^{q-1}e^{-i\frac{\ell q\pi}{n}}\sin \frac{k(q-j)\pi}{n}ZX^{j}Z  \\ &-\frac{i}{2n^2}\sum_{q=0}^{n-2}\sum_{j=0}^{n-q-2}e^{-i\frac{\ell q\pi}{n}}\sin \frac{k(q+j+1)\pi}{n}YX^{j-1}Y\\
        &+\frac{i}{2n^2}\sum_{q=0}^{n-2}\sum_{j=q+1}^{n-1}e^{-i\frac{\ell q\pi}{n}}\sin \frac{k(j-q)\pi}{n}ZX^{j}Z \\ &+\frac{i}{2n^2}\sum_{q=1}^{n-1}\sum_{j=n-q}^{n-1}e^{-i\frac{\ell q\pi}{n}}\sin \frac{k(2n-j-q-1)\pi}{n}YX^{j-1}Y
\end{align*}
We add the following four terms which are equal to $0$ to the RHS of the above equation:
\begin{align*}
   & e^{i\frac{\ell(q+1)\pi}{n}}\sin \frac{k(q-q)\pi}{n}YX^{q-1}Y, &e^{i\frac{\ell(q+1)\pi}{n}}\sin \frac{k(-(n-q-1)-q-1)\pi}{n}ZX^{n-q-1}Z, \\ 
   & e^{-i\frac{\ell q\pi}{n}}\sin \frac{k(q-q)\pi}{n}ZX^qZ, &e^{-i\frac{\ell q \pi}{n}}\sin \frac{k((n-q-1)+q+1)}{n}YX^{n-q-2}Y.
\end{align*}
 Then we have
 \begin{align*}
     [H_k,V_\ell]
        =&- \frac{i}{2n^2}\sum_{j=0}^{n-1}\sum_{q=0}^{n-1}\left(e^{i\frac{\ell(q+1)\pi}{n}}\sin \frac{k(j-q)\pi}{n}+e^{-i\frac{\ell q\pi}{n}}\sin \frac{k(j+q+1)}{n}\right)YX^{j-1}Y \\&+ \frac{i}{2n^2}\sum_{j=0}^{n-1}\sum_{q=0}^{n-1}\left(e^{i\frac{\ell(q+1)\pi}{n}}\sin \frac{k(j+q+1)\pi}{n}+e^{-i\frac{\ell q\pi}{n}}\sin \frac{k(j-q)}{n}\right)ZX^{j}Z.
 \end{align*}
According to Eq. \eqref{eq:expsin}, we can obtain that
\begin{align*}
    &[H_k,V_\ell] \\= &-\frac{i}{2n^2} \sum_{j=0}^{n-1}\sum_{q=0}^{n-1}\Big( -ie^{i\frac{(k(2j+1)+\ell)\pi}{2n}} \cos \frac{(\ell-k)(2q+1)\pi}{2n}+ i e^{i\frac{(\ell-k(2j+1))\pi}{2n}} \cos \frac{(\ell+k)(2q+1)\pi}{2n}\Big) YX^{j-1}Y\\
        &+\frac{i}{2n^2} \sum_{j=0}^{n-1}\sum_{q=0}^{n-1} \Big( -ie^{i\frac{(k(2j+1)+\ell)\pi}{2n}} \cos \frac{(k+\ell)(2q+1)\pi}{2n}+ie^{i\frac{(\ell-k(2j+1))\pi}{2n}} \cos \frac{(\ell-k)(2q+1)\pi}{2n}\Big) ZX^{j}Z.
\end{align*}
By Lemma \ref{lem:sumodd}, we know that if $k\ne \ell$, then $\sum_{q=0}^{n-1} \cos \frac{(k\pm \ell) (2q+1)\pi }{2n} = 0$ and thus $[H_k,V_\ell] = 0$. If $k = \ell$, then $\sum_{q=0}^{n-1} \cos \frac{(k+ \ell) (2q+1)\pi }{2n} = 0$ and \[ [H_k,V_k] =  -\frac{1}{2n} \sum_{j=0}^{n-1} \lp e^{ik(j+1)\pi/n} YX^{j-1} Y + e^{-i kj\pi/n} ZX^j Z\rp = -2V_k. \]

\noindent{\bf (3) $[U_k,V_\ell]=0$ if $k\neq \ell$ and $[U_k,V_k]=H_k$.}\\
First, we prove that $[U_k,V_k]=H_k$. Plugging the definitions of $U_k$ and $V_k$ into the commutator $[U_k, V_k]$ and using Table \ref{table:Cn-generalized-orbit-commutator},
\begin{multline*}
 [U_k,V_k]
    =  \frac{1}{8n^2}\sum_{p=0}^{n-1}\sum_{q=0}^{n-1}\Big( \Big( e^{-i\frac{k(p-q)\pi}{n}}-e^{i\frac{k(p-q)\pi}{n}} \Big) YX^{p-q-1}Z\\+\Big( e^{-i\frac{k(p+q+1)\pi}{n}}-e^{i\frac{k(p+q+1)\pi}{n}}\Big) YX^{p+q}Z\Big).   
\end{multline*}
Then we can obtain that
\[[U_k,V_k]
    =  -\frac{i}{4n^2}\sum_{p=0}^{n-1}\sum_{q=0}^{n-1}\Big( \sin \frac{k(p-q)\pi}{n} YX^{p-q-1}Z+\sin \frac{k(p+q+1)\pi}{n} YX^{p+q}Z\Big).\]
According to the definition in Table \ref{table:Cn-def-generalized-orbit}, $YX^{k}Z\defeq -YX^{-k-2}Z$ if $k\le -1$ and $YX^{k}Z\defeq -YX^{2n-k-2}Z$ if $k\ge n-1$. we transform $YX^{p-q-1}Z \to -YX^{-p+q-1}Z$ if $p\le q\le n-1$ and $YX^{p+q}Z \to -Y^{2n-p-q-2}Z$ if $n-p-1\le q \le n-1$. Then we have
\begin{align*}
    &[U_k,V_k]\\
    = &-\frac{i}{4n^2}\Big(\sum_{p=1}^{n-1}\sum_{q=0}^{p-1} \sin \frac{k(p-q)\pi}{n} YX^{p-q-1}Z+\sum_{p=0}^{n-1}\sum_{q=p}^{n-1} \sin \frac{k(p-q)\pi}{n}(-YX^{-p+q-1}Z)\\
     &+\sum_{p=0}^{n-2}\sum_{q=0}^{n-p-2}\sin \frac{k(p+q+1)\pi}{n} YX^{p+q}Z+\sum_{p=0}^{n-1}\sum_{q=n-p-1}^{n-1}\sin \frac{k(p+q+1)\pi}{n}(-YX^{2n-p-q-2}Z)\Big).
\end{align*}
For the terms $YX^{p-q-1}Z$, $YX^{-p+q-1}Z$, $YX^{p+q}Z$ and $YX^{2n-p-q-2}Z$, we replace them by term $YX^{j-1}Z$. Therefore, we have
    \begin{align*}
     [U_k,V_k]
    = &-\frac{i}{4n^2}\Big(\sum_{p=1}^{n-1}\sum_{j=1}^{p} \sin \frac{kj\pi}{n} YX^{j-1}Z-\sum_{p=0}^{n-1}\sum_{j=0}^{n-p-1} \sin \frac{-kj\pi}{n} YX^{j-1}Z\\
    &+\sum_{p=0}^{n-2}\sum_{j=p+1}^{n-1}\sin \frac{kj\pi}{n} YX^{j-1}Z-\sum_{p=0}^{n-1}\sum_{j=n-p}^{n}\sin \frac{k(2n-j)\pi}{n} YX^{j-1}Z\Big)
    \end{align*}
    The above equation can be simplified as
    \begin{align*}
        [U_k,V_k]   = & -\frac{i}{4n^2}\Big(\sum_{j=1}^{n-1} n\sin\frac{kj\pi}{n}YX^{j-1}Z+\sum_{j=0}^{n} n\sin\frac{kj\pi}{n}YX^{j-1}Z\Big)\\
    = & -\frac{i}{2n} \sum_{j=1}^{n-1}\sin \frac{kj\pi}{n}YX^{j-1}Z
    = H_k.
    \end{align*}
    
Second, we prove that $[U_k,V_\ell]=0$ if $k\neq \ell$. For any $x,y\in \{U_j,V_j: 1\le j \le n-2\}$, $[x,y]$ can be represented as $[x,y]=\sum_{j=0}^{n-2}\alpha_j YX^jZ$ 
based on Table \ref{table:Cn-generalized-orbit-commutator} and definitions of $U_j,V_j$. Then for any $1\le k \le n-1$,
\begin{equation}\label{eq:HUV}
    [H_k,[x,y]]=\sum_{j=0}^{n-2}\alpha_j [H_k,YX^jZ]= -\frac{i}{2n}\sum_{j=0}^{n-2}\sum_{j'=1}^{n-1}\alpha_j\sin \frac{kj'\pi}{n} [YX^{j'-1}Z,YX^jZ]=0,
\end{equation}
since $[YX^{j'-1}Z,YX^jZ]=0$ (Table \ref{table:Cn-generalized-orbit-commutator}).  According to $[H_k,U_k]=2U_k$, $[V_\ell,H_k]=0$, Eq. \eqref{eq:HUV} and Jacobi identity, we have
\begin{align*}
    [U_k,V_\ell] &=-\frac{1}{2}[V_\ell,[H_k,U_k]]=\frac{1}{2}[H_k,[U_k,V_\ell]]+\frac{1}{2}[U_k,[V_\ell,H_k]]=\frac{1}{2}[H_k,[U_k,V_\ell]]=0.
\end{align*}

\noindent{\bf (4) $[U_k,U_\ell]=0$ and $[V_k,V_\ell]=0$.}\\
If $k=\ell$, then $[U_k,U_\ell]=[V_k,V_\ell]=0$. 
If $k\neq \ell$, based on $[H_\ell,U_\ell]=2U_\ell$, $[H_\ell,V_\ell]=-2V_\ell$, $[H_\ell, U_k]=[H_\ell, V_k]=0$, Eq. \eqref{eq:HUV} and Jacobi identity, we have 
\begin{align*}
    &[U_k,U_\ell]=\frac{1}{2}[U_k,[H_\ell,U_\ell]]=-\frac{1}{2}[H_\ell,[U_\ell,U_k]]-\frac{1}{2}[U_\ell, [U_k,H_\ell]]=\frac{1}{2}[H_\ell,[U_k,U_\ell]]=0,\\
    &[V_k,V_\ell]=-\frac{1}{2}[V_k,[H_\ell,V_\ell]]=\frac{1}{2}[H_\ell,[V_\ell,V_k]]+\frac{1}{2}[V_\ell, [V_k,H_\ell]]=-\frac{1}{2}[H_\ell,[V_k,V_\ell]]=0.
\end{align*}

\noindent{\bf (5) $[H_k,H_\ell]=0$.}\\
Based on $[YX^{j-1}Z, YX^{j'-1}Z]=0$ (Table \ref{table:Cn-generalized-orbit-commutator}), we have
\[[H_k,H_\ell]=-\frac{1}{4n^2} \sum_{j=1}^{n-1}\sum_{j'=1}^{n-1} \sin \frac{kj\pi}{n} \sin \frac{\ell j'\pi}{n} [YX^{j-1}Z,YX^{j'-1}Z] = 0.\]
\end{proof}


\section{Cycle Graphs - Other Proofs}

This sections includes proofs of Theorems~\ref{thm:Cn-basis}, \ref{thm:cycle-center}, \ref{thm:AB}, \ref{thm:Cn-ABn-by-ABk}; Facts~\ref{fact:adH-diag}-\ref{fact:dim-cartan-greater-three}, and Lemmas \ref{lem:direct-sum-complexification}, \ref{thm:Cn-generating-basis}. 
We also show how to derive Eq.~\eqref{eq:lambda_j}, that shows that $\lambda_j = 16\cos(j\pi/n)$, and Eq.~\eqref{eq:H-YXZ}, which expresses the $H_k$ operators in terms of the $YX^{j}Z$ operators. Throughout this section, we will use $\.g$ to denote $\.g_{C_n}$.

\CnBasis*

\begin{proof}
Define a Lie algebra $\.s \defeq \spn_{\mbR} \+B$. Since the elements in $\+B$  are all summations of distinct Pauli strings, they are orthogonal (with respect to the Hilbert-Schmidt inner product $\langle A,B\rangle = \tr(A^\dagger B)$), and thus also linearly independent. Therefore $\dim(\.s) = 3n-1$, and what remains is to show that $\.g = \.s$.

\paragraph{1. The inclusion $\.g \subseteq  \.s$.} Note that both generators $X$ and ${ZZ}$ are in $\+B$. Table \ref{tab:Cn-adA-adB-on-orbits} shows the actions $ad_{X}$ and $ad_{{ZZ}}$ on the elements in the basis $\+B$, and the actions can be verified by direct calculations. We can see from the table that the actions map each basis element in $\+B$ to within $\.s$. Since $\.g$ is generated by repeated applications of $ad_X$ and $ad_{ZZ}$ to the generators (as in Algorithm \ref{alg:DLA-generation}), $\langle \{X,ZZ\}\rangle_{\text{Lie},\mbR}\subseteq\.s$.

\paragraph{2. The inclusion $\.s \subseteq  \.g$.} We need to show that all the basis vectors in Eq. \eqref{eq:Cn-basis} can be generated from $X$ and ${ZZ}$. Using the notation in Table \ref{tab:Cn-adA-adB-on-orbits}, it suffices to show that (1) $Z X^t Z$, $Y X^t Y$ and $Y X^t Z$ can be generated for all $t= 0, 1, \ldots, n-2$, and (2) 
$ZX^{n-1}Z \defeq X^{n-1}$  can be generated. 

When $t=0$, ${YX^{t-1}Y}$ and ${ZX^tZ}$ are the generators $X$ and $ZZ$. The claim follows from the following relations, which hold for all $t=0,1,\ldots, n-2$:
\begin{align*}
    {YX^tZ} & = {\frac{1}{2}} [X, {ZX^tZ}], \\
    {YX^tY} & = {ZX^tZ} + \frac{1}{4}  [{X}, {YX^tZ}],  \\
    {ZX^{t+1}Z} & = {YX^{t-1}Y} + \frac{1}{4} [{ZZ}, {YX^tZ}].
\end{align*}  
\end{proof}

\CycleCenter*

\begin{proof}
    Any $H\in\.g$ can be expressed as a linear combination of the basis vectors, viz., 
    \[H = \alpha X + \sum_{t=0}^{n-2}(\beta_{1,t} {ZX^tZ} + \beta_{2,t} {YX^tY} + \beta_{3,t} {YX^tZ}) + \gamma {X^{n-1}},\quad \text{for }\alpha, \beta_{j,k},\gamma \in\mathbb{R}.\]
    It is easily verified by induction (and Jacobi identity) that $H\in \.c$ if and only if $[X,H] = [ZZ,H] = 0$. Using Table \ref{tab:Cn-adA-adB-on-orbits}:
    \begin{align*}
         [X,H] 
        = \ & \alpha [X,X] + \sum_{t=0}^{n-2} (\beta_{1,t} [X, {ZX^tZ}] + \beta_{2,t} [X,{YX^tY}] + \beta_{3,t} [X,{YX^tZ}]) + \gamma [X,{X^{n-1}}] \\
        =\ & {\sum_{t=0}^{n-2}2(\beta_{1,t}-\beta_{2,t}) {YX^tZ}+\sum_{t=0}^{n-2} 4\beta_{3,t}({YX^tY} - {ZX^tZ})}
    \end{align*}
    As the basis vectors are linearly independent, $[X,H] = 0$ requires that 
    \[\beta_{1,t}=\beta_{2,t}, \ \beta_{3,t}=0, \quad \forall t\in \{0, 1, \ldots, n-2\}.\]
    To simplify notation, let $\beta_k:=\beta_{1,k}=\beta_{2,k}$. Then $H = \alpha X + \sum_{t=0}^{n-2}(\beta_{t} {ZX^tZ} + \beta_{t} {YX^tY} ) + \gamma {X^{n-1}}$. Again, from Table \ref{tab:Cn-adA-adB-on-orbits}:
    \begin{align*}
        [{ZZ},H]  & =  \alpha [{ZZ},X] + \sum_{t=0}^{n-2} (\beta_{t} [{ZZ}, {ZX^tZ}] + \beta_{t} [{ZZ},{YX^tY}]) + \gamma [{ZZ},{X^{n-1}}]\\
          &{ =-2 \alpha {YZ} -\sum_{t=1}^{n-2} 2 \beta_{t} {YX^{t-1}Z} + \sum_{t=0}^{n-3} 2\beta_{t} {YX^{t+1}Z} - 2\gamma {YX^{n-2}Z}}\\
          &{= -2 (\alpha+\beta_1) {YZ} + \sum_{t=1}^{n-3}2(\beta_{t-1} -\beta_{t+1})YX^tZ +2(\beta_{n-3}-\gamma)YX^{n-2}Z}
    \end{align*}
    Forcing $[{ZZ},H] = 0$ gives that
    \begin{align*}
        & \alpha+\beta_{1}=0, \quad  \beta_{n-3}=\gamma, \text{ and } \beta_{t-1}=\beta_{t+1}, \quad   \forall t = 1, 2, \ldots, n-3.
    \end{align*}
    If $n$ is odd, then $-\alpha=\beta_1=\beta_3=\cdots =\beta_{n-4}=\beta_{n-2}$ and $\beta_0=\beta_2=\cdots =\beta_{n-5}=\beta_{n-3}=\gamma$; if $n$ is even, then $-\alpha=\beta_1=\beta_3=\cdots =\beta_{n-5}=\beta_{n-3}=\gamma$ and $\beta_0=\beta_2=\cdots =\beta_{n-4}=\beta_{n-2}$.
\end{proof}

\DirectSumComplexification*

\begin{proof}
    \begin{enumerate}
        \item Suppose that $\.g = \.g_1 \oplus \cdots \oplus \.g_n$. Take any basis $\mathcal B_j$ for vector space $\.g_j$, and combine them as a basis $\mathcal B = \mathcal B_1\cup \cdots \cup \mathcal B_n$ for vector space $\.g$. Then each element $A$ in $\.g$ is a block diagonal matrix $diag(A_1, \ldots, A_n)$ under this basis $\mathcal B$. By the definition of complexification, $\.g_{\mb{C}}$ contains the matrices $A+Bi$ where $A,B\in \.g$. But $A+Bi$ is exactly $diag(A_1 + B_1 i, \ldots, A_n + B_n i)$,  which form the complex Lie algebra $(\.g_1)_{\mb{C}} \oplus \cdots \oplus (\.g_n)_{\mb{C}}$. 
        \item Suppose that $\.g_{\mb{C}} = \.g_1' \oplus \cdots \oplus \.g_n'$ for some complex Lie algebras $\.g_1', \ldots, \.g_n'$. Again, take the basis $\mathcal B_j'$ for vector space $\.g_j'$, and combine them as a basis $\mathcal B' = \mathcal B_1'\cup \cdots \cup \mathcal B_n'$ for vector space $\.g_{\mb{C}}$. Now take any matrix $A\in \.g$, and write it in blocks as $A = [A_{jk}]_{j,k\in [n]}$, where $A_{jk}$ is the $(j,k)$-th block of $A$. Since the complexification $\.g_{\mb{C}}$ contains all matrices $A+Bi$ where $A,B\in \.g$, in particular, it contains $A+Ai$. For each off-diagonal block $(j,k)$ with $j\ne k$, $A+Ai$ inside this block is $A_{jk} + A_{jk}i$. Suppose $A_{jk} = B+Ci$ with $B,C$ both real matrices. Then $A_{jk} + A_{jk}i = (B-C) + (B+C)i$. Since all matrices in $\.g_{\mb{C}}$ has 0 off-diagonal blocks, we know that $B-C = B+C = 0$, which enforces $B = C = 0$. Therefore, any matrix $A\in \.g$ has all off-diagonal blocks being 0. Therefore $A = diag(A_1, \ldots, A_n)$, and thus $\.g = \.g_1\oplus \cdots \oplus \.g_n$ for some $\.g_1, \ldots, \.g_n$, with matrices in $\.g_j$ inside $\.g_j'$ as a space, and $\.g_j' = (\.g_j)_{\mb{C}}$.
    \end{enumerate}
\end{proof}

\CanonicalFromAB*

\begin{proof} From the $\.{sl}(2,\mb{C})$ commutation relations we have
\eq{
[H_k, A]&= \sum_{j=1}^{n-1}\left[ H_k, a_j^{(u)}U_j + a_j^{(v)}V_j + a_j^{(h)}H_j\right]=2\lp a_k^{(u)}U_k - a_k^{(v)}V_k\rp,\\ 
[H_k, B] &= 2\lp b_k^{(u)}U_k - b_k^{(v)}V_k\rp.
}
It follows from Eq.~\eqref{eq:uvh-alpha} that
\eql{
\tilde{U}_k^{(m_k, c_k)} &= -\frac{2i}{c} \lp \lp a_k^{(u)} + e^{-i\pi m_k} b_k^{(u)}\rp U_k - \lp a_k^{(v)} + e^{-i\pi m_k}b_k^{(v)}\rp V_k\rp, \label{eq:u_alpha_m_c}\\
\tilde{V}_k^{(m_k, c_k)} &= \frac{2i}{c} \lp\lp a_k^{(u)} + e^{i\pi m_k} b_k^{(u)}\rp U_k - \lp a_k^{(v)} + e^{i\pi m_k}b_k^{(v)}\rp V_k\rp,\label{eq:v_alpha_m_c}
}
and thus
\eq{
[H_k, \tilde{U}_k^{(m_k, c_k)}]&= -\frac{4i}{c} \lp \lp a_k^{(u)} + e^{-i\pi m_k} b_k^{(u)}\rp U_k + \lp a_k^{(v)} + e^{-i\pi m_k}b_k^{(v)}\rp V_k\rp, \\
[H_k, \tilde{V}_k^{(m_k, c_k)}]&= \frac{4i}{c} \lp\lp a_k^{(u)} + e^{i\pi m_k} b_k^{(u)}\rp U_k + \lp a_k^{(v)} + e^{i\pi m_k}b_k^{(v)}\rp V_k\rp.
}
The requirements that $[H_k, \tilde{V}_k^{(m_k, c_k)}]=-2\tilde{V}_k^{(m_k, c_k)}$ and $[H_k, \tilde{U}_k^{(m_k, c_k)}]=2\tilde{U}_k^{(m_k, c_k)}$ give Eq.~\eqref{eq:C1} and Eq.~\eqref{eq:C2}, respectively. However, if Eq.~\eqref{eq:C1} and Eq.~\eqref{eq:C2} are satisfied, then Eq.~\eqref{eq:C3} and Eq.~\eqref{eq:C4} must also be true, otherwise one or both of $\tilde{U}_k, \tilde{V}_k$ will be zero.  Next, we consider
\eq{
[&\tilde{U}_k^{(m_k, c_k)}, \tilde{V}_k^{(m_k, c_k)}]\\
=& \frac{4}{c^2} \lp -\lp a_k^{(u)}+e^{-i\pi m_k}b_k^{(u)}\rp
 \lp a_k^{(v)} + e^{i\pi m_k}b_k^{(v)}\rp
 +
\lp a_k^{(u)} + e^{i\pi m_k}b_k^{(u)}\rp 
\lp a_k^{(v)} + e^{-i\pi m_k}b_k^{(v)}\rp\rp H_k \\
=&\frac{4}{c^2}\lp a_k^{(v)}b_k^{(u)} - a_k^{(u)} b_k^{(v)} \rp
 \lp e^{i\pi m_k} - e^{-i\pi m_k}\rp H_k,
}
and Eq.~\eqref{eq:C5} follows from the requirement that $[\tilde{U}_k^{(m_k, c_k)}, \tilde{V}_k^{(m_k, c_k)}]=H_k$. 

In the other direction, if Eq.~$\eqref{eq:C1}$ to Eq.~$\eqref{eq:C5}$ hold then the $\.{sl}(2,\mb{C})$ commutation relations Eq.~\eqref{eq:sl2C-canonical-basis} are satisfied by $\tilde{U}_k^{(m,c)}, \tilde{V}_k^{(m,c)},H_k$, for all $k\in[n-1]$, and we see from Eq.~\eqref{eq:u_alpha_m_c} and Eq.~\eqref{eq:v_alpha_m_c} that $\tilde{U}_k^{(m_k,c_k)}$ and $\tilde{V}_k^{(m_k,c_k)}$ are proportional to $U_k$ and $V_k$, respectively, and non-zero. 
\end{proof}

\CnGenBasis*

\begin{proof}
   The proof is by induction on $k$. For $k=1$, we have $AB = [A,B] = 2YZ$. For $k=2, \ldots, n$, assuming that $(AB)^{k-1} = \sum_{j=0}^{k-2} c_{k-1,j} YX^j Z$, we can  directly calculate (using  the rules in Table \ref{tab:Cn-adA-adB-on-orbits}) that 
   \begin{align*}
     (AB)^k = AB(AB)^{k-1} =& \sum_{j=0}^{k-2} 4c_{k-1,j} [X,ZX^{j+1} Z - YX^{j-1} Y] \\= & \sum_{j=0}^{k-2} 8c_{k-1,j} (YX^{j+1} Z + YX^{j-1} Z).  
   \end{align*}
    Noting that $c_{0,-1} = c_{k-1,k-1} = c_{k-1,k} = 0$ and $YX^{-1}Z = 0$, we obtain \[(AB)^k = 8\sum_{j=0}^{k-1} c_{k-1,j-1} YX^{j} Z + 8\sum_{j=0}^{k-1} c_{k-1,j+1} YX^{j} Z = \sum_{j=0}^{k-1} 8(c_{k-1,j-1} + c_{k-1,j+1}) YX^{j} Z,\] 
    as claimed.
\end{proof}

The recursive relation for the coefficients $c_{k,j}$ in Lemma~\ref{thm:Cn-generating-basis} be solved to give an explicit formula, as shown in the next lemma. In the proof, we will use the tridiagonal Toeplitz matrix $T_{n}(a,b,c)$, whose subdiagonal elements are all $a$, diagonal elements are all $b$, superdiagonal elements are all $c$, and all other elements are $0$.
\begin{lem}\label{thm:Cn-Cartan-coeffi}
    For any $(k,j)$ with $1\le k \le n$ and $0\le j \le k-1$, we have $(AB)^k  = \sum_{j=0}^{k-1} c_{k,j} YX^j Z$, with
    \[c_{k,j} = 2\cdot 8^{k-1} (T^{k-1})_{1,j+1} = \frac{2^{4k-2}}{n} \sum_{j'=1}^{n-1} \sin\frac{(j+1)j'\pi}{n} \sin \frac{j'\pi}{n} \cos^{k-1} \frac{j'\pi}{n},\]
    where $T$ is the tridiagonal Toeplitz matrix $T = T_{n-1}(1,0,1)$.
\end{lem}
\begin{proof}
    Define an $(n-1)\times (n-1)$ matrix $C'$ with 
    the $(k,j)$-th entry being $c_{k,j-1}$. Let $T$ be the tridiagonal Toeplitz matrix $T = T_{n-1}(1,0,1)$, whose subdiagonal and superdiagonal elements are $1$, and all other elements are 0. 
    The recursion in Eq. \eqref{eq:ABk-recursion} can be then rephrased as $\bra{k} C' = 8 \bra{k-1} C' T$; indeed the $j$-th entry of the RHS is 
    \[8 \sum_{j'=1}^{n-1} C'_{k-1,j'} T_{j',j} = 8 \sum_{j'=1}^{n-1} c_{k-1,j'-1} T_{j',j} = 8 (c_{k-1,j-2} + c_{k-1, j}) = c_{k,j-1} = C'_{k,j},\] 
    the $j$-th entry of the LHS. Repeatedly applying this yields 
    \[c_{k,j} = \bra{k} C' \ket{j+1} = 8^{k-1} \bra{1} C' T^{k-1} \ket{j+1} = 8^{k-1} \sum_{t=1}^{n-1} C'_{1,t} (T^{k-1})_{t,j} = 2\cdot 8^{k-1} (T^{k-1})_{1,j+1},\] 
    where in the last equality we used the observation that the first row of $C$ is $(2, 0, \ldots, 0)$. 

    The power of $T$ can be computed by its spectral decomposition, which is well known to be $T = \sum_{j=1}^{n-1} \mu_j \ket{v_j}\bra{v_j}$ with eigenvalues $\mu_j = 2\cos \frac{j\pi}{n}$ and (normalized) eigenvectors $\ket{v_j} = \sqrt{\frac{2}{n}} \sum_{j'} \sin \frac{jj'\pi}{n}\ket{j'}$. The power of $T$ is therefore $T^{k-1} = \sum_{j=1}^{n-1} 2^{k-1} \cos^{k-1} \frac{j\pi}{n} \ket{v_j}\bra{v_j}$, and 
    \begin{align*}
        c_{k,j} = 2\cdot 8^{k-1}  (T^{k-1})_{1,j+1}& = 2\cdot 16^{k-1} \sum_{j'=1}^{n-1} (\cos^{k-1} \frac{j'\pi}{n}) \frac{2}{n} \sin\frac{(j+1)j'\pi}{n} \sin\frac{j'\pi}{n}  \\ &= \frac{2^{4k-2}}{n} \sum_{j'=1}^{n-1} \sin\frac{(j+1)j'\pi}{n} \sin \frac{j'\pi}{n} \cos^{k-1} \frac{j'\pi}{n}.
    \end{align*}
\end{proof}

To prove Theorem~\ref{thm:Cn-ABn-by-ABk} we will make use of Vieta's formulas and the Cayley-Hamilton theorem:

\begin{lem}[Vieta's formulas, \cite{coolidge2004treatise}]\label{lem:vieta} Let $P(x)=\sum_{j=0}^n a_j x^j$ be a degree-$n$ polynomial with coefficients $a_j\in\mb{C}$, and let $z_1, \ldots, z_n\in\mb{C}$ be the roots of $P$. Then, for $1\le k \le n$,
\eq{
(-1)^{k}a_{n-k}/a_n&=\sigma_k(z_1, \ldots, z_n),
}
where $\sigma_k(z_1, \ldots, z_n)\defeq \sum_{1\le j_1<j_2<\cdots < j_k\le n}z_{j_1}z_{j_2}\cdots z_{j_k}$ is the elementary symmetric polynomial of degree $k$.
\end{lem}

\begin{lem}[Cayley-Hamilton theorem, \cite{mirsky2012introduction}]\label{lem:cayley-ham} Let $A\in\mb{C}^{n\times n}$ be a square matrix, and $p_A(\lambda) = \det(\lambda I-A)$ its characteristic polynomial. Then $p_A(A) = 0$.
\end{lem}

\CnABnbyAbk*

\begin{proof}
    For any fixed $k$ consider the vector $c_k \defeq  (c_{k,0}, \ldots, c_{k,n-2})^T$, where $c_{k,j}$ are the coefficients in Lemma~\ref{thm:Cn-Cartan-coeffi}. Note that $T = T_{n-1}(1,0,1)$ is symmetric, therefore the Proposition~\ref{thm:Cn-Cartan-coeffi} shows that $c_k = 2\cdot 8^{k-1} T^{k-1}\ket{e_1}$. Let $p_T(\lambda) = \lambda^{n-1} - \sum_{k=0}^{n-2} d_k \lambda^k$ be the characteristic polynomial of $T$. By Vieta's formulas (Lemma \ref{lem:vieta}), we know that 
    \begin{equation}\label{eq:Cn-Toeplitz-recur-coeffi}
        d_k = (-1)^{n-k} \sigma_{n-k-1} \left(2\cos \frac{\pi}{n}, \ldots, 2\cos \frac{(n-1)\pi}{n}\right),    
    \end{equation} 
    where we used the fact that the roots of the characteristic polynomial of a diagonalizable square matrix $A$ are just the eigenvalues of $A$. By Lemma \ref{lem:cayley-ham} we have 
    \begin{equation}\label{eq:Toeplitz-recursion}
        T^{n-1} = \sum_{k=0}^{n-2} d_k T^k. 
    \end{equation}
    Therefore, Lemmas \ref{thm:Cn-generating-basis} and \ref{thm:Cn-Cartan-coeffi} combined gives $(AB)^n = \sum_{j=0}^{n-1} 2\cdot 8^{n-1} \left(T^{n-1}\right)_{1,j+1} YX^j Z$. Plugging Eq.~\eqref{eq:Toeplitz-recursion} in and changing the order of summation, we see that 
    \[(AB)^n = \sum_{k=1}^{n-1} 8^{n-k} d_{k-1} \sum_{j=0}^{n-1}  2\cdot 8^{k-1} \left(T^{k-1}\right)_{1,j+1} YX^j Z.\]
    Now apply Lemmas \ref{thm:Cn-generating-basis} and \ref{thm:Cn-Cartan-coeffi} again to change it back to $(AB)^k$ form: \[(AB)^n = \sum_{k=1}^{n-1} 8^{n-k} d_{k-1} (AB)^k.\] Finally, use Eq.~\eqref{eq:Cn-Toeplitz-recur-coeffi} to replace $d_k$, and absorb the factor $8^{n-k}$ into the function $\sigma_{n-k}$, we get the claimed formula for $(AB)^n$.
\end{proof}

\ABisCartan*

\begin{proof} From Lemma \ref{thm:Cn-generating-basis} and Table \ref{table:Cn-generalized-orbit-commutator} it follows that $(AB)^j$ and $(AB)^k$ commute. Thus $\.h$ is Abelian. From Lemma \ref{thm:Cn-generating-basis}, each $(AB)^k$ contains an orbit $YX^{k-1}Z$ that $(AB)^1, \ldots, (AB)^{k-1}$ do not have. Thus, $(AB)^1, \ldots, (AB)^{k-1}$ are linearly independent, and $\dim(\.h) = n-1$. Recall that $[\.g,\.g]_\mb{C}$ is isomorphic to the direct sum of the $n-1$ copies of $\.{sl}(2,\mbC)$, which implies that any Cartan subalgebra is of dimension $n-1$. Putting these together, we see that $\.h$ is a maximal Abelian subalgebra. By Fact~\ref{fact:adH-diag}, $ad_H$ is diagonalizable, and $\.h$ is thus a Cartan subalgebra. 
\end{proof}

\GenNoOverlap*

\AsquaredAB*

\begin{proof}
    By rules in Table \ref{tab:Cn-adA-adB-on-orbits} and Lemma \ref{thm:Cn-generating-basis},
    \[
    AA(AB)^k = A \sum_{j=0}^{k-1} 4c_{k,j} (YX^j Y - ZX^j Z) = \sum_{j=0}^{k-1} 4c_{k,j} (-2YX^j Z - 2YX^j Z) = -16(AB)^k
    \]
    and 
    \[
    BB(AB)^k=B \sum_{j=0}^{k-1} 4c_{k,j} (ZX^{j+1} Z-YX^{j-1} Y) = \sum_{j=0}^{k-1} 4c_{k,j} (-2YX^j Z-2YX^j Z) = -16(AB)^k.
    \]
\end{proof}

We now show how to derive Eq.~\eqref{eq:lambda_j}, stated below as Proposition~\ref{prop:lambda_k}.

\begin{prop}\label{prop:lambda_k}For all $k\in [n-1], \lambda_k = 16\cos(k\pi/n)$.
\end{prop}

\begin{proof}
Recall from Eq.~\eqref{eq:el-sym-poly} that  $(AB)^n = \sum_{k=1}^{n-1} a_k (AB)^k$, where \[a_k = (-1)^{n-k+1}\sigma_{n-k}(16\cos(\pi/n), \ldots,  16\cos((n-1)\pi/n),\] and $\sigma_k$ is the $k^{\rm th}$ elementary symmetric polynomial in $n-1$ variables. Eq.~\eqref{eq:ABk-nulambdaH} says that $(AB)^k = \sum_{j=1}^{n-1} \nu_j\lambda_j^{k-1}H_j$, for all $k\in [n]$.  Combining these two expressions gives
\eq{
\sum_{j=1}^{n-1}\nu_j \lambda_j^{n-1}H_j &= \sum_{k=1}^{n-1}a_k \sum_{j=1}^{n-1} \nu_j \lambda_j^{k-1}H_j 
\Ra \sum_{j=1}^{n-1}\lp \lambda_j^{n-1}-\sum_{k=1}^{n-1}a_k \lambda_j^{k-1}\rp\nu_j H_j=0
}
Since $(AB), (AB)^2\ldots, (AB)^{n-1}$ are linearly independent, from Eq.~\ref{eq:AB-vandermonde} we see that the $\lambda_j$ must be distinct and the $\nu_j$ must be non-zero.  Thus, for each $j\in [n-1]$, $\lambda_j^{n-1} - \sum_{k=1}^{n-1}a_k\lambda_j^{k-1}=0$, i.e., the $\lambda_j$ are the $n-1$ distinct roots of the polynomial equation
\eq{
x^n + \sum_{k=1}^{n-1} (-1)^{n-k}\sigma_{n-k}\lp 16\cos (\pi/n), \ldots, 16\cos((n-1)\pi/n)\rp x^k.
}
By Vieta's formulas (Lemma~\ref{lem:vieta}), the roots of this equation are \[16\cos(\pi/n), 16\cos(2\pi/n), \ldots, 16\cos((n-1)\pi/n)\] and we are free to choose the ordering of the roots so that $\lambda_j = 16\cos(j\pi/n)$.
\end{proof}

Finally, we show how to derive Eq.~\eqref{eq:H-YXZ}, stated below as Proposition~\ref{prop:H-YXZ}. We shall need the following lemma regarding inverses of Vandermonde matrices:

\begin{lem}[Inverse of Vandermonde matrix, \cite{rawashdeh2019simple}]\label{eq:inverse-vandermonde}
    Let $\lambda_1,\lambda_2,\cdots, \lambda_{n}$ be distinct real numbers and $\Lambda=[\Lambda_{km}]\in \mbR^{n\times n}$  with $\Lambda_{km}=\lambda_m^{k-1}$ be the Vandermonde matrix. Then the inverse of $\Lambda$ is the matrix whose elements are
    \[(\Lambda^{-1})_{km}=(-1)^{n-m} \frac{\sigma_{n-m}(\lambda_{[n]-k})}
    {\prod_{1 \leq j \leq n, j\neq k} (\lambda_k - \lambda_j)},\]
    where $ \lambda_{[n]-k}= \lambda_{1}, \ldots,
    \lambda_{k-1}, \lambda_{k+1}, \ldots , \lambda_{n}$, $\sigma_0(\lambda_{[n]-k})\defeq 1$ and $\sigma_j(\lambda_{[n]-k})$ is the elementary symmetric polynomial of degree $j$ for $j\ge 1$.
\end{lem}

\begin{prop}\label{prop:H-YXZ} For all $k\in[n-1]$, $H_k = -\frac{i}{2n}\sum_{j=1}^{n-1}\sin\frac{jk\pi}{n}YX^{j-1}Z.$  
\end{prop}

\begin{proof} 

Define $C\in \mbR^{(n-1)\times (n-1)}$ by $C = [c_{k,j}]$ where $c_{k,j}$ is as in Lemma \ref{thm:Cn-generating-basis}. Define the Vandermonde matrix $\Lambda \in \mbR^{(n-1)\times (n-1)}$ by $\Lambda_{kj} = \lambda_j^{k-1}$. Let $N = diag(\nu_1, \ldots, \nu_{n-1})$. 
    Denote the formal vectors 
    \begin{align*}
        &\overrightarrow{AB} = (AB, (AB)^2, \ldots, (AB)^{n-1})^T, \\
        &\overrightarrow{YZ} = (YZ, YXZ, \ldots, YX^{n-2}Z)^T,\\
        &\overrightarrow{H} = (H_1, H_2, \ldots, H_{n-1})^T.
    \end{align*}
Lemma~\ref{thm:Cn-generating-basis} and Eq.~\eqref{eq:AB-vandermonde} imply that
\begin{equation}\label{eq:Cn-AB-YZ-H}
        \overrightarrow{AB} = C\cdot  \overrightarrow{YZ} \quad \text{and} \quad \overrightarrow{AB} = \Lambda \cdot N \cdot \overrightarrow{H},
    \end{equation}
where $\lambda_j = 16\cos(j\pi/n)$ and $\nu_j = 8i\sin(j\pi/n)$.
Combining the two equations in Eq.~\eqref{eq:Cn-AB-YZ-H} gives 
    \begin{equation}\label{eq:relation-AB-YZ}
        N\cdot \overrightarrow{H} = \Lambda^{-1}\cdot C\cdot \overrightarrow{YZ}.
    \end{equation}
    To evaluate $(\Lambda^{-1})_{jm}$ we define the degree $n-2$ polynomials
    \begin{equation}\label{eq:poly-pk}
      p_k(x) = \prod_{1 \leq j \leq n-1, j\neq k} (x - \lambda_j),  
    \end{equation}
    for $1 \leq i \leq n-1.$ Then Lemma \ref{eq:inverse-vandermonde} implies that
    \begin{equation}\label{eq:van-inv}
    (\Lambda^{-1})_{km} = (-1)^{n-1-m} \frac{\sigma_{n-1-m}(\lambda_{[n-1]-k})}
    {p_k(\lambda_{k})}.
    \end{equation}
    Putting this into Eq.~\eqref{eq:relation-AB-YZ}, we have
    \[\nu_k H_k  = \sum_{m=1}^{n-1} \sum_{j=1}^{n-1} (\Lambda^{-1})_{km}c_{m,j-1}YX^{j-1}Z.\]
    Substituting the $c$ values (Eq.~\eqref{eq:ckj-solved}), the $\Lambda^{-1}$ values (Eq. \eqref{eq:van-inv}), and $\lambda_{j} = 16\cos (j\pi/n)$, we have 
    \begin{align*}
        \nu_k H_k=&\frac{4}{n \cdot p_k(\lambda_{k})}  \sum_{m=1}^{n-1}\sum_{j=1}^{n-1} (-1)^{n-1-m} \sigma_{n-1-m} (\lambda_{[n-1]-k}) 
    \lp\sum_{\ell=1}^{n-1} \sin \frac{\ell j\pi}{n}  \sin \frac{\ell\pi}{n} \lambda_\ell^{m-1}\rp YX^{j-1}Z,\\
     =& \frac{4}{n p_k(\lambda_{k})}  \sum_{j=1}^{n-1}\sum_{\ell=1}^{n-1} \sin \frac{\ell j\pi}{n}  \sin \frac{\ell \pi}{n} 
\sum_{m=1}^{n-1} (-1)^{n-1-m} \sigma_{n-1-m} (\lambda_{[n-1]-k})  \lambda_\ell^{m-1} YX^{j-1}Z.
    \end{align*}
    The coefficient of $x^{n-2}$ in the degree $(n-2)$ polynomial $p_k(x)$  (Eq. \eqref{eq:poly-pk}) is $1$. For any $1\le \ell\le n-1$ and $\ell\neq k$, $\lambda_\ell$ is the root of $p_k(x)$. By Vieta's formulas (Lemma \ref{lem:vieta}), $p_k(x)$ can be represented as
    \[p_k(x)=\sum_{m=1}^{n-1} (-1)^{n-1-m} \sigma_{n-1-m} (\lambda_{[n-1]-k})  x^{m-1},\]  
    from which $\nu_kH_k$ can be simplified to
    \[ \nu_k H_k=\frac{4}{n} \sum_{j=1}^{n-1} \sum_{\ell=1}^{n-1} \frac{\sin \frac{\ell j\pi}{n}  \sin \frac{\ell \pi}{n}  p_k(\lambda_\ell )} {p_k(\lambda_{k})}YX^{j-1}Z.\]
    Since $p_k(\lambda_\ell)=0$ for all $k\neq \ell$, we have
    \[\nu_k H_k= \frac{4}{n} \sum_{j=1}^{n-1} \frac{\sin \frac{kj\pi}{n}  \sin \frac{k\pi}{n}  p_k(\lambda_k)} {p_k(\lambda_{k})} YX^{j-1}Z=\frac{4\sin(k\pi/n)}{n} \sum_{j=1}^{n-1} \sin \frac{kj\pi}{n}YX^{j-1}Z,\]  
and substituting $\nu_j = 8i\sin(j\pi/n)$ completes the proof.
\end{proof}


\section{Cycle Graphs - Purity Calculations}
Here we prove Theorem~\ref{thm:Cn-exp-var}, a proof of which was outlined in Section~\ref{sec:purity} of the main text.  We first need some preliminary results.  Recall from the main text that, for any Lie subalgebra $\.s\subseteq\.g_{C_n}$, $\mcP_{\.s}(\rho)= -\frac{i}{2^n}\mcP_{\.s}(X + X^{(n-1)})$.

\begin{lem}\label{lem:Cn-purity}
    For any Lie subalgebra $\.s\subseteq \.g_{C_n}$, take an orthogonal basis $B_1,\ldots ,B_s$ of $\.s$. Suppose $B_j = \sum_H c_{j,H} H$ is the linear combination of $B_j$ in orbits $H$ of $n$ distinct Pauli strings, and let $c_j$ be the coefficient vector $(c_{j,H})_H$. Then 
    \[\mcP_{\.s}(\rho) = \frac{n}{2^n} \sum_{j=1}^s \frac{|c_{j,X} + c_{j,X^{n-1}}|^2 }{\|c_j\|_2^2}, \quad \mcP_{\.s}(O) = 2^n \sum_{j=1}^s \frac{|c_{j,ZZ}|^2}{\|c_j\|_2^2}\]
\end{lem}
\begin{proof}
    Let $\bar B_j= B_j / \|B_j\|_F$, then $\bar B_1, \ldots, \bar B_s$ is an orthonormal basis of $\.s$. 
    Denote by $\Pi_{\.s}$ the orthogonal projection onto the subspace $\.s$. Then by definition
    \[
        \mcP_{\.s}(\rho) = 4^{-n} \|\Pi_{\.s} (X+X^{n-1})\|_F^2 =  4^{-n} \sum_{j=1}^s |\langle X+X^{n-1},\bar{B}_j\rangle |^2 
        = 4^{-n} \sum_{j=1}^s \frac{|\langle X+X^{n-1},{B}_j\rangle |^2}{\|B_j\|_F^2}.
    \]
    $B_j$ is a linear combination of different (and orthogonal) Pauli orbit-sums, and only the components of $X$ and $X^{n-1}$ have contribution to the inner product in the numerator. More specifically, 
    \[
        \langle X+X^{n-1},{B}_j\rangle = \langle X+X^{n-1},c_{j,X} X + c_{j,X^{n-1}} X^{n-1} \rangle = n2^n (c_{j,X} + c_{j,X^{n-1}}), 
    \]
    as $\langle X, X\rangle = \langle X^{n-1}, X^{n-1}\rangle = n2^n$ and $\langle X, X^{n-1}\rangle = 0$. Finally note that $\|B_j\|_F^2 = \|c_j\|_2^2 n2^n$, therefore  
    \[
        \mcP_{\.s}(\rho) = \frac{n}{2^n} \sum_{j=1}^s \frac{|c_{j,X} + c_{j,X^{n-1}}|^2 }{\|c_j\|_2^2}.
    \]
    Very similar calculations yield that 
    \[
        \mcP_{\.s}(O) = \frac{1}{n} \|\Pi_{\.s} ZZ\|_F^2 = \frac{1}{n} \sum_{j=1}^s \frac{|\langle ZZ,{B}_j\rangle |^2}{\|B_j\|_F^2}  = \frac{1}{n} \sum_{j=1}^s \frac{|\langle ZZ, c_{j,ZZ} ZZ \rangle |^2}{\|c_j\|_2^2 n2^n} = 2^n \sum_{j=1}^s \frac{|c_{j,ZZ}|^2}{\|c_j\|_2^2 }.
    \]
 
\end{proof}
\begin{thm}\label{thm:c-purity-Cn}
    The purity of the initial state $\rho$ and the measurement $O$ in the center $\.c$ of $\.g_{C_n}$ is 
    \[\mcP_{\.c}(\rho) = 
    \parity(n)/2^{n-1} \quad \text{and} \quad \mcP_{\.c}(O) = 2^n/n,\]
    where $\parity(n) = 1$ if $n$ is odd and 0 if $n$ is even.
\end{thm}
\begin{proof}
    Recall $c_1$ and $c_2$ from Theorem \ref{thm:cycle-center} such that $\.c=\spn_\mb{R}\{c_1, c_2\}$:
     \begin{align*}
    c_1 = - X +\sum_{t=1}^{\frac{n-1}{2}} \big({ZX^{2t-1}Z} + {YX^{2t-1}Y}\big),\text{~~and~~}
     c_2 = {X^{n-1}} +\sum_{t=0}^{\frac{n-3}{2}}\big({ZX^{2t}Z} + {YX^{2t}Y}\big)
\end{align*}
if $n$ is odd, and 
\begin{align*}
    c_1  = {X^{n-1}} - X + \sum_{t=1}^{\frac{n-2}{2}}\big({ZX^{2t-1}Z} + {YX^{2t-1}Y}\big),\text{~~and~~} c_2  = \sum_{t=0}^{\frac{n-2}{2}} \big({ZX^{2t}Z} + {YX^{2t}Y}\big)
\end{align*}
if $n$ is even.
 
    Note that $c_1$ and $c_2$ both have $n$ orbit-sums all with $\pm 1$ coefficients. Thus $\|c_1\|_2^2 = \|c_2\|_2^2 = n$.
    
    \paragraph{$n$ odd:} $c_1$ has $X$ and no  $X^{n-1}$ or $ZZ$; $c_2$ has $X^{n-1}$ and  $ZZ$ but no $X$. By Lemma \ref{lem:Cn-purity}, we have
    \[
        \mcP_{\.c}(\rho) = \frac{n}{2^n} \left(\frac{1}{n} + \frac{1}{n}\right) = \frac{1}{2^{n-1}}, \quad \text{and} \quad \mcP_{\.c}(O) = 2^n\left (0+\frac{1}{n}\right) = \frac{2^n}{n}.
    \]

    \paragraph{$n$ even:} $c_1$ has a component of $X^{n-1}-X$ (which is orthogonal to $X+X^{n-1}$) and no component of $ZZ$; $c_2$ has a component  of $ZZ$ and no component of $X$ or $X^{n-1}$. Again, by Lemma \ref{lem:Cn-purity}:
    \[
        \mcP_{\.c}(\rho) = \frac{n}{2^n} \left(0+0\right) = 0, \quad \text{and} \quad \mcP_{\.c}(O) = 2^n\left (0+\frac{1}{n}\right) = \frac{2^n}{n}.
    \]
\end{proof}

\begin{thm}\label{thm:Cn-purity-gj}
    The purity of the initial state $\rho$ and the measurement $O$ in the subspace $\.g_k$ of $\.g_{C_n}$ is 
    \[\mcP_{\.g_k}(\rho) =
    \frac{\parity(k)}{2^{n-2}}, \quad \text{ and } \quad \mcP_{\.g_k}(O) = \frac{2^{n}}{n}.\]
\end{thm}
\begin{proof}
    The subspace $\.g_k = \spn_{\mbR}\{\tilde X_k, \tilde Y_k, \tilde Z_k\}$ as defined in Theorem~\ref{thm:cycle-iso-su2}. By Lemma \ref{lem:Cn-purity} and some basic trigonometrical identities,
    we can see that the purity of $\rho$ is 
    \begin{align*}
        \mcP_{\.g_k}(\rho) & = \frac{n}{2^n} \cdot \frac{\lp\sin\frac{k\pi}{n} + \sin \frac{k(n-1)\pi}{n}\rp^2}{\sum_{j=0}^{n-1}\lp \sin^2 \frac{k(j+1)\pi}{n} + \sin^2 \frac{kj\pi}{n}\rp } + \frac{n}{2^n} \cdot \frac{\lp\cos\frac{k(n-1)\pi}{n} - \cos \frac{k\pi}{n}\rp^2}{\sum_{j=0}^{n-1}\lp \cos^2 \frac{k(j+1)\pi}{n} + \cos^2 \frac{kj\pi}{n}\rp} \\
        &=\frac{n}{2^n} \cdot \frac{\lp 2\sin \frac{k\pi}{2}\cos \frac{k(n-2)\pi}{2n}\rp^2}{\sum_{j=0}^{n-1}\lp 2\sin^2 \frac{kj\pi}{n}\rp } + \frac{n}{2^n} \cdot \frac{\lp -2 \sin \frac{k\pi}{2}\sin \frac{k(n-2)\pi}{2n}\rp^2}{\sum_{j=0}^{n-1}\lp 2\cos^2 \frac{kj\pi}{n}\rp }\\
        &=\frac{n}{2^n} \cdot \frac{4\sin^2\frac{k\pi}{2}\cos^2 \frac{k(n-2)\pi}{2n}}{\sum_{j=0}^{n-1}\lp 1-\cos \frac{2kj\pi}{n}\rp } + \frac{n}{2^n} \cdot \frac{4\sin^2\frac{k\pi}{2}\sin^2 \frac{k(n-2)\pi}{2n}}{\sum_{j=0}^{n-1}\lp \cos \frac{2kj\pi}{n}+1\rp }\\
        &=\frac{n}{2^n} \cdot \frac{4 \sin^2\frac{k\pi}{2}\cos^2 \frac{k(n-2)\pi}{2n}}{n} + \frac{n}{2^n} \cdot \frac{4\sin^2\frac{k\pi}{2}\sin^2 \frac{k(n-2)\pi}{2n}}{n }\\
        & = \frac{\sin^2\frac{k\pi}{2}}{2^{n-2}} = \frac{\parity(k)}{2^{n-2}},
    \end{align*}
    and the purity of $O$ is 
    \begin{multline}
        \mcP_{\.g_k}(O) = \frac{2^n}{\sum_{j=0}^{n-1}\lp \cos^2 \frac{k(j+1)\pi}{n} + \cos^2 \frac{kj\pi}{n}\rp } \\= \frac{2^n}{\sum_{j=0}^{n-1}\lp 2 \cos^2 \frac{kj\pi}{n}\rp}=\frac{2^n}{\sum_{j=0}^{n-1}\lp  \cos \frac{2kj\pi}{n}+1\rp}
        =\frac{2^n}{n}.
    \end{multline}
\end{proof}

We are now in a position to prove Theorem~\ref{thm:Cn-exp-var}.

\CnExpVar*

\begin{proof}
    Let $\{\tilde{c}_j=c_j/\sqrt{n^2 2^n}: j=1,2\}$ denote the orthonormal basis of $\.c$, where $c_1,c_2$ are defined in Theorem \ref{thm:cycle-center}. The $\rho_{\.c}$ and $O_{\.c}$ are $\rho_{\.c}=\langle \tilde{c}_1,\rho \rangle  \tilde{c}_1+\langle \tilde{c}_2,\rho \rangle  \tilde{c}_2$ and  $O_{\.c}=\langle \tilde{c}_1, O \rangle  \tilde{c}_1+\langle \tilde{c}_2, O \rangle  \tilde{c}_2$. Since $\tilde{c}_1$ and $\tilde{c_2}$ consist of distinct Pauli orbit-sums, we have
     \begin{align*}
        \tr(\rho_{\.c}O_{\.c})&=\langle \tilde{c}_1,\rho \rangle  \langle \tilde{c}_1,O \rangle \tr((\tilde{c}_1)^2)+\langle \tilde{c}_2,\rho \rangle  \langle \tilde{c}_2,O \rangle \tr((\tilde{c}_2)^2) \\ &=\frac{\langle c_1,\rho \rangle  \langle c_1,O \rangle}{(n^2 2^n)^2} \tr(c_1^2)+\frac{\langle c_2,\rho \rangle  \langle c_2,O \rangle}{(n^2 2^n)^2} \tr(c_2^2).
     \end{align*}
    If $n$ is odd, we can compute that $\langle c_1, O\rangle=0$ and
    \begin{align*}
          \langle c_2, \rho \rangle & =  \langle X^{n-1}, -\frac{i}{2^n}X^{n-1} \rangle=-in, \quad
          \langle c_2, O\rangle  =  \langle ZZ, \frac{-i}{\sqrt{n}}ZZ \rangle=-i\sqrt{n}2^n,\quad 
          \tr(c_2^2)=-n^2 2^n.
    \end{align*}
    Therefore, $\tr(\rho_{\.c}O_{\.c})=\frac{1}{\sqrt{n}}$ if $n$ is odd . If $n$ is even, $\langle c_1,O\rangle=\langle c_2,\rho\rangle=0$, which implies $\tr(\rho_{\.c}O_{\.c})=0$.
    Directly computing Eq.~\eqref{eq:VQA-exp-var} gives
    \[
        \av_\theta[\ell(\rho,{O};\theta)] = \tr(\rho_{\.c}O_{\.c})=\frac{\parity(n)}{\sqrt{n}}.
    \]
    Similarly, combining Theorem \ref{thm:Cn-purity-gj} and Eq.~\eqref{eq:VQA-exp-var} gives
    \[
        \var_\theta[\ell(\rho,{O};\theta)] =\sum_{k=1}^{n-1}\frac{ \mcP_{\.g_k}(\rho)\mcP_{\.g_k}(O)}{\dim(\.g_k)}
        = \sum_{k=1}^{n-1}\frac{\frac{\parity(k)
        }{2^{n-2}} \frac{2^n}{n}}{3} 
        = \frac{4}{3n} \sum_{k=1}^{n-1} \parity(k)
        =  \frac{2(n-\parity(n))}{3n}.
    \]
\end{proof}

\section{Complete Graphs - Properties of Pauli Orbit-sums}

Here we give proofs of Lemma~\ref{lem:adzz-pqr}, as well as Facts~\ref{fact:Kn-01o-11o-10e}, \ref{fact:anyoddodd-oddeveneven}, \ref{fact:linear-independent}. In this section $\.g$ will be used to denote $\.g_{K_n}$. We first prove the following additional facts (Facts~\ref{fact:Kn-p1o-to-poo}-\ref{fact:Kn-p0e-induction}) which will be needed in this and other sections. 

\begin{fact}[``$*1o \Rightarrow \ast oo$'' ]\label{fact:Kn-p1o-to-poo}
    {For any fixed $p$, if $X^pY^1Z^{r}\in [\.g,\.g]$ for all odd $r$, then $X^{p}Y^{q}Z^{r}\in [\.g,\.g]$ for all odd $q$ and odd $r$. }
\end{fact}
\begin{proof}
    Fix an arbitrary $p$. We will prove the conclusion by induction on odd $q$. The base case is just the assumption that $X^pY^1Z^{r} \in [\.g,\.g] $ for all odd $r$. Now, for any odd $q\ge 3$, if $X^pY^{q'}Z^{r}\in [\.g,\.g]$ for all odd $r$ and all $q'\le q-2$, by Eq. \eqref{eq:Kn-adA}, we have
    \begin{align}\label{eq:Kn-induction-q}
        \underbrace{[X^1,[X^1,X^{p}Y^{q-2}Z^{r+2}]]}_{\in [\.g,\.g]} =\gamma_1 \underbrace{X^pY^{q-4}Z^{r+4}}_{\in [\.g,\.g] }+\gamma_2 \underbrace{X^pY^{q-2}Z^{r+2}}_{\in [\.g,\.g]}+\gamma_3 X^{p}Y^{q}Z^{r},
    \end{align}
    where $\gamma_1=4(r+4)(r+3)$, $\gamma_2=-4(r+3)(q-1)-4(q-1)(r+2)$ and $\gamma_3=4q(q-1)\neq 0$. The LHS and the first two terms on the RHS are all in $[\.g,\.g]$, thus $X^{p}Y^{q}Z^{r}$ must therefore also be in $[\.g,\.g]$.
\end{proof}
\begin{fact}[``$o0e \Rightarrow oee$'']\label{fact:Kn-o0e-to-oee}
    {For any fixed odd $p$, if $X^pZ^{r}\in \.g$ for all even $r$, then $X^{p}Y^{q}Z^{r}\in \.g$ for all even $q$ and even $r$.}
\end{fact} 
\begin{proof}
    Fix an arbitrary odd $p$. We will prove the conclusion by induction on even $q$. The base case is just the assumption that $X^p Z^{r} \in\.g$ for all even $r$. For any even $q\ge 2$, assume that $X^pY^{q'}Z^{r}\in\.g$ all for all even $r$ and all $q'\le q-2$.
    Again Eq.~\eqref{eq:Kn-induction-q} implies that $X^{p}Y^{q}Z^{r}\in \.g$.
\end{proof}

\begin{fact}\label{fact:Kn-200-020}
    When the number of vertices $n\ge 3$ is odd, $Y^2, X^2\in \.g$.
\end{fact}
\begin{proof}
    From Fact \ref{fact:Kn-01o-11o-10e}, $Y^1Z^1\in\.g$ and therefore $[X^1,Y^1Z^1]\in\.g$. Since $[X^1,Y^1Z^1]=-4Z^2+4Y^2$, we have $ Y^2\in\.g$. Again, from Fact \ref{fact:Kn-01o-11o-10e}, $X^1Y^1Z^r\in \.g$ for all odd $r$. Therefore, $[Z^2,X^1Y^1Z^r]\in \.g$. Then, using Eq.  \eqref{eq:Kn-adB} we have
    \begin{align*}
        & [Z^2,X^1Y^1Z^1]=4(n-2)(X^2-Y^2) + 8(X^2Z^2-Y^2Z^2)\in\.g,\\
        & [Z^2,X^1Y^1Z^{r}]=4(n-r-1)(X^2Z^{r-1}-Y^2Z^{r-1})+4(r+1)(X^2Z^{r+1}-Y^2Z^{r+1})\in\.g,\\
        &[Z^2,X^1Y^1Z^{n-2}] = 4(X^2Z^{n-3}-Y^2Z^{n-3})\in\.g,
    \end{align*}
    where $3\le r \le n-4$ and $r$ is odd.
    Substituting the last equation into the second one for $r=n-4$, we see that $X^2Z^{n-5}-Y^2Z^{n-5}\in\.g$. Substituting this in the second equation for $r=n-6$, and so on, until finally we get $X^2Z^2-Y^2Z^2\in\.g$. Now the first equation implies that $X^2-Y^2\in\.g$. As $Y^2\in\.g$, it follows that $X^2\in\.g$.
\end{proof}

\begin{fact}[``$*1o$ induction on $p$'']\label{fact:Kn-p1o-induction}
     For any fixed $p$, if $X^{p'}Y^1Z^{r}\in [\.g,\.g]$ for all odd $r$ and all $p'\le p-1$, then $X^{p}Y^1Z^r\in [\.g,\.g]$ for all odd $r$. 
\end{fact}
\begin{proof}
    Fix an arbitrary $p$, and assume $X^{p'}Y^1Z^{r'}\in [\.g,\.g]$ for all odd $r'$ and all $p'\le p-1$. From Fact \ref{fact:Kn-p1o-to-poo}, $X^{p'}Y^3Z^{r'}\in [\.g,\.g]$ for all odd $r'$ and all $p'\le p-1$.  
    We prove the claim by induction on $r$. For the base case $r=1$, 
    \begin{multline*}
        \underbrace{[X^1, [Z^2, X^{p-1}Y^1Z^1]]}_{\in [\.g,\.g]}=8p X^pY^1Z^1 + 8(n-p) \underbrace{X^{p-2}Y^1Z^1}_{\in[\.g,\.g]}+48\underbrace{X^{p-2}Y^1Z^3}_{\in [\.g,\.g] }\\-48 \underbrace{X^{p-2}Y^{3}Z^1}_{\in [\.g,\.g]} \Rightarrow X^pY^1Z^1\in [\.g,\.g],
    \end{multline*}
    where all the ``$\in [\.g,\.g]$'' use the assumption that $X^{p'} Y^1 Z^{r'}\in [\.g,\.g]$ {and the fact $X^{p'}Y^3Z^{r'}\in [\.g,\.g]$} for all odd $r'$ and all $p'\le p-1$.
    
    For the induction step, assume that $X^{p}Y^1Z^{r-2}\in [\.g,\.g]$. Then 
    \begin{multline*}
        \underbrace{[X^1,[Z^2,X^{p-1}Y^1Z^r]]}_{\in [\.g,\.g]} = \gamma_1'\underbrace{X^pY^1Z^{r-2}}_{\in [\.g,\.g]}  + \gamma_2'X^{p}Y^1Z^r\\ + \gamma_3'\underbrace{X^{p-2}Y^1Z^r}_{\in [\.g,\.g] } + \gamma_4'\underbrace{X^{p-2}Y^3Z^{r-2}}_{\in [\.g,\.g] } + \gamma_5'\underbrace{X^{p-2}Y^1Z^{r+2}}_{\in [\.g,\.g] } + \gamma_6'\underbrace{X^{p-2}Y^3Z^r}_{\in [\.g,\.g]}, 
    \end{multline*}
    which implies that $X^p Y^1 Z^{r}\in [\.g,\.g]$. Here $\gamma_1'=4(n-r-p+1)p$, $\gamma_2'=4(r+1)p$, $\gamma_3'=8(n-r-p+1)r$, $\gamma_4'=-24(n-r-p+1)$, $\gamma_5'=8(r+1)(r+2)$ and $\gamma_6'=-24(r+1)$. All the ``$\in [\.g,\.g]$'' use one of the assumptions and fact: (1) $X^{p'} Y^1 Z^{r'}\in [\.g,\.g]$ for all $p'<p$ and all odd $r'$, (2) $X^{p} Y^1 Z^{r''}\in [\.g,\.g]$ for all $r''<r$, { and (3)$X^{p'} Y^3 Z^{r'}\in [\.g,\.g]$ for all $p'<p$ and all odd $r'$.}
\end{proof}

\begin{fact}[``$*0e$ induction on $p$'']\label{fact:Kn-p0e-induction}
    For even $n$ and any fixed even $p$, if $X^pY^1Z^{r'}\in \.g$ for all odd $r'$, and $X^{p'}Z^{r''}\in \.g$ for all even $r''$, all odd $p' \le p-1$, then $X^{p+1}Z^r\in \.g$ for all even $r$. 
\end{fact}
\begin{proof}
    If $X^{p'}Z^{r''}\in \.g$ for all even $r''$, by Fact \ref{fact:Kn-o0e-to-oee} $X^{p'}Y^{2}Z^{r''}\in \.g$ for {all even $r''$}. Then, for all odd $r'$, if $1\le r' \le n-p-3$, we have
    \begin{equation*}
    \underbrace{[Z^2,X^pY^1Z^{r'}]}_{\in\.g}=\gamma_1^{r'}X^{p+1}Z^{r'-1} + \gamma_2^{r'} X^{p+1}Z^{r'+1} + \gamma_3^{r'} \underbrace{X^{p-1}Y^2Z^{r'-1}}_{\in\.g} + \gamma_4^{r'}\underbrace{X^{p-1}Y^{2}Z^{r'+1}}_{\in\.g},
    \end{equation*}
    and if $r'=n-p-1$, we have
    \begin{equation*}
        \underbrace{[Z^2,X^pY^1Z^{r'}]}_{\in\.g}=\gamma_1^{n-p-1}X^{p+1}Z^{n-p-2}+ \gamma_2^{n-p-1}\underbrace{X^{p-1}Y^2Z^{n-p-2}}_{\in\.g},
    \end{equation*}
    where 
    \begin{equation*}
        \gamma_1^{r'}=2(n-r'-p)(p+1), \quad \gamma_{2}^{r'}=2(r'+1)(p+1), \quad \gamma_3^{r'}=-4(n-r'-p), \quad  \gamma_4^{r'}=-4(r'+1)
    \end{equation*}
    for $1\le r' \le n-p-3$ and $\gamma_1^{n-p-1}= 2(p+1)$, $\gamma_2^{n-p-1}=-4$.
    Note that the case of $r' = n-p-1$ implies that $X^{p+1}Z^{n-p-2}\in \.g$. {(This case can happen as the parity of $p$, $n$ and $r'$ match.)} Applying this to the case of $r' = n-p-3$ gives $X^{p+1}Z^{n-p-4}\in \.g$. Again applying it in turn to the case of $r'=n-p-5$ gives $X^{p+1}Z^{n-p-6}\in \.g$. Continuing this process gives $X^{p+1}Z^r \in\.g$ for any even $r$.
\end{proof}

We are now ready to prove the lemma and facts from the main text.

\AdZZPQR*

\begin{proof}
Eq.~\eqref{eq:Kn-adB} gives
    \begin{multline*}
        [Z^2,X^pY^qZ^r] = \alpha_1 X^{p+1}Y^{q-1}Z^{r-1}+\alpha_2 X^{p+1}Y^{q-1}Z^{r+1}   \\ - \alpha_3 X^{p-1}Y^{q+1}Z^{r-1}-\alpha_4X^{p-1}Y^{q+1}Z^{r+1}, 
    \end{multline*}  
    where $\alpha_1 = 2(n-p-q-r+1)(p+1)$, $\alpha_2=2(p+1)(r+1)$, $\alpha_3 = 2(n-p-q-r+1)(q+1)$, $\alpha_4 = 2(q+1)(r+1)$. Thus,
    \begin{align*}
        & \ (q-1)ad_{Z^2}(X^pY^qZ^{r}) - (r+1)ad_{Z^2}(X^{p}Y^{q-2}Z^{r+2}) =\\
         & \ 2(q-1) (n-p-q-r+1)(p+1) X^{p+1}Y^{q-1}Z^{r-1} + 2(q-1)(p+1)(r+1) X^{p+1}Y^{q-1}Z^{r+1}  \\
        & - 2(n-p-q-r+1)(q-1)(q+1) X^{p-1}Y^{q+1}Z^{r-1} - 2(q-1)(r+1)(q+1) X^{p-1}Y^{q+1}Z^{r+1} \\
        & - 2(r+1) (n-p-q-r+1)(p+1) X^{p+1}Y^{q-3}Z^{r+1} - 2(r+3)(p+1)(r+1) X^{p+1}Y^{q-3}Z^{r+3}  \\
        & + 2(n-p-q-r+1)(q-1)(r+1) X^{p-1}Y^{q-1}Z^{r+1} + 2(q-1)(r+1)(r+3) X^{p-1}Y^{q-1}Z^{r+3} 
    \end{align*}
    Rearranging these 8 terms (pairing up terms 1 and 5, 2 and 6, 3 and 7, 4 and 8) gives  the result.
\end{proof}

\ZeroOneOdd*

\begin{proof}
    We will show that $Y^1Z^{r-1}, X^1Y^1Z^{r-1}\in {[\.g,\.g]},~X^1Z^{r}\in\.g$ for all even $r$ by induction on $r$. 
    The proof is by induction on $r$. For the base case $r = 2$, we can easily verify  (by Eqs. \eqref{eq:Kn-adA} and \eqref{eq:Kn-adB}) that
    \begin{align*}
     & Y^1Z^1 = [X^1,Z^2]/2\in [\.g,\.g],\\
     & X^1Z^2 = [Z^2, Y^1Z^1]/4 - (n-1)X^1/2 \in \.g \\
     & X^1Y^1Z^1 = [X^1,X^1Z^2]/2\in [\.g,\.g].
    \end{align*}
    Now for any $r\ge 4$, assume that $Y^1Z^{r-3},~X^1Y^1Z^{r-3}\in [\.g,\.g],~X^1Z^{r-2}\in \.g$. Then, 
    \begin{align*}
     & \underbrace{[Z^2,X^1Z^{r-2}]}_{\in [\.g,\.g]}=-2(n-r+2)\underbrace{Y^1Z^{r-3}}_{\in [\.g,\.g]}-2(r-1)Y^1Z^{r-1} \Rightarrow Y^1Z^{r-1}\in [\.g,\.g],\\
     & \underbrace{[Z^2,Y^1Z^{r-1}]}_{\in [\.g,\.g]}=2(n-r+1)\underbrace{X^1Z^{r-2}}_{\in\.g}+4rX^1Z^{r}\Rightarrow X^1Z^{r}\in \.g,\\
     & \underbrace{[X^1, X^1Z^{r}]}_{\in [\.g,\.g]}=2X^1Y^1Z^{r-1} \Rightarrow X^1Y^1Z^{r-1}\in [\.g,\.g].
    \end{align*}
\end{proof}

\StarOddOdd*

\begin{proof}
\begin{enumerate}
    \item []
    \item By Fact \ref{fact:Kn-01o-11o-10e}, $Y^1Z^r\in [\.g,\.g]$ for all odd $r$. Then Fact \ref{fact:Kn-p1o-induction} gives $X^pY^1Z^r\in [\.g,\.g]$ for all $p$ and all odd $r$. Finally, by Fact \ref{fact:Kn-p1o-to-poo}, $X^pY^qZ^r\in [\.g,\.g]$ for all $p$ and all odd $q$ and odd $r$. 
    
    \item By point 1 above, $X^p Y^1 Z^r\in \.g$ for all $p$ and all odd $r$. Fact \ref{fact:Kn-01o-11o-10e} gives $X^1 Z^r\in \.g$ for all even $r$. Repeatedly applying Fact \ref{fact:Kn-p0e-induction} then gives $X^p Z^r\in \.g$ for all odd $p$ and all even $r$ when $n$ is even. The result then follows from Fact \ref{fact:Kn-o0e-to-oee}.

    \item Fact \ref{fact:Kn-01o-11o-10e} gives $X^1 Z^r\in \.g$ for all even $r$. The result then follows from Fact \ref{fact:Kn-o0e-to-oee}.
\end{enumerate}
\end{proof}

\LinearIndependence*

\begin{proof}
    We show that if 
    \begin{equation}\label{eq:Kn-lin-indep}
    \sum_{p\ge 0,~q,r \text{ are odd}}\alpha_{p,q,r} [X^1, X^pY^qZ^r]+\sum_{p\ge 0,~r \text{ is odd}}\beta_{p,r} [Z^2, X^pY^1Z^r]=0,
    \end{equation}
then all coefficients $\alpha_{p,q,r}$ and $\beta_{p,r}$ are $0$. 
\begin{enumerate}
    \item \textbf{Regime 1 ($p\ge 0$, odd $q\ge 3$, odd $r$):} 
    We show this by induction on decreasing values of $q$, i.e., we show that if $\alpha_{p,q_{\max},r}=0$ for all $p,r$, then $\alpha_{p,q_{\max}-2,r}=0$ for all $p,r$, and so on. 
    The base case is $q=q_{\max} \defeq n-p-1-\parity(n-p)$, in which case $r=1$. 
    From Eq. \eqref{eq:Kn-adA}, only $[X^1,X^{p}Y^{q_{\max}}Z^1]$ can contribute to term $X^{p}Y^{q_{\max}+1}$, and thus Eq.~\eqref{eq:Kn-lin-indep} implies that $\alpha_{p,q_{\max},1}=0$. Now assume that $\alpha_{p,q',r}=0$  for all $p$, all odd $q'$ with $q+2\le q' \le q_{\max}$, and all odd $r$ with $1\le r \le n-p-q'$. The $\alpha_{p,q',r}[X^1,X^{p}Y^{q'}Z^{r}]$ does not contribute to any term as the coefficient is 0. By Eq. \eqref{eq:Kn-adB}, only $[X^1,X^{p}Y^{q}Z^{r}]$ can contribute to term $X^{p}Y^{q+1}Z^{r-1}$,  for any $p$, $q\ge 3$ and $1\le r \le n-p-q$. Therefore, Eq.~\eqref{eq:Kn-lin-indep} implies that $\alpha_{p,q,r}$, the coefficient for $X^{p}Y^{q}Z^{r}$, is 0.

    \item \textbf{Regime 2 ($p\ge 0$, $q=1$, odd $r$):} 
    This can be shown by induction on increasing value of $r$.  
    The base case is $r=1$.
    Since for all $p$, only $[Z^2, X^{p}Y^1Z^{1}]$ can contribute to term $X^{p+1}$, we see that $\beta_{p,1}=0$. 
    Since $\beta_{p,1}=0$, only $[X^1,X^{p}Y^1Z^1]$ contributes to term $X^{p}Y^2$, which implies $\alpha_{p,1,1}=0$ for all $p$. For all $p$, assume that $\alpha_{p,1,r'}=0$ and $\beta_{p,r'}=0$ for all $1\le r'\le r-2$. Since $\alpha_{p,1,r'}=0$ and $\beta_{p,r'}=0$, only $[Z^2,X^{p}Y^1Z^{r}]$ can contribute to term $X^{p+1}Z^{r-1}$. Therefore, $\beta_{p,r}=0$ for all $p$. Since $\beta_{p,r}=0$ for all $p$, only $[X^1, X^{p}Y^1Z^{r}]$ can contribute term $X^{p}Y^2Z^{r-1}$, which implies $\alpha_{p,1,r}=0$.
\end{enumerate}
\end{proof}

\section{Complete Graphs - Explicit basis for odd $n$}\label{sec:Kn-oddn}

Theorem~\ref{thm:Kn-basis-even} in the main text gives a basis for $\.g_{K_n}$ when $n$ is even.  Here we show a similar result for $n$ odd. In the remainder of this section, we will drop the subscript $K_n$ and simply use $\.g$ to denote $\.g_{K_n}$.
  
Define the following sets:
\begin{align*}
    & Q\defeq\{X^pY^qZ^r:\text{~$q$ and $r$ are odd}\}\cup \{X^1Y^qZ^r:\text{$q$ and $r$ are even}\}\cup \{X^2,Y^2,Z^2\},\\ 
    & U_1\defeq\{X^pY^qZ^r:\text{$q$ and $r$ are odd}\}\backslash (\{X^1Y^qZ^r: \text{$q$ and $r$ are odd}\} \cup \{Y^1Z^1\}),\\ 
    & U_2\defeq\{X^pY^1Z^r:\text{~$r$ is odd}\}\backslash (\{Y^1Z^r:\text{$r$ is odd}\} \cup \{X^1Y^1Z^{n-2}\}),\\ 
    & V_1\defeq\{[X^1, v]:~v\in U_1\}, \\ 
    & V_2\defeq\{[Z^2, v]:~v\in U_2\}. 
\end{align*}
Note that $U_1, U_2\subseteq Q$ and $U_1\cap U_2 \neq \emptyset$. 
\begin{thm}\label{thm:Kn-oddn-basis} 
    When $n$ is odd,
    \begin{equation*}
        R\defeq Q\cup V_1\cup V_2
    \end{equation*}
    is a basis for $\.g$, and $\dim(\.g) = \frac{1}{12}\lp n^3+6n^2-n+18\rp$.
\end{thm}
\begin{proof}
Let $\.r\defeq\spn_{\mathbb{R}} R$. We will show that (1) $|R| = \frac{n^3+6n^2-n+18}{12}$, (2) the basis vectors in $R$ are linearly independent, (3) $\.g\subseteq \.r$, (4) $\.r\subseteq \.g$.
\end{proof}

\paragraph{(1) Linear independence of $R$.}  
{Pauli string orbits corresponding to distinct value of $(p,q,r)$ are mutually orthogonal, and therefore vectors in $Q$ and mutually orthogonal, and vectors in $\{X^pY^qZ^r:\text{~$q$ and $r$ are odd}\}\subseteq Q$ are linearly independent from, and orthogonal to, those in $V_1\cup V_2$ (which are orbits consisting of even numbers of both $Y$ and $Z$ operators).} We now show that the vectors in $\tilde{Q}\cup \{X^2,Y^2,Z^2\} \cup V_1\cup V_2$ are also linearly independent, {and $|\tilde{Q}\cup \{X^2,Y^2,Z^2\} \cup V_1\cup V_2| = |\tilde Q| + \abs{V_1}+\abs{V_2} + |\{X^2,Y^2,Z^2\}|$}, where $\tilde{Q}:=\{X^1Y^qZ^r:\text{$q$ and $r$ are even}\}$. To see this, note that if
\begin{multline*}
    \sum_{(q,r)\neq (1,1)\atop q, r \text{ are odd}}\alpha_{0,q,r}[X^1,Y^qZ^r]+\sum_{p\ge 2,\atop q, r \text{ are odd}}\alpha_{p,q,r} [X^1, X^pY^{q}Z^r]+\sum_{q, r \text{ are even}}\beta_{q,r} X^1Y^{q}Z^r\\+\sum_{r\le n-4,\atop r\text{ is odd}}\gamma_{1,r}[Z^2,X^1Y^1Z^r]+\sum_{p\ge 2 \atop r\text{ is odd}}\gamma_{p,r}[Z^2,X^pY^1Z^r]+\delta_1 X^2+\delta_2 Y^2+\delta_3 Z^2=0,
\end{multline*}
then 

\begin{enumerate}
    \item $\alpha_{0,q,r}=0$ (for $p=0$, odd $q\ge 3$, odd $r$) and $\alpha_{p,q,r}=0$ (for $p\ge 2$, odd $q\ge 3$, odd $r$) by the same argument as in regime 1 of the proof of Fact \ref{fact:linear-independent}.

    \item $\alpha_{0,1,r}=0$ (for $p=0$, $q=1$, odd $r\ge 3$) since only $[X^{1},Y^1Z^r]$ can contribute term $Z^{r+1}$.

    \item $\beta_{q,r}=0$ (for even $q\ge 4$, even $r$) since only basis vector $X^1Y^qZ^r$ contributes term $X^1Y^qZ^r$ for all even $q\ge 4$ and all even $r$.

    \item $\gamma_{1,r}=0$ (for $p=1$, odd $r \le n-4$) We show this by induction on decreasing values of $r$. The base case is $r=n-4$. Since only $[Z^2,X^1Y^1Z^{n-4}]$ can contribute term $Y^2Z^{n-3}$, then $\gamma_{1,n-4}=0$. Assume that $\gamma_{1,r'}=0$ for all odd $r'$, $r+2\le r'\le n-4$. Since $\gamma_{1,r'}=0$, only $[Z^2,X^1Y^1Z^r]$ can contribute term $Y^2Z^{r+1}$ and thus $\gamma_{1,r}=0$.

    \item $\delta_j=0$ ($j=1,2,3$) since, with $\gamma_{1,r}=0$ for $p=1$, odd $r\le n-4$ (as shown above), only vectors $X^2$, $Y^2$ and $Z^2$ can contribute terms $X^2$, $Y^2$ and $Z^2$.

    \item $\alpha_{p,1,r}=0$ and $\gamma_{p,r}=0$ (for $p\ge 2$, $q=1$, odd $r$) by the same argument as in regime 2 of the proof of Fact \ref{fact:linear-independent}.

    \item $\beta_{q,r}=0$ (for $q=0,2$, even $r$) since , with $\gamma_{p,r}=0$ for $p\ge 2$, $q=1$, $r$ odd (as shown above), only basis vector $X^1Y^qZ^r$ can contribute $X^1Y^qZ^r$ when $q=0,2$ and $r$ is even.

\end{enumerate}
i.e., all coefficients $\alpha_{p,q,r}$, $\beta_{p,r}$, $\gamma_{p,r}$ and $\delta_1,\delta_2,\delta_3$ are $0$.

\paragraph{(2) $|R|=\frac{1}{12}\lp n^3 + 6n^2 - n + 18\rp$.} The above proof of linear independence of $R$ also shows that $\abs{R}= \abs{Q} + \abs{V_1} + \abs{V_2}$. The cardinality of $R$ then follows from the fact that 
\begin{align*}
    &|Q|=3+\sum_{t=1}^{(n-1)/2}t(n-2t+1)+\sum_{t=0}^{(n-1)/2}(t+1)=\frac{n^3+6n^2+11n+78}{24},\\
    &|U_1|=|V_1|=\sum_{t=1}^{(n-1)/2}t(n-2t+1)-\sum_{t=1}^{(n-1)/2}t-1=\frac{n^3-n-24}{24},\\
    &|U_2|=|V_2|=\sum_{t=1}^{(n-1)/2}(n-2t+1)-\frac{n-1}{2}-1=\frac{n^2-2n-3}{4},
\end{align*}

\paragraph{(3) $\.g\subseteq \.r$.} 
As in the even $n$ case, we will show that $\.r$ is closed under both $ad_{X^1}$ and $ad_{Z^2}$. Since both generators $X^1,Z^2\in Q\subseteq R$, we have that $\.g\subseteq \.r$. 

We summarize the proof that $\.r$ is closed under actions $ad_{X^1}$ and $ad_{Z^2}$ in Tables \ref{table:Kn-odd-adA-on-basis} and \ref{table:Kn-odd-adB-on-basis}. In these tables, the value of an entry, for instance, $V_{1}$ for entry $(ad_{X^1},U_{1})$, means that $[X^1,P]\in \spn_{\mbR} V_{1}$ for any $P\in U_{1}$.  Unlike the even $n$ case, here $ad_X$ and $ad_Z$ impose different structures on the commutator tables, and we need to consider two different decompositions of $Q$, 
\begin{equation}\label{eq:Q_decompse}
    Q = U_1 \cup Q' \cup Q'' = U_2 \cup Q_1 \cup Q_2 \cup Q_3 \cup Q_4
\end{equation}
where 
\[
    Q' = \{X^1Y^qZ^r:~q+r\text{~is even}\}, \quad Q'' = \{Y^1Z^1,X^2,Y^2,Z^2\}
\]
and 
\begin{align*}
    Q_1 & = \{X^pY^qZ^r:\text{$q$ and $r$ are odd, $q\ge 3$}\}\cup \{X^1Y^1Z^{n-2}\}, \quad 
    & Q_2 = \{Y^1Z^r:\text{$r$ is odd}\}, \\
    Q_3 & = \{X^1Y^qZ^r:\text{$q$ and $r$ are even}\}, \quad 
    & Q_4 = \{X^2,Y^2,Z^2\}.\quad \quad 
\end{align*}

\begin{table}[]
\caption{Adjoint map on basis $Q\cup V_1\cup V_2$ of $\.r$ when $n$ is odd, where $Q=U_1\cup Q'\cup Q''$ (Eq. \eqref{eq:Q_decompse}).}
    \label{table:Kn-odd-adA-on-basis}
    \centering{
    \begin{tabular}{c| c c c| c c}
    & $U_1$ & $Q'$ & $Q''$ & $V_1$ & $V_2$\\
    \hline
    $(p,q,r)$ & $*oo-011-1oo$ & $1ee\cup 1oo$ & $011\cup 200\cup 020\cup 002$ & $\subseteq *ee$ & $\subseteq *ee$\\
    \hline
    \multirow{2}*{$ad_{X^1}$} & $V_1$ & $Q'\subseteq Q$ & $Q''\subseteq Q$ & $Q$ &$Q$\\
    & {\color{blue} (by def)} & {\color{blue} (parity)} & {\color{blue} (computation)}  & {\color{blue} (parity)}  & {\color{blue} (parity)} 
    \end{tabular}
    }
    
\end{table}
\begin{table}[]    
\caption{Adjoint map on basis $Q\cup V_1\cup V_2$ of $\.r$ when $n$ is odd, where $Q = U_2\cup \bigcup_{k=1}^4 Q_k$ (Eq. \eqref{eq:Q_decompse}).}
    \centering\resizebox{\textwidth}{!}{
    \begin{tabular}{c|c c c c c| c c}
         &  $U_2$ & $Q_1$ & $Q_2$& $Q_3$& $Q_4$ & $V_1$ & $V_2$\\
         \hline
        $(p,q,r)$  & $p(>0)1o-11(n-2)$ & $p(\ge 3)oo\cup 11(n-2) $ &$01o$ & $1ee$ & $200\cup 020\cup 002$ & $\subseteq *ee$ & $\subseteq *ee$\\
        \hline
        \multirow{2}*{$ad_{Z^2}$} & $V_2$ & $V_1\cup V_2\cup Q$ & $Q_3 \subseteq Q$ & $Q$& $U_2\subseteq Q$ & $Q$& $Q$\\
        & {\color{blue} (by def)} & {\color{blue}(complicated)} & {\color{blue} (computation)} &{\color{blue} (parity)} & {\color{blue} (computation)} & {\color{blue} (parity)} & {\color{blue} (parity)}
    \end{tabular}}
    \label{table:Kn-odd-adB-on-basis}
\end{table}
Again, the parentheses in the tables indicate the proof method:
\begin{itemize}
    \item ``by def'': $ad_{X^1}(U_1) = V_1$ and $ad_{Z^2}(U_2) = V_2$ hold by the definitions of $V_1$ and $V_2$. 
    \item ``parity'': for instance, $ad_{X^1}(Q') \subseteq \spn_{\mathbb R} Q$ because the value/parities requirement for $(p,q,r)$ for $Q'$ is $p=1$ and $q+r$ is even. Recall that $ad_{X^1}$ changes the parities of $Y$ and $Z$ and keeps the exponent of $X$ unchanged. Thus, $ad_{X^1}(Q) \subseteq Q$. Other ``parity'' entries follow in a similar way. 
    \item ``computation'': these follow from direct computation of specific elements: $ad_{X^1}(Y^1Z^1)=4Y^2-4Z^2$, $ad_{X^1}(X^2)=0$, $ad_{X^1}(Z^2)=2YZ$ and $ad_{X^1}(Y^2)=-2YZ$. Therefore,  $ad_{X^1}(Q'')\subseteq \spn_{\mbR} Q''$.  
    For $ad_{Z^2}(Q_2)$, take  any odd $r$, $ad_{Z^2}(Y^1Z^r)=2(n-r)X^1Z^{r-1}+2(r+1)X^1Z^{r+1}\in \spn_{\mbR} Q_3$. 
    For $ad_{Z^2}(Q_4)$: $ad_{Z^2}(X^2)=-2X^1Y^1Z^1\in\spn_{\mbR} U_2$, $ad_{Z^2}(Y^2)=2X^1Y^1Z^1\in\spn_{\mbR} U_2$, $ad_{Z^2}(Z^2)=0$. Therefore, $ad_{Z^2}(Q_4)\subseteq \spn_\mbR U_2$. 
    \item ``complicated'': $ad_{Z^2}(Q_1)$ is the only complicated case, which will be proved to be in $\spn_{\mathbb R} (V_1\cup V_2\cup Q)$. For any $X^pY^qZ^r\in Q_1$ (i.e., $q,r$ odd,  $q\ge 3$), by Lemma~\ref{lem:adzz-pqr} we have
            \begin{multline}\label{eq:Kn-adB-Q1-induction}
           ad_{Z^2}(X^pY^qZ^{r})=\beta_1\underbrace{ad_{X^1}(X^{p+1}Y^{q-2}Z^{r})}_{\in \spn_\mbR V_1\cup Q}+\beta_2\underbrace{ad_{X^1}(X^{p+1}Y^{q-2}Z^{r+2})}_{\in \spn_\mbR V_1\cup Q}\\
           -\beta_3\underbrace{ad_{X^1}(X^{p-1}Y^{q}Z^{r})}_{\in \spn_\mbR V_1\cup Q}-\beta_4\underbrace{ad_{X^1}(X^{p-1}Y^{q}Z^{r+2})}_{\in \spn_\mbR V_1\cup Q}+\beta_5 ad_{Z^2}(X^{p}Y^{q-2}Z^{r+2}).
       \end{multline}
where $\beta_1,\ldots, \beta_5\in\mb{R}$.
       The first 4 terms in Eq. \eqref{eq:Kn-adB-Q1-induction} are all in $\spn_\mbR V_1\cup Q$ because the orbits inside the $ad_{X^1}(\cdot )$ are all in $Q$ and $ad_{X^1}(Q)\subseteq \spn_\mbR (V_1\cup Q)$ (Table \ref{table:Kn-odd-adA-on-basis}). 
       The last term in Eq.~\eqref{eq:Kn-adB-Q1-induction} has $Y$'s exponent decreased from $q$ to $q-2$. Apply Eq.~\eqref{eq:Kn-adB-Q1-induction} again to this last term, and continue this process until $Y$'s exponent drops to 1, yielding an expression for $ad_{Z^2}(X^pY^qZ^{r})$ as $4(q-1)/2 = 2(q-1)$ terms in $\spn_{\mbR}V_1\cup Q$, plus a final term of the form $ad_{Z^2}(X^{p}Y^{1}Z^{q+r-1})$. 
       
       We now show that $ad_{Z^2}(X^{p}Y^{1}Z^{q+r-1})\in \spn_{\mbR} V_2\cup Q$, and thus $ad_{Z^2}(Q_1)\subseteq \spn_\mbR  V_1\cup V_2\cup Q$.
       
       Note that $q+r-1$ is odd. There are two cases to consider.
       
           \noindent\textbf{Case 1:} $(p,q+r-1)\neq (1,n-2)\Ra X^p Y^1 Z^{q+r-1}\in U_2\cup Q_2$. From Table \ref{table:Kn-odd-adB-on-basis}, $ad_{Z^2}(Q_2)\subseteq \spn_{\mbR} Q$ and,  by definition, $ad_{Z^2}(U_2)=V_2$. 

           \noindent\textbf{Case 2:} $(p,q+r-1) = (1,n-2)$. To compute $ad_{Z^2}(X^1 Y^1 Z^{n-2})$ we make use of the fact that, by definition, for $1\le r \le n-4$,  $r$ odd, $ad_{Z^2}(X^1Y^1Z^{r})\in V_2$. Then, from Eq. \eqref{eq:Kn-adB}, the following equations hold:
           \begin{align*}
               & \underbrace{ad_{Z^2}(X^1Y^1Z^1)}_{\in V_2}=4(n-2)\underbrace{(X^2-Y^2)}_{\in \spn_\mbR Q}+8(X^2Z^2-Y^2Z^2) \\
               & \underbrace{ad_{Z^2}(X^1Y^1Z^{2t-1})}_{\in V_2}=4(n-2t)(X^2Z^{2t-2}-Y^2Z^{2t-2})+8t(X^2Z^{2t}-Y^2Z^{2t}), \\
               &ad_{Z^2}(X^1Y^1Z^{n-2})=4(X^2Z^{n-3}-Y^2Z^{n-3}),
           \end{align*}
           where $2\le t \le (n-3)/2$.
           Note that the third equation has only two terms on the RHS due to our convention that $X^p Y^q Z^r = 0$ if $p+q+r > n$. The first equation implies that $X^2Z^2-Y^2Z^2\in \spn_{\mbR} V_2\cup Q$. Applying this to the second equation with $t = 2$ gives $X^2Z^4-Y^2Z^4\in \spn_{\mbR} V_2\cup Q$. Applying this again to the second equation with $t=3$ gives $X^2Z^6-Y^2Z^6\in \spn_{\mbR} V_2\cup Q$. Continuing this process gives $X^2Z^{n-3}-Y^2Z^{n-3}\in \spn_{\mbR} V_2\cup Q$ and thus, by the third equation,  $ad_{Z^2}(X^1Y^1Z^{n-2})\in\spn_{\mbR} V_2\cup Q$.
\end{itemize}

\paragraph{(5) $\.r\subseteq \.g$.}
Recall that $R=Q\cup V_1\cup V_2.$ Since $V_1\defeq ad_{X^1}(U_1)$, $V_2\defeq ad_{Z^2}(U_2)$ and $U_1, U_2\subseteq Q$, then $Q\subseteq \.g\Ra R\subseteq \.g$. Recall that
\begin{equation*}
   Q\defeq \{X^pY^qZ^r:\text{~$q$ and $r$ are odd}\}\cup \{X^1Y^qZ^r:\text{$q$ and $r$ are even}\}\cup \{X^2,Y^2,Z^2\}.
\end{equation*}
Then,
\begin{itemize}
    \item By Fact~\ref{fact:anyoddodd-oddeveneven}, $X^pY^qZ^r\in \.g$ for all $p$ and all odd $q$ and odd $r$.

    \item By Fact~\ref{fact:anyoddodd-oddeveneven},  $\{X^1Y^qZ^r:\text{$q$ and $r$ are even}\}\subseteq\.g$.

    \item By Fact \ref{fact:Kn-200-020}, $X^2, Y^2\in\.g$. Since $Z^2$ is a generator, it is in $\.g$ as well. 
\end{itemize}

\section{Complete Graphs - Center and Semisimple Component}

In this section, we prove Theorem~\ref{thm:kn-semisimple-center}. In the remainder of this section, we will drop the subscript $K_n$ and simply use $\.g$ to denote $\.g_{K_n}$.

\begin{lem}\label{lem:adzz-pqr-commsubalgebra}
Let $L\defeq\{[X^1,X^pY^qZ^r], [Z^2,X^pY^1Z^r]:\text{$q$ and $r$ are odd}\}$, then $[Z^2,X^pY^qZ^r]\in \spn_{\mbR}L$ for any $p$ and any odd $q\ge 3$ and odd $r$.
\end{lem}
\begin{proof}
From Lemma~\ref{lem:adzz-pqr} it follows that
    \begin{align}\label{eq:Kn-Z2-induction}
        [Z^2,X^pY^qZ^{r}] = &\beta_1 ad_{X^1}(X^{p+1}Y^{q-2}Z^{r})+\beta_2 ad_{X^1}(X^{p+1}Y^{q-2}Z^{r+2}) \nonumber \\
        & -\beta_3 ad_{X^1}(X^{p-1}Y^{q}Z^{r})-\beta_4 ad_{X^1}(X^{p-1}Y^{q}Z^{r+2})+\beta_5 ad_{Z^2}(X^{p}Y^{q-2}Z^{r+2}).
    \end{align}
for some $\beta_1,\ldots,\beta_5\in\mb{R}$.
The first four terms are in $L$, and the last term has $Y$'s exponent decreased from $q$ to $q-2$. Apply Eq.~\eqref{eq:Kn-Z2-induction} again to this last term, and continue this process until $Y$'s exponent drops to 1, yielding an expression for $[Z^2, X^pY^qZ^{r}]$ as $4(q-1)/2 = 2(q-1)$ terms in $L$, plus a last term of form $[Z^2,X^{p}Y^{1}Z^{q+r-1}]$. Since $q+r-1$ is odd and the last term $[Z^2,X^{p}Y^{1}Z^{q+r-1}]\in L$.
\end{proof}

\KnSemiSimpleCenter*


\begin{proof}
Define 
\eq{
S\defeq\{X^pY^qZ^r,[X^1,X^pY^qZ^r], [Z^2,X^pY^1Z^r]:p\ge 0, ~\text{ $q$ and $r$ are odd}\} =  K\cup L}
where
\eq{
K&\defeq \{X^pY^qZ^r:\text{$q$ and $r$ are odd}\},\\
L&\defeq\{[X^1,X^pY^qZ^r], [Z^2,X^pY^1Z^r]:\text{$q$ and $r$ are odd}\},
}
and let \[\.s = \spn_\mb{R} S.\] 
We will show that $S$ is a basis for $[\.g,\.g]$, and calculate the dimension of $\.s$. The dimension of $\.c$ follows from the fact that $\dim(\.c) = \dim(\.g) - \dim([\.g,\.g])$, and the dimension of $\.g$ proved in Theorems \ref{thm:Kn-basis-even} ($n$ even) and \ref{thm:Kn-oddn-basis} ($n$ odd).

\paragraph{(1) Linear independence of vectors in $S$.} 
Since different Pauli strings are orthogonal, the vectors in $K$ are all linearly independent. Vectors in $K$ are linearly independent from those in 
$L$ (which consist of Pauli strings with even parities for $Y$ and $Z$). By Fact \ref{fact:linear-independent}, vectors in $L$ are also linearly independent.

\paragraph{(2) Dimension.} From the linear independence of the vectors in $S$, it follows that $\dim(\.s) = \abs{K} + \abs{\{[X^1, X^p Y^q Z^r] : p\ge0 \text{ and }q,r \text{ odd}\}} + \abs{\{[Z^2, X^p Y^1 Z^r] : p\ge0 \text{ and }r \text{ odd}\}\}}$. Then, noting that
\eq{
\abs{K}=\abs{\{[X^1, X^p Y^q Z^r] : p\ge0 \text{ and }q,r \text{ odd}\}} 
&= \abs{\{(p,q,r) : p,q,r\ge 0, q,r \text{ odd }, p+q+r \le n\} } \\
\abs{\{[Z^2, X^p Y^1 Z^r ] : p\ge 0,\text{ and }q,r \text{ odd} \}} &= \abs{\{(p,r) : p,r\ge 0, r \text{ odd }, p+r \le n\}}
}
gives
\begin{multline}
\dim(\.s)=\\ \begin{cases}
2\cdot \sum_{t=1}^{n/2} t(n-2t+1)+\sum_{t=1}^{n/2}(n-2t+1)=\frac{1}{12}\lp n^3+6n^2+2n\rp &\quad (n \text{ even})\\
2\cdot \sum_{t=1}^{(n-1)/2} t(n-2t+1)+\sum_{t=1}^{(n-1)/2}(n-2t+1)= \frac{1}{12}\lp n^3+6n^2-n-6\rp &\quad (n \text{ odd})
\end{cases}
\end{multline}

\paragraph{(3) $\.s\subseteq[\.g,\.g]$.} By Fact~\ref{fact:anyoddodd-oddeveneven}, $K\subseteq [\.g,\.g]$, and vectors in $L$ are also in $[\.g,\.g]$ since they are commutators of $X^1$ and $Z^2$ with elements from $K$.

\paragraph{(4) $[\.g,\.g]\subseteq \.s$.}  Note that if $R$ is a basis of $\.g$, then $[\.g,\.g]=\spn_\mbR \{ad_{X^1}(R) \cup ad_{Z^2}(R)\}$. It will therefore suffice to show that $\spn_\mbR \{ad_{X^1}(R) \cup ad_{Z^2}(R)\}\subseteq \.s$. 

There are two cases to consider. The term `by parity' used below has the same meaning as in the main text and in the Supplementary Information section \ref{sec:Kn-oddn}.

    \paragraph{Case 1 ($n$ even):}  $R\defeq\bigcup_{k=1}^5 J_k \cup L_1\cup L_2$ is a basis for $\.g$ (Theorem~\ref{thm:Kn-basis-even}). It can be verified that
    \begin{align*}
        & ad_{X^1}(J_1\cup J_2 \cup J_3 \cup \{Y^1Z^1\})=\spn_\mbR \{[X^1,X^pY^qZ^r]: \text{$q$, $r$ are odd}\}\subseteq \spn_\mbR L, & \text{(by definition)}\\
        & ad_{X^1}(J_4\cup J_5\cup L_1\cup L_2\backslash\{Y^1Z^1\})\subseteq \spn_\mbR K, & \text{(by parity)}\\ 
        & ad_{Z^2}(\{X^pY^1Z^r:\text{$r$ is odd}\})=\spn_\mbR \{[Z^2,X^pY^1Z^r]: \text{$r$ is odd}\}\subseteq \spn_\mbR L, & \text{(by definition)}\\
        & ad_{Z^2}(J_1\cup J_2 \cup J_3 \cup \{Y^1Z^1\}\backslash\{X^pY^1Z^r:\text{$r$ is odd}\})\subseteq \spn_\mbR L, & \text{(by Lemma~\ref{lem:adzz-pqr-commsubalgebra})}\\
        & ad_{Z^2}(J_4\cup J_5\cup L_1\cup L_2\backslash\{Y^1Z^1\})\subseteq \spn_\mbR K,  & \text{(by parity)}
    \end{align*}
    and thus, $[\.g,\.g]=\spn_\mbR \{ad_{X^1}(R) \cup ad_{Z^2}(R)\}\subseteq \.s$.
    
    \paragraph{Case 2 ($n$ odd):}  $R\defeq Q\cup V_1\cup V_2$ is a basis of $\.g$ (Theorem~\ref{thm:Kn-oddn-basis}). Again, it can be verified that
    \begin{align*}
        & ad_{X^1}(\{X^pY^qZ^r:\text{$q$, $r$ are odd}\})=\spn_\mbR \{[X^1,X^pY^qZ^r]: \text{$q$, $r$ are odd}\}\subseteq \spn_{\mbR}L,  & \text{(by definition)}\\
        & ad_{X^1}(Q\cup V_1\cup V_2\backslash\{X^pY^qZ^r:\text{$q$, $r$ are odd}\})\subseteq \spn_\mbR K, & \text{(by parity)}\\
        & ad_{Z^2}(\{X^pY^1Z^r:\text{$r$ is odd}\})=\spn_{\mbR}\{[Z^2,X^pY^1Z^r]:\text{$r$ is odd}\}\subseteq \spn_{\mbR} L,  & \text{(by definition)}\\
        & ad_{Z^2}(\{X^pY^qZ^r:\text{$q$, $r$ are odd}\}\backslash\{X^pY^1Z^r:\text{$r$ is odd}\})\subseteq \spn_\mbR L,& \text{(by Lemma~\ref{lem:adzz-pqr-commsubalgebra})}\\
        & ad_{Z^2}(Q\cup V_1\cup V_2\backslash\{X^pY^qZ^r:\text{$q$, $r$ are odd}\})\subseteq \spn_\mbR K, & \text{(by parity)}
    \end{align*}
and thus, $[\.g,\.g]=\spn_\mbR \{ad_{X^1}(R) \cup ad_{Z^2}(R)\}\subseteq \.s$.
\end{proof}

\section{$K_4$ decomposition}

Here we prove a decomposition of $[\.g_{K_n},\.g_{K_n}]$ into simple components, for the special case that $n=4$.
\begin{lem}\label{lem:k4-semisimple}$[\.g_{K_4},\.g_{K_4}]\cong\.{su}(2)\oplus \.{su}(2)\oplus \.{su}(3)$.
\end{lem}
\begin{proof}
 Define the following $v_1,v_2,v_3,v_4\in[\.g_{K_4},\.g_{K_4}]_\mbC$:
\begin{align*}
    v_1\defeq  &\frac{1}{144}([Z^2,[Z^2,[Z^2,[X^1,Z^2]]]]+40[Z^2,[X^1,Z^2]])=X^1,\\
    v_2 \defeq & 16[Z^2,[X^1,Z^2]]+[X^1,[Z^2,[X^1,[X^1,Z^2]]]]
        = 192X^1+64(X^1Y^2+X^1Z^2),\\
    v_3 \defeq &-448[Z^2,[X^1,Z^2]]+[Z^2,[Z^2,[X^1,[Z^2,[X^1,[X^1,Z^2]]]]]]\\
        = & -3840X^1+1536X^3-1792(X^1Y^2+X^1Z^2),\\ 
    v_4 \defeq &[X^1,[X^1,[Z^2,[Z^2,[X^1,[X^1,Z^2]]]]]]\\ &-3[X^1,[Z^2,[X^1,[Z^2,[X^1,[X^1,Z^2]]]]]]-32[Z^2,[Z^2,[X^1,[X^1,Z^2]],\\
       = & -4096X^2+2048(Y^2+Z^2)-2048(X^2Y^2+X^2Z^2)+3072(Y^4+Z^4)+1024Y^2Z^2.
\end{align*}
    It can be easily verified that $v_1,v_2,v_3,v_4$ are linearly independent. By Eqs. \eqref{eq:Kn-adA} and \eqref{eq:Kn-adB}, $[v_s,v_t]=0$ for any $1\le s < t\le 4$. Therefore $\.h'\defeq\spn_{\mbC}\{v_1,v_2,v_3,v_4\}$ is commutative. Let $\.h$ denote a Cartan subalgebra of $[\.g_{K_4},\.g_{K_4}]_\mbC$. By the same argument as used in the proof of Theorem \ref{thm:cycle-Cartan}, $\dim(\.h)\ge \dim(\.h')=4$. We decompose $[\.g_{K_4},\.g_{K_4}]_{\mbC}$ as the direct sum of simple complex Lie algebras $[\.g_{K_4},\.g_{K_4}]_{\mbC}=\.g_1\oplus \.g_2\oplus \cdots \oplus \.g_{k}$. Since $\dim([\.g_{K_4},\.g_{K_4}]_\mbC)=\dim([\.g_{K_4},\.g_{K_4}])=14$ (Theorem \ref{thm:kn-semisimple-center}), $\.g_1,\.g_2,\ldots,\.g_k$ are isomorphic to one of the following simple complex Lie algebras $A_1$, $A_2$, $B_2$ and $G_2$ whose dimensions are no more than $14$. By Theorem \ref{thm:classification}$, \dim(A_1)=3$, $\dim(A_2)=8$, $\dim(B_2)=10$ and $\dim(G_2)=14$, with the subscripts denoting the dimensions of their Cartan subalgebras. 
    Assume that there are $n_1$, $n_2$, $n_3$, $n_4$ simple Lie algebras which are isomorphic to $A_1$, $A_2$, $B_2$ and $G_2$ respectively in the direct sum of $[\.g_{K_4},\.g_{K_4}]_\mbC$. We have
    \begin{align}
        & \dim(\.h)=n_1+2n_2+2n_3+2n_4\ge 4 \label{eq:cartanK4}\\
        & \dim([\.g_{K_4},\.g_{K_4}]_{\mbC})=3n_1+8n_2+10n_3+14n_4=14.\label{eq:dimensionK4}
    \end{align}
    If $n_4\ge 1$, then $(n_1,n_2,n_3,n_4)=(0,0,0,1)$ which contradicts Eq. \eqref{eq:cartanK4}. This implies $n_4=0$. If $n_3\ge 1$, Eq. \eqref{eq:dimensionK4} has no solution, which implies $n_3=0$. If $n_2\neq 1$, Eq. \eqref{eq:dimensionK4} has no solution, which implies $n_2=1$. Then the above linear system has only one solution $(n_1,n_2,n_3,n_4)=(2,1,0,0)$. Therefore,  $[\.g_{K_4},\.g_{K_4}]_\mbC\cong A_1\oplus A_1 \oplus A_2=\.{sl}(2,\mbC)\oplus \.{sl}(2,\mbC)\oplus \.{sl}(3,\mbC)$ and, by the same argument used in the proof of Theorem \ref{thm:Cn-decomp}, $[\.g_{K_4},\.g_{K_4}]\cong \.{su}(2)\oplus \.{su}(2)\oplus \.{su}(3)$.
\end{proof}

\end{document}